\documentclass[sigconf,nonacm]{acmart}
\renewcommand\footnotetextcopyrightpermission[1]{}
\usepackage{algorithm}
\usepackage{multirow}
\usepackage{multicol}
\usepackage{float}
\usepackage{enumitem}
\usepackage{tcolorbox}

\usepackage{booktabs}
\usepackage{tabularx}
\usepackage{subcaption}
\usepackage{algpseudocode}
\usepackage{tikz}
\usetikzlibrary{positioning, shapes.geometric, arrows.meta, calc}
\usepackage{placeins}  
\usepackage{afterpage} 

\definecolor{HCSCstageNavy}{HTML}{2C4870}
\definecolor{HCSCstageNavyFill}{HTML}{E5ECF6}
\definecolor{HCSCblueOut}{HTML}{4A6FA5}
\definecolor{HCSCblueFill}{HTML}{E8F0FA}
\definecolor{HCSCgreenOut}{HTML}{3F7E3F}
\definecolor{HCSCgreenFill}{HTML}{E5F2E5}
\definecolor{HCSCorangeOut}{HTML}{C97A1C}
\definecolor{HCSCorangeFill}{HTML}{FDF0DC}
\definecolor{HCSCredOut}{HTML}{B0383E}
\definecolor{HCSCredFill}{HTML}{FBE5E5}
\definecolor{HCSCgrayOut}{HTML}{555555}
\definecolor{HCSCgrayFill}{HTML}{ECECEC}

\usepackage{float}
\usepackage{amsmath}

\newtheorem{assumption}{Assumption}[section]

\theoremstyle{plain}
\newtheorem{theorem}{Theorem}[section]
\newtheorem{lemma}{Lemma}[section]
\newtheorem{corollary}{Corollary}[section]
\newtheorem{proposition}{Proposition}[section]

\theoremstyle{definition}
\newtheorem{definition}{Definition}[section]
\newtheorem{example}{Example}[section]

\theoremstyle{remark}
\newtheorem{remark}{Remark}[section]

\newcommand{\cE}{\mathcal{E}}             
\newcommand{\PH}[1]{{\color{red}?.??}}
\newcommand{\core}{\widehat{C}}           
\newcommand{\agg}{\mathcal{A}}            
\newcommand{\quant}[1]{Q_{\eta}(#1)}      
\newcommand{\cert}[1]{\mathrm{Cert}(#1)}  
\newcommand{\hstar}{h^{\star}}            
\newcommand{\ystar}{\mathbf{y}^{\star}}   
\newcommand{\thalpha}{\theta_{\alpha}}    
\newcommand{\eps}{\bar{\varepsilon}_{\alpha}} 
\newcommand{\Honest}{\mathcal{H}}         
\newcommand{\Byz}{\mathcal{B}}            
\newcommand{\Nset}{\mathcal{N}}           

\providecommand{\CR}{\ensuremath{\mathrm{CR}}}          
\providecommand{\SEM}{\ensuremath{s_{\textsc{sem}}}}    
\providecommand{\INVgold}{\ensuremath{\mathrm{inv}_{\mathrm{gold}}}}
\providecommand{\INVhmaj}{\ensuremath{\mathrm{inv}_{\mathrm{hmaj}}}}
\providecommand{\byzcore}{\ensuremath{\beta_{\mathrm{core}}}}
\providecommand{\VCC}{\ensuremath{\mathrm{VCC}}}
\providecommand{\marginmin}{\ensuremath{m_{\min}}}

\providecommand{\Commit}{\textsc{commit}}
\providecommand{\Abort}{\textsc{abort}}

\AtBeginDocument{%
  }

\begin{document}

\title{Certifiable Semantic Agreement Among LLM Agents:
What the Admissibility Instrument Decides}

\author{Haoran Xu}
\email{2614067X@student.gla.ac.uk}
\orcid{0009-0000-6659-2043}
\affiliation{%
  \institution{University of Glasgow}
  \city{Glasgow}
  \country{United Kingdom}
}

\author{Lei Zhang}
\email{lei.zhang@glasgow.ac.uk}
\orcid{0000-0002-4767-3849}
\affiliation{%
  \institution{University of Glasgow}
  \city{Glasgow}
  \country{United Kingdom}
}

\author{Iadh Ounis}
\email{Iadh.Ounis@glasgow.ac.uk}
\orcid{0000-0003-4701-3223}
\affiliation{
    \institution{University of Glasgow}
    \city{Glasgow}
    \country{United Kingdom}
}

\author{Xianbin Wang}
\email{xianbin.wang@uwo.ca}
\orcid{0000-0003-4890-0748}
\affiliation{
    \institution{University of Western Ontario}
    \city{London}
    \country{Canada}
}

\renewcommand{\shortauthors}{H. Xu et al.}


\begin{abstract}
Can a committee of large-language-model (LLM) agents reach agreement that is
\emph{certifiable at the level of meaning}, rather than only at the level of a
label? We build a protocol to find out. \textsc{H-CSC} produces one of three
typed outcomes per round under a common $2f{+}1$ distinct-signer certificate: a
\emph{semantic commit} whose digest binds a quantised aggregate over a
within-verdict admissible core, a \emph{verdict commit} that binds only the
label, or a typed \emph{abort}. Using it as an instrument, we report what
certifiable semantic agreement costs and what it buys.

Our answer is conditional, and the condition is not the protocol. We prove a
\emph{containment lemma}: whenever the semantic core is large enough to make the
committed verdict deterministically valid, the verdict margin already exceeds
$f$, so at matched deterministic guarantees the semantic predicate is strictly
stronger than the verdict predicate and \emph{no coverage separation from
certificate-wrapped majority is possible}. Measurement confirms it --- the two
rules commit the identical task set. The value therefore cannot lie in the
verdict layer, and we look where the lemma points.

There we find that the \emph{admissibility instrument}, not the protocol, is the
design variable that matters. Against adversaries that preserve the verdict and
corrupt only the reasoning, a $442$\,MB fine-tuned sentence encoder reaches
AUROC $0.621$--$0.744$ at $0$--$8\%$ true-positive rate (at $5\%$ honest false
positives), while a training-free lexical-conformity predicate reaches
$0.865$--$0.982$ at $38$--$80\%$ ($50$ tasks, $200$ attacks, $400$ honest); the lexical predicate is also exactly
deterministic, discharging an encoder-determinism assumption the digest proofs
depend on. The mechanism generalises beyond our protocol: honest agents disperse
further than attacks displace ($90$th-percentile honest angular distance
$0.992$\,rad against an admissibility radius of $0.65$), so on free-text
rationales no embedding-level filter is viable --- and this is a property of the
\emph{output schema}, which when tightened collapses honest dispersion
twenty-fold and reverses the inequality.

Finally we show the committed digest is not inert. Under non-disclosure --- an
auditor holding the commitment record but not the agent texts, the regime that
data-minimisation and retention limits actually produce --- a similarity-preserving
committed sketch separates faithful derived statements from substituted ones at
AUROC $0.741$ ($95\%$ CI $[0.69, 0.79]$, $n{=}350$), where a cryptographic hash
and a verdict-only digest score exactly $0.500$. The capability is concentrated
on reasoning substitution ($0.770$) and absent on scope shift ($0.617$), a
boundary that coincides exactly with the semantics neither representation
encodes. We also report a tie-break capture vulnerability we found in our own
protocol and its four-line fix. We claim no safety or coverage advantage over
verdict-only certification; the containment lemma explains why there is none to
claim.
\end{abstract}
\maketitle

\section{Introduction}
\label{sec:introduction}
\label{sec:intro}

Multi-agent systems built on LLMs increasingly produce \emph{natural-language
artefacts}: a verdict on a scientific claim, a triage recommendation, a plan.
Classical Byzantine agreement gives deterministic finality for discrete values
--- a quorum of $2f{+}1$ distinct signatures on one byte-string --- and recent
multi-agent LLM work gives the opposite extreme, a soft aggregation with no
certificate and no abort behaviour. Neither fits a setting in which honest
agents given the same evidence emit textually distinct proposals that mean the
same thing, and a Byzantine fraction emits fluent proposals that do not.

The obvious construction, wrapping a majority verdict in a $2f{+}1$ certificate,
records \emph{what} was decided and nothing about \emph{on what grounds}. Two
rounds that commit the same label on mutually incompatible reasoning produce
byte-identical records. \emph{Certified semantic commitment} is the natural
response: additionally bind a deterministic function of the committed content,
so the record commits to more than a token. This paper asks whether that is
worth doing, and answers by building the protocol and measuring.

\paragraph{A concrete failure that motivates the question.}
Our own protocol admits the following attack, which we found while auditing it.
On a claim where the honest agents split $5$--$3$ for \textsc{insufficient}, two
Byzantine agents both vote \textsc{support}. The delivered view is a $5$--$5$
tie; the protocol's published tie-break order resolves it to \textsc{support};
the semantic core within that group is three honest agents and both Byzantine
agents, which meets the $2f{+}1$ quorum and the angular admissibility radius. The
protocol emits a fully valid, fully certified \emph{semantic commit} for the
honest \emph{minority} verdict, with $40\%$ of its core Byzantine. The root cause
is that the semantic branch never consulted the verdict margin. Two agents
converted a losing position into a certified one by reading a tie-break rule off
the specification. The fix is four lines
(\S\ref{sec:margin-gate}); the point is that a protocol can be provably correct
in every stated theorem and still be captured, which is why the rest of this
paper is about characterising boundaries rather than asserting guarantees.

\paragraph{Contributions.}
\begin{enumerate}[leftmargin=1.6em,itemsep=0.25em]
\item \textbf{A typed-finality protocol and a complete guarantee
characterisation} (\S\ref{sec:protocol}--\S\ref{sec:correctness}): a four-level
lattice separating digest agreement, certificate validity, semantic proximity
and verdict correctness, with an explicit statement of which levels are
deterministic and for which commit type. Its centrepiece is the
\textbf{containment lemma} and its corollary: at matched deterministic
guarantees the semantic admissibility predicate is strictly stronger than the
verdict predicate, so coverage separation from certificate-wrapped majority is
impossible \emph{by proof}. Empirically the two rules commit the identical task
set ($0.72$ static / $0.64$ rushing, zero discordant tasks).

\item \textbf{A vulnerability in our own protocol, with its fix and the fix's
price} (\S\ref{sec:margin-gate}): the tie-break capture above; gating both
branches at margin $\ge 1$ drives the honest-reference-invalid rate to
$0.00/0.00$ at a cost of $0.02$ static and $0.00$ rushing commit rate.

\item \textbf{A systematic study of the admissibility instrument}, the first we are aware of,
(\S\ref{sec:eval-r2}): learned embedding geometry versus lexical conformity
versus evidence-id overlap versus an LLM judge, across two verdict-preserving
attack families and two evasion levels. The learned encoder is dominated in
every cell by a training-free lexical predicate that is additionally exactly
deterministic, which discharges the encoder-determinism assumption on which
byte-identical digests depend.

\item \textbf{The mechanism, and its dependence on the output schema}
(\S\ref{sec:eval-r3}, \S\ref{sec:eval-r6}): honest dispersion exceeds attack
displacement on free-text rationales, so no admissibility filter of any kind can
work there; tightening the output schema collapses honest dispersion twenty-fold
while leaving attack displacement unchanged, and reverses the inequality. This
is a statement about agent populations and output formats, not about our
encoder.

\item \textbf{A downstream use in which the committed sketch beats a
cryptographic hash} (\S\ref{sec:audit}): under non-disclosure, an auditor
holding only the commitment record separates faithful derived statements from
substituted ones at AUROC $0.741$ $[0.69,0.79]$, $n{=}350$, against exactly
$0.500$ for a plain hash and for a verdict-only digest.
\end{enumerate}

\paragraph{What we do not claim.}
We do not claim that \textsc{H-CSC} is safer than, or covers more rounds than,
verdict-only certification; contribution~1 proves that it cannot be, at matched
guarantees. We do not claim that semantic admissibility defends against
evidential poisoning; we measure that it does not
(\S\ref{sec:eval-r1}). And we do not claim that the downstream
capability of contribution~5, at a $21\%$ true-positive rate, is on its own a
sufficient reason to deploy the machinery. What we claim is a map: where the
value cannot be, where it can, how large it is, and what determines it.

\paragraph{Relation to a prior version.}
An earlier version of this work claimed parity with the strongest verdict-only
baseline plus an additional semantic digest, and presented the digest's
existence as the contribution. That framing was not supportable: the parity is
provable rather than empirical (contribution~1), the attack suite used to
establish it could not exercise the semantic layer at all
(\S\ref{sec:eval-attack-audit}), and the digest's value was asserted rather than
measured. Appendix~\ref{app:errata} lists the specific corrections.

\section{Related Work}
\label{sec:related_work}
\label{sec:related} 

Four literatures bear on this paper: multi-agent LLM
collaboration and answer aggregation; Byzantine fault tolerance
applied to LLM-agent systems; Byzantine-robust aggregation in
distributed learning, together with the attacks that defeat it;
and approximate agreement over non-identical values. We review
each against a single question: \emph{does it already supply, or
already rule out, a certified single-round commitment over
natural-language proposals whose admissibility is decided by a
learned semantic predicate?} The answer that emerges structures
this revision. The typed commitment object appears to be new.
The \emph{instrument} we used to decide admissibility --- a
learned embedding plus an angular ball --- is not new, and the
way it fails in \S\ref{sec:eval-r2} is a close analogue of a
failure the distributed-learning community documented years ago.

\paragraph{Multi-agent LLM collaboration and aggregation.}
A large body of work studies multi-agent LLM collaboration ---
debate, deliberation, committee-style
workflows --- as a means of improving task performance and
reliability~\cite{park2023generative, li2024improving,
du2023improving, liang2024encouraging, chen2024reconcile}.
These systems typically consolidate stochastic outputs by voting
or pooling; recent evaluation frameworks adopt majority vote over
many LLM judgments explicitly, which makes both the prevalence
of aggregation and the operational cost of LLM non-determinism
visible~\cite{zhao2025lmc}. A parallel line derives principled
aggregation rules that optimise expected accuracy under
statistical models of judge reliability~\cite{ai2025beyond}.
Closest to this thread on the adversarial side is
SAC~\cite{lee2026sac}, a decentralised filter-and-refine scheme
for multi-agent LLM systems over $(f{+}1)$-robust communication
graphs, which restores task accuracy on mathematics and
commonsense benchmarks in the presence of Byzantine agents.
None of this work is a consensus protocol in the fault-tolerance
sense: a trusted coordinator or benign participation is normally
assumed, and no certified single outcome is produced. We also do
not compete with SAC on its axis. H-CSC is not an
accuracy-improving aggregator and, as
\S\ref{sec:eval-setup} reports, does not improve accuracy over a
verdict-only baseline; the question it addresses is what kind of
finality a round supports and what artefact the round leaves
behind.

\paragraph{Byzantine fault tolerance for LLM-agent systems.}
The most directly comparable work is
CP-WBFT~\cite{zheng2026cpwbft}, a confidence-probe-weighted BFT
consensus procedure for multi-agent LLM systems, evaluated over
several communication topologies and reported to remain reliable
at high fault rates. CP-WBFT was peer-reviewed and published
before the first version of this paper appeared; its omission
from that version was an error, and we correct it here. We claim
no priority over CP-WBFT for the observation that Byzantine
fault tolerance is a productive frame for LLM-agent reliability.
The two differ in what they produce rather than in whether they
use BFT: CP-WBFT weights agent votes by a probe of agent
confidence and returns a single agreed answer, whereas H-CSC
returns a typed, quorum-certified commitment object --- a
$2f{+}1$ distinct-signer certificate bound to either an
embedding-derived quantised aggregate or a verdict-level payload
--- that a third party can re-verify offline without re-running
the round.
Closest in framing is Semantic Quorum Assurance
(SQA)~\cite{he2026sqa}, which forms risk-adaptive quorums over
diverse validator agents, binds proposals to cryptographic
evidence chains, and gates execution behind approval. SQA is the
work that most directly answers the objection that motivated
this revision: it names a concrete downstream consumer --- the
execution gate --- and measures harm prevented rather than
structure emitted. We differ in operating inside a formal
$n \ge 3f{+}1$ model with a distinct-signer certificate and
typed finality, and in certifying an embedding-derived aggregate
rather than a validator-vote tally. SQA's evaluation design,
however, is the standard this revision holds itself to, and
\S\ref{sec:eval-r2} adopts it.
Blockchain-driven weighted BFT voting for multi-LLM
networks~\cite{luo2025wbft} and consensus-inspired robustness
for communicative multi-agent systems under communication
attacks~\cite{mao2024ibgp} share the mechanism but not the
commitment object. Anand and Pappas~\cite{anand2026resilient}
report that LLM agents fail to converge in settings where
classical resilient-consensus theory guarantees convergence,
which is useful evidence that the naive-aggregation baseline
this whole line displaces is genuinely broken. Classical BFT
itself~\cite{lamport1982byzantine, castro1999, yin2019,
dwork1988partial, bracha1987asynchronous} is the substrate we
borrow from; it agrees on byte-identical values, which LLM
proposals do not supply.

\paragraph{Provenance and accountability for agent systems.}
As LLM collaboration becomes agentic and tool-integrated,
adversarial participation is a first-order concern: prompt
injection propagates \emph{across} interconnected agents,
enabling persistent compromise of a
workflow~\cite{lee2024promptinfection}, and systematisations
catalogue cross-agent injection and protocol-layer
vulnerabilities in deployed agent
ecosystems~\cite{ferrag2025protocol, ferrag2025threats}. The
defensive response has largely been trace-level: a recent survey
of evidence tracing and execution provenance for LLM
agents~\cite{agentprov2026survey} organises the field into trace
logging, attribution, and audit-graph construction. That
taxonomy contains no category for Byzantine-fault-tolerant
agreement or cryptographic commitment over agent outputs. The
present paper is best read as an attempt to occupy that empty
cell, and \S\ref{sec:eval-r2} is an attempt to establish
which instrument should decide what is admissible once one does.

\paragraph{Certified artefacts whose value is not accuracy.}
Several security primitives are valuable while improving no
accuracy metric, and they share an evaluation shape worth
naming. Randomised smoothing~\cite{cohen2019smoothing} reports a
certified-accuracy-versus-radius curve --- how much guarantee at
how much coverage --- rather than an accuracy win, and clean
accuracy in fact decreases. Supply-chain and accountability
systems such as in-toto~\cite{torresarias2019intoto} and
PeerReview~\cite{haeberlen2007peerreview} are evaluated by
replaying real compromises and by a
completeness/soundness pair, not by a task score. We adopt this
shape rather than an accuracy comparison: the
harm-versus-coverage frontier of \S\ref{sec:eval-r2} is a
guarantee-versus-coverage curve, and its most important property
is that for the learned predicate the curve is flat.

\paragraph{Byzantine-robust aggregation, and the attacks that
defeat it.}
Robust aggregation replaces the mean with an estimator that
tolerates a bounded fraction of arbitrary inputs: distance-based
selection (Krum, Multi-Krum)~\cite{blanchard2017},
coordinate-wise and trimmed
median~\cite{yin2018median, chen2017}, Bulyan's recursive
composition~\cite{mhamdi2018}, geometric-median smoothing
(RFA)~\cite{pillutla2022}, and the high-dimensional robust
estimators behind them~\cite{diakonikolas2019}. H-CSC's
aggregation stage is a member of this family: it is a geometric
median over a selected core.
The part of this literature that v1 failed to engage with, and
that is now central to our positioning, is the attack side. ``A
Little Is Enough''~\cite{baruch2019little} shows that a small
number of colluding participants can defeat distance-based
defences by perturbing within the empirical variance of the
honest population, and ``Fall of Empires''~\cite{xie2020fall}
shows the same by manipulating inner products so that malicious
vectors remain geometrically ordinary. The common lesson is that
a geometric admissibility filter is only as discriminative as
the gap between honest dispersion and adversarial displacement.
\emph{Our negative results are the natural-language analogue of
these two attacks.} In \S\ref{sec:eval-r2} we construct
verdict-preserving evidence-poisoning attacks (n${=}40$ attacks
over $10$ tasks) that keep the adversary's verdict identical to
the honest plurality and cite only real evidence identifiers. The
honest population's own angular dispersion around its geometric
median has median $0.528$~rad and $90$th percentile
$0.992$~rad, while the attacks sit at median $0.58$--$0.85$~rad;
the operating threshold $\thalpha = 0.65$~rad therefore barely
contains the honest population itself, and $90\%$ of the
evasive-tier attacks land inside the ball. This is
\cite{baruch2019little}'s inequality restated in an embedding
space: the perturbation the adversary needs is smaller than the
noise the honest agents already produce. Stating this connection
is also, we think, the right answer to the criticism that a
protocol built from standard primitives is ``structural'': the
question is not whether the primitives are standard but whether
the composed predicate discriminates, and the honest way to
settle that is to measure it.

\paragraph{Approximate agreement, and why H-CSC is not an
instance of it.}
Approximate agreement relaxes exact equality: honest nodes must
output values within $\varepsilon$ of each other and inside the
convex hull of honest inputs. The scalar
formulation~\cite{dolev1986reaching}, the multidimensional
extension of Mendes and Herlihy~\cite{mendes2013multidimensional},
iterative approximate consensus over arbitrary directed
graphs~\cite{vaidya2013}, and recent work coupling robust
aggregation with an agreement
subroutine~\cite{milentijevic2025approxagree} are the canonical
references. H-CSC is not an instance of this family, for three
reasons, and it is worth being precise about them because the
resemblance is superficially strong.
First, the output requirement differs in kind. Approximate
agreement permits honest outputs to differ by up to
$\varepsilon$ and reaches that state by iterated convergence
rounds. H-CSC requires a \emph{byte-identical} digest
$\hstar$ across honest signers in one round, because the digest
is the object that gets certified and re-verified; the
quantiser $\quant{\cdot}$ exists precisely to convert
$\varepsilon$-closeness into byte equality
(Lemma~\ref{lem:digest-consistency}).
Second, the validity notion differs. Approximate-agreement
validity says the output lies in the convex hull of honest
inputs. On our data that guarantee is close to vacuous: the
honest hull is wide enough ($90$th-percentile dispersion
$0.992$~rad against $\thalpha = 0.65$~rad) that
verdict-preserving attacks lie inside it, so hull membership
certifies nothing an adversary cannot arrange.
Third, H-CSC's decision is typed and discrete. The protocol does
not converge toward a point; it evaluates a predicate and emits
one of three outcomes, one of which is an abort, and the
correctness question is which typed outcome was emitted rather
than how far apart two honest outputs are. The distinction
matters for this revision: the guarantee we prove in
\S\ref{sec:correctness} is a containment relation between two
admissibility conditions (Lemma~\ref{lem:containment}), not a
convergence rate.

\paragraph{Positioning of this version.}
Relative to the works above, the contribution we defend is
narrow and mostly negative. The typed commitment object and its
guarantee lattice are, as far as we can determine, not supplied
by CP-WBFT~\cite{zheng2026cpwbft}, SQA~\cite{he2026sqa},
SAC~\cite{lee2026sac}, or the agent-provenance
literature~\cite{agentprov2026survey}. The learned semantic
encoder that v1 used to decide admissibility is, on our
measurements, the wrong instrument for the job, and a
deterministic lexical-conformity predicate dominates it on the
harder attack family (\S\ref{sec:eval-r2}) while also
discharging the encoder-determinism assumption that
Assumption~\ref{assn:deterministic-cE} and
Lemma~\ref{lem:digest-consistency} rely on. Earlier formulations
of this work explored topology-based filtering as the central
mechanism; we retain that exploration in
Appendix~\ref{appendix:topology} as a design-space ablation
only, because on every measured benchmark it yields
commit/abort decisions either identical to the H-CSC main path
or strictly more conservative without improving
honest-reference validity.


\section{System Model and Problem Formulation}
\label{sec:problem}
\label{sec:system_model}
We consider a distributed system consisting of a set of $n$ nodes $\mathcal{N}=\{1,\dots,n\}$, where a subset $\mathcal{B}\subset\mathcal{N}$ of size at most $f$ may be Byzantine. We assume $n \geq 3f+1$, which is the tight bound for deterministic Byzantine agreement in partially synchronous systems~\cite{lamport1982byzantine,dwork1988partial}. This bound is both necessary and sufficient: with fewer than $3f+1$ nodes, an adversary controlling $f$ Byzantine nodes can partition honest nodes into two groups that commit conflicting values, violating safety~\cite{lamport1982byzantine}. The $3f+1$ threshold has been foundational since Lamport et al.'s seminal work on the Byzantine Generals Problem~\cite{lamport1982byzantine} and remains the standard fault tolerance limit for protocols requiring exact agreement under partial synchrony~\cite{dwork1988partial}.

\subsection{Adversary Model}
\label{subsec:adv}

We consider a Byzantine adversary that controls the set of faulty
nodes $\mathcal{B}$ with $|\mathcal{B}|\le f$. We evaluate the
protocol under two adversary modes:
\emph{static}, where Byzantine submissions are generated
independently of the honest broadcasts in the same round, and
\emph{rushing}, where the Byzantine reads the round's honest
broadcasts before submitting (subject to the network and
cryptographic constraints below). Rushing is the strict worst case
for our analysis; both modes are evaluated paired in
Section~\ref{sec:evaluation}. The semantic-poisoning capability
formalised in Definition~\ref{def:semantic-poisoning} is
independent of the timing mode.

\paragraph{Capabilities.}
Byzantine nodes may deviate arbitrarily from the protocol, including choosing arbitrary texts $x_b\in\mathcal{X}$ and coordinating across corrupted nodes.
We focus on \emph{semantic poisoning} behavior, where the adversary crafts inputs that are intended to distort the collectively committed meaning while remaining embedding-proximal to honest inputs (i.e., difficult to filter using geometry alone).

\paragraph{Cryptographic and communication constraints.}
Nodes communicate over authenticated point-to-point channels and signatures are unforgeable.
RBC is used for disseminating round inputs.
RBC does not prevent Byzantine senders from attempting equivocation, but it guarantees that honest nodes cannot end up with conflicting delivered values from the same sender (Agreement), nor can only a strict subset of honest nodes deliver a value from that sender (Totality/Uniform delivery).
Thus, any successful round proceeds from a consistent delivered view.


\subsection{Terminology used throughout}
\label{subsec:terms}

The formal definitions are collected in Appendix~\ref{app:moved}; the body uses
them in the following senses.

A \emph{semantic proposal} is a pair $z_i=(s_i,\rho_i)$: the natural-language
string an agent emitted, and the structured object extracted from it (verdict,
confidence, evidence-id list, rationale, claim text). Agents do not agree on
$s_i$; they agree on a deterministic function of $\rho_i$. The \emph{delivered
view} $V(r)$ of round $r$ collects the proposals that reached the protocol
together with their unit-norm embeddings $\mathbf{e}_i=\cE(\mathrm{CanonText}(\rho_i))$;
a \emph{successful round} delivers at least $n-f$ of them, and quorum thresholds
always reference the global bound $2f{+}1$ rather than the delivered count.

A round is \emph{admissible at radius $\thalpha$} when the candidate verdict group
contains a subset $\core$ of at least $2f{+}1$ proposals and a centre $h\in\core$
with $\angle(\mathbf{e}_h,\mathbf{e}_j)\le\thalpha$ for every $j\in\core$. The
protocol emits exactly one \emph{typed outcome} per round: a
\textsc{semantic\_commit} when such a core exists, whose digest binds the
quantised aggregate $\quant{\ystar}$ over it; a \textsc{verdict\_commit} when the
verdict-level signal is admissible but the semantic one is not, whose digest binds
only a verdict-level payload; or a typed \Abort{} carrying a reason from a closed
set. Both commit types travel under the same $2f{+}1$ distinct-signer certificate.

Two notions of validity are reported separately and must not be conflated.
\emph{Dataset-gold validity} compares the committed verdict against the
benchmark's reference label. \emph{Evidence-bounded honest-reference validity}
compares it against the honest plurality of the delivered view under the
protocol's own tie-break order --- the plurality, not a strict majority. The two
diverge whenever the gold label asserts more than the visible evidence supports,
and the honest-reference rate is the protocol-side safety metric; the protocol
itself never observes either.

\section{Hierarchical Certified Semantic Commitment Protocol}
\label{sec:protocol}
\label{sec:framework} 


\subsection{Overview}
\label{sec:protocol-overview}

\sloppy
H-CSC is a finality-control protocol for adversarial LLM-agent
collaboration.
Figure~\ref{fig:hcsc-arch} gives the architecture;
Algorithm~\ref{alg:certified-commitment-round} gives the formal
pseudocode, and Figure~\ref{fig:hcsc-protocol-flow} in
Appendix~\ref{app:moved} redraws the same round at stage
granularity. It operates on a single delivered view $V(r)$ in
round $r$, where each agent $i$ contributes a structured
proposal $z_i = (s_i, \rho_i)$ with the canonical fields
$(\textit{verdict}, \textit{confidence}, \textit{evidence ids},
\textit{rationale}, \textit{claim text})$. The protocol reads
the deterministic embedding
$\mathbf{e}_i = \cE(\mathrm{CanonText}(\rho_i))$ and emits one
of three typed outcomes
(Definition~\ref{def:csc-object}): a \textsc{semantic\_commit}
when both the verdict-level and the within-verdict
embedding-level finality signals are admissible, a
\textsc{verdict\_commit} when only the verdict-level signal is
admissible, or an \emph{abort} when neither is.

\begin{figure*}[t]
\centering
\includegraphics[width=\textwidth]{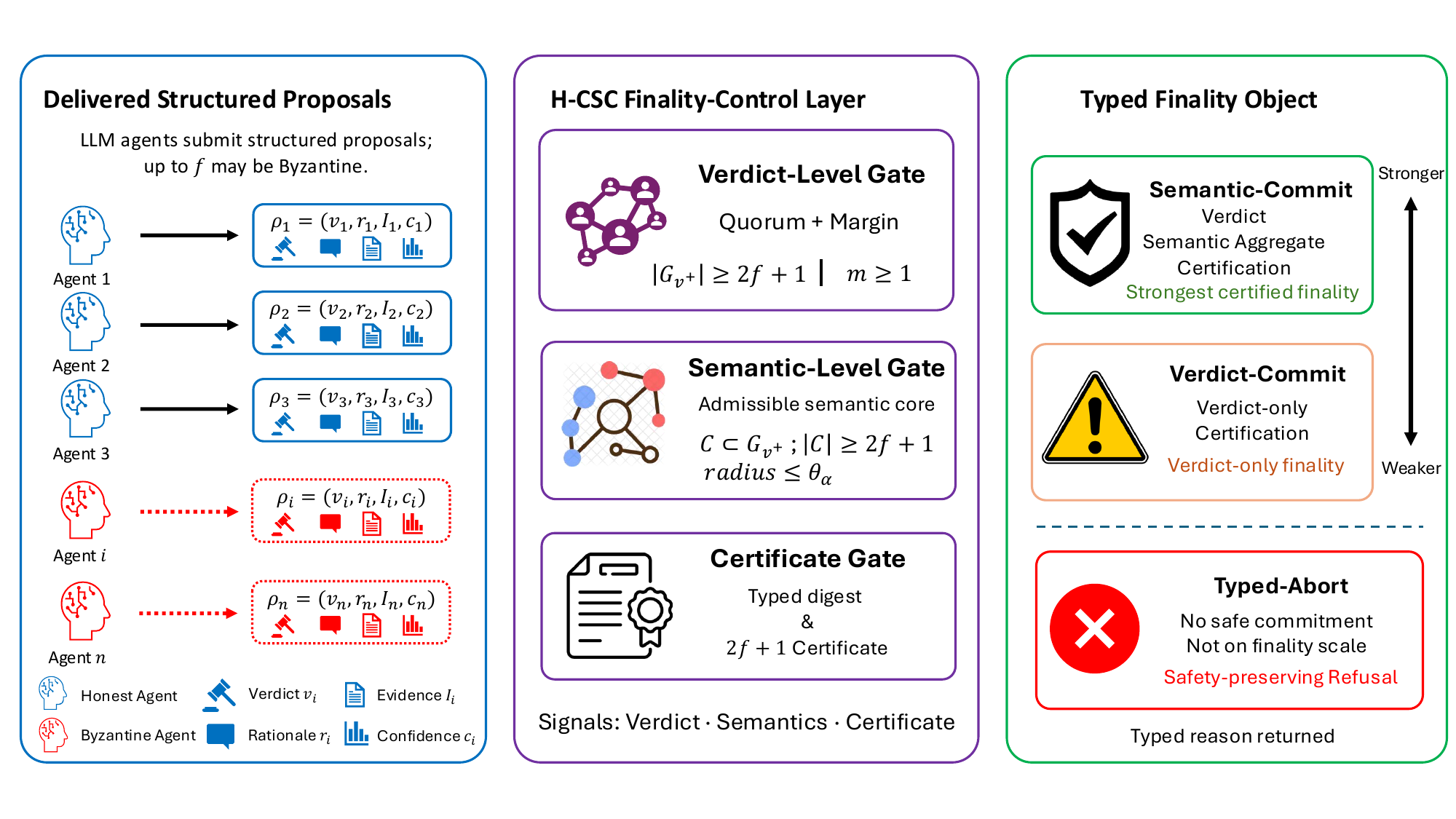}
\caption{\textbf{H-CSC finality-control architecture.} Given structured
proposals $\rho_i=(v_i,r_i,I_i,c_i)$ (verdict, rationale, evidence ids,
confidence) from $n$ LLM agents, of which up to $f$ may be Byzantine, H-CSC
applies three gates in sequence --- verdict support, semantic admissibility,
and certificate feasibility --- and then exposes the strongest typed outcome
those gates certify: \emph{semantic-commit}, \emph{verdict-commit}, or a typed
\emph{abort}. Semantic-commit provides embedding-backed finality,
verdict-commit provides verdict-only finality, and typed abort reports that no
safe commitment is available. The verdict-level gate is drawn with the margin
condition $m\ge 1$ of \S\ref{sec:margin-gate}, which this paper adds to
\emph{both} branches; \S\ref{sec:eval-r2} then asks which instrument the
semantic-level gate should be built from.}
\label{fig:hcsc-arch}
\end{figure*}

The four-stage pipeline plus the typed-decision step is
summarised in
Algorithm~\ref{alg:certified-commitment-round}. The protocol's
entire decision logic is a deterministic function of the
finality-signal vector $\Psi(r)$
(Definition~\ref{def:finality-signal}); it does not consult
honest, gold, evidence-truth, source-reliability, or
external-judge signals. The H-CSC main path uses
$\thalpha = 0.65$,
$\textit{candidate\_rule} = \textsc{largest\_verdict}$,
$\textit{fallback} = \textsc{margin\_1}$, and
$\textit{topology\_gate} = \bot$.\footnote{The reference
implementation of H-CSC is retained by the authors and is planned
to accompany a future reproducibility package. It layers the verdict grouping
and hierarchical typed-decision step on top of the
strict-semantic CSC primitives (admissibility predicate,
quantiser, digest, certificate). Promoting the H-CSC
typed-commit logic into a production-grade module remains
future work.}

\FloatBarrier  
\begin{algorithm*}[t]
\caption{Hierarchical Certified Semantic Commitment Round}
\label{alg:hcsc-round}
\label{alg:certified-commitment-round} 
\footnotesize
\setlength{\baselineskip}{0.95\baselineskip}
\begin{algorithmic}[1]
\Require Delivered structured proposals
$\{(i, z_i) : i \in I_r\}$ with $z_i = (s_i, \rho_i)$ and
$I_r \subseteq \Nset$, $|I_r| \ge n - f$ (the successful-round
delivered index set established by RBC,
\S\ref{subsec:net});
encoder $\cE$;
fault bound $f$ with $n\ge 3f{+}1$;
admissibility radius $\thalpha$;
diameter sister threshold $\eps$;
candidate rule $\textsc{largest\_verdict}$;
fallback condition $\textsc{margin\_1}$;
robust aggregator $\agg$;
quantizer $Q_\eta$ with step $\eta$;
round id $r$;
parameter digest $h_\pi$.
\Ensure A typed commit object: $\mathcal{O}^{\textsc{sem}}$,
$\mathcal{O}^{\textsc{ver}}$, or $\mathcal{O}^{\textsc{abort}}$
with a typed reason.
\Statex
\State \Comment{\emph{Stage 1: validate, canonicalise, encode.}}
\For{$i \in I_r$}
    \State $\rho_i \gets \textsc{Canonicalise}(z_i)$
    \Comment{\textit{verdict, confidence, evidence ids, rationale, claim}}
    \State $\mathbf{e}_i \gets \cE(\mathrm{CanonText}(\rho_i))$
    \State $\mathbf{e}_i \gets \mathbf{e}_i / \|\mathbf{e}_i\|_2$
\EndFor
\State $V(r) \gets \{(i,\rho_i,\mathbf{e}_i) : i \in I_r\}$
\Statex
\State \Comment{\emph{Stage 2: group by verdict and select candidate.}}
\State $\{G_v\}_v \gets$ partition $V(r)$ by $\rho_i.\textit{verdict}$
\State $v^\star \gets \arg\max_v |G_v|$
\Comment{tie-break: $\textsc{support} > \textsc{refute} > \textsc{insufficient}$}
\State $\textit{top\_count} \gets |G_{v^\star}|$
\State $\textit{verdict\_margin} \gets \textit{top\_count} - \max_{v \neq v^\star} |G_v|$
\If{$\textit{top\_count} < 2f{+}1$}
    \State \Return $\Abort(\texttt{verdict\_below\_quorum},\, \textit{top\_count},\, 2f{+}1)$
\EndIf
\Statex
\State \Comment{\emph{Stage 3: within-verdict semantic core extraction.}}
\State $\core \gets \textsc{LargestAngularComponent}\bigl(G_{v^\star},\, \thalpha\bigr)$
\State $\textit{semantic\_path\_ok} \gets (\core \neq \emptyset) \,\land\, (|\core| \ge 2f{+}1)$
\If{$\textit{semantic\_path\_ok}$}
    \State $r^\star \gets \min_{h\in\core}\,\max_{j\in\core}\,
    \angle(\mathbf{e}_h,\mathbf{e}_j)$
    \If{$r^\star > \thalpha$}
        \State $\textit{semantic\_path\_ok} \gets \textsc{false}$
        \State $\textit{semantic\_fail\_reason} \gets \texttt{admissibility\_failed}$
    \EndIf
\Else
    \State $\textit{semantic\_fail\_reason} \gets \texttt{core\_below\_quorum}$
\EndIf
\Statex
\If{$\textit{semantic\_path\_ok}$}
    \State \Comment{\emph{Stage 4a: aggregate, quantise, semantic digest.}}
    \State $\ystar \gets \agg(\{\mathbf{e}_i:i\in\core\})$
    \If{$\ystar$ is degenerate (zero norm or numerical failure)}
        \State \Return $\Abort(\texttt{aggregation\_failed})$
    \EndIf
    \State $\ystar \gets \ystar / \|\ystar\|_2$
    \State $\widetilde{\mathbf{y}} \gets Q_\eta(\ystar)$
    \State $\hstar \gets H\bigl(\texttt{"semantic\_commit"}\,\Vert\,\widetilde{\mathbf{y}}\,\Vert\,h_\pi\,\Vert\,r\,\Vert\,v^\star\bigr)$
    \State $\textit{commit\_type} \gets \textsc{semantic\_commit}$
    \State $\Sigma_{\textit{src}} \gets \core$
\Else
    \If{$\textit{verdict\_margin} \ge 1$}
        \State \Comment{\emph{Stage 4b: verdict-fallback path.}}
        \State $\textsc{VerdictPayload} \gets$ canonical-serialise $(v^\star, \textit{top\_count}, \textit{verdict\_margin}, n, f, r)$
        \State $\hstar \gets H\bigl(\texttt{"verdict\_commit"}\,\Vert\,\textsc{VerdictPayload}\,\Vert\,h_\pi\,\Vert\,r\bigr)$
        \State $\widetilde{\mathbf{y}} \gets \bot$
        \Comment{no semantic aggregate}
        \State $\textit{commit\_type} \gets \textsc{verdict\_commit}$
        \State $\Sigma_{\textit{src}} \gets G_{v^\star}$
        \Comment{any $2f{+}1$ valid signatures suffice}
    \Else
        \State \Return $\Abort(\texttt{v2\_both\_paths\_failed:semantic\_core\_failed:} \textit{semantic\_fail\_reason})$
    \EndIf
\EndIf
\Statex
\State \Comment{\emph{Stage 5: collect quorum certificate.}}
\State $\Sigma \gets \{\sigma_j : j \in \Sigma_{\textit{src}},\ \sigma_j \text{ valid on } \hstar\}$
\If{$|\Sigma| < 2f{+}1$}
    \State \Return $\Abort(\texttt{insufficient\_signers},\, |\Sigma|,\, 2f{+}1)$
\EndIf
\State $\cert{\hstar} \gets \Sigma$
\Statex
\If{$\textit{commit\_type} = \textsc{semantic\_commit}$}
    \State \Return $\mathcal{O}^{\textsc{sem}}\bigl(\hstar,\, \pi,\, \cert{\hstar},\, \core,\, \widetilde{\mathbf{y}},\, v^\star\bigr)$
\Else
    \State \Return $\mathcal{O}^{\textsc{ver}}\bigl(\hstar,\, \pi,\, \cert{\hstar},\, G_{v^\star},\, v^\star\bigr)$
\EndIf
\end{algorithmic}
\end{algorithm*}

\subsection{Proposal canonicalisation and encoding}
\label{subsec:phase1} 

A proposal $z_i = (s_i, \rho_i)$ is built per agent. The
structured part $\rho_i$ has the canonical fields
$(\textit{verdict}, \textit{confidence}, \textit{evidence ids},
\textit{rationale}, \textit{claim text})$. The encoder $\cE$
maps $\rho_i$ to a unit-norm embedding via a deterministic
canonicalisation of $\rho_i$ followed by a transformer-based
sentence encoder~\cite{reimers2019,devlin2019};
we use CRSE (a Byzantine-robust trained encoder) by default
and document the encoder identity in the
protocol parameters $\pi$ so that two honest signers reach
byte-identical embeddings. Importantly, the encoder reads only
the canonical proposal text (a deterministic serialisation of
the structured fields); the evidence-id list is part of that
serialisation but is \emph{not} consulted by the protocol's
typed-decision step (Remark~\ref{rem:hcsc-scope}).

\begin{figure*}[t]
    \centering
    \includegraphics[width=\linewidth]{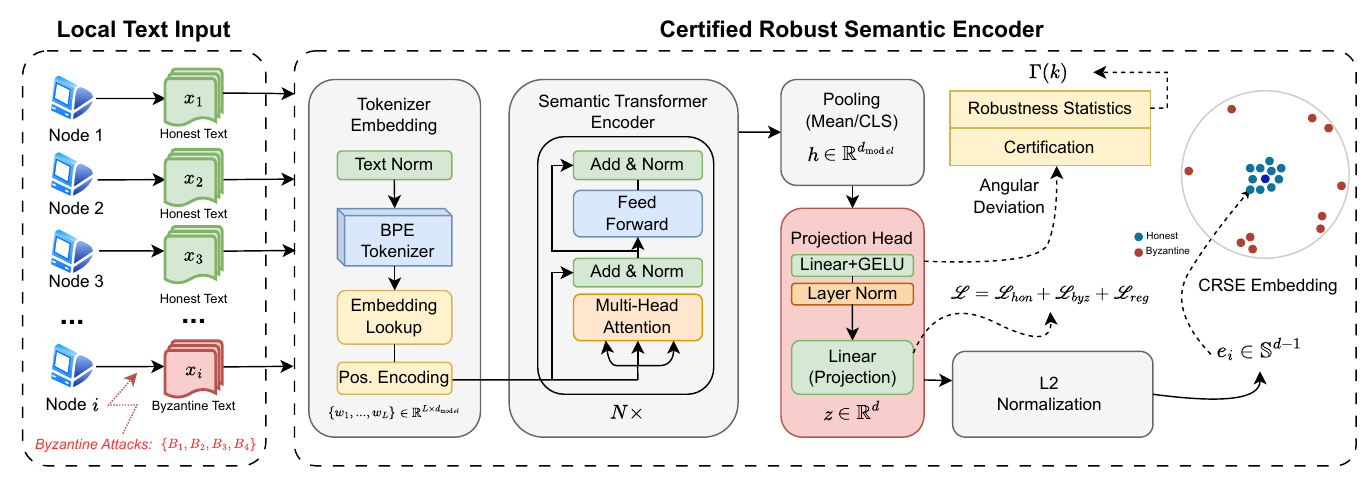}
    \caption{\textbf{Instantiation of the encoder primitive (CRSE).}
    The contrastive training objective minimises intra-honest
    distance and maximises the honest--Byzantine angular margin,
    producing an embedding space in which an honest sub-quorum
    typically satisfies the admissibility predicate of
    Definition~\ref{def:admissible-round}; empirical rates are
    reported in Section~\ref{sec:evaluation}.}
    \Description{A neural network architecture diagram for the Certified Robust Semantic Encoder. It shows a contrastive learning setup that clusters honest intents and pushes away semantic attacks in the embedding space.}
    \label{fig:crse_arch}
\end{figure*}

\subsection{Verdict grouping}

The protocol partitions $V(r)$ into verdict groups
$\{G_v\}_v$ with $G_v = \{i\in V(r) : \rho_i.\textit{verdict} = v\}$.
The candidate verdict group is selected by the
\textit{candidate\_rule} parameter, fixed to
\textsc{largest\_verdict} for the H-CSC main path:
$v^\star = \arg\max_v |G_v|$ with deterministic tie-break by the
fixed priority $\textsc{support} > \textsc{refute} > \textsc{insufficient}$
(applied only on exact ties). The candidate group's
$\textit{verdict\_margin}$ is the difference between the largest
and the second-largest group counts. If
$|G_{v^\star}| < 2f{+}1$ the protocol immediately emits
$\mathcal{O}^{\textsc{abort}}$ with typed reason
\texttt{verdict\_below\_quorum} and halts.

\subsection{Within-verdict semantic finality}
\label{subsec:phase2} 

Within $G_{v^\star}$, the protocol runs the
\textsc{angular\_component} core extraction at radius
$\thalpha$: the largest connected component of the angular-
threshold graph
$\{(i,j) : \angle(\mathbf{e}_i,\mathbf{e}_j) \le \thalpha\}$.
If the largest component has size $< 2f{+}1$, the within-
verdict semantic core check fails. Otherwise, the
admissibility predicate of
Definition~\ref{def:admissible-round} is evaluated on the
selected component: there must exist a centre node
$h \in \core$ such that
$\angle(\mathbf{e}_h, \mathbf{e}_j) \le \thalpha$ for every
$j \in \core$. If the predicate fails, the within-verdict
semantic core check fails. The H-CSC default uses
$\thalpha = 0.65$~rad ($\approx 37.2^\circ$), wider than the
strict-semantic CSC's $\thalpha = 0.55$~rad to accommodate
honest agents that agree on the verdict but use different
phrasings within the rationale.

\subsection{Hierarchical typed-decision step}

The typed-decision step combines the verdict-level and
embedding-level signals:

\begin{enumerate}
\item \textbf{Verdict-quorum gate.} If
$|G_{v^\star}| < 2f{+}1$, output
$\mathcal{O}^{\textsc{abort}}$ with reason
\texttt{verdict\_below\_quorum} (already triggered in the
verdict-grouping subsection above).
\item \textbf{Semantic-path attempt.} If the within-verdict
semantic core is admissible (size $\ge 2f{+}1$ and the radius
predicate passes), proceed to the semantic-commit
construction (\S\ref{sec:protocol-aggregate-digest}).
\item \textbf{Verdict-fallback gate.} Otherwise, if
$\textit{verdict\_margin} \ge 1$ (the
\textsc{margin\_1} fallback condition, fixed in the H-CSC main
path), proceed to the verdict-commit construction.
\item \textbf{Both paths failed.} Otherwise, output
$\mathcal{O}^{\textsc{abort}}$ with a typed reason recording
which gates failed.
\end{enumerate}

The fallback condition $\textit{margin}\ge 1$ is a calibrated
empirical operating point, not a deterministic honest-reference
guarantee: with $n=10$, $f=2$, and $\textit{top\_count}(V) \ge
2f{+}1 = 5$, $\textit{margin}\ge 1$ selects a unique top verdict
group and empirically recovers coverage on real-LLM rounds whose
honest agents agree at the verdict level but disperse at the
rationale level. Deterministic honest-reference validity would
require $\textit{verdict\_margin}^{\min} > f$ instead, as proven
in Theorem~\ref{thm:hcsc-verdict-validity}
(Section~\ref{sec:correctness}). We retain
\textsc{margin\_1} as the operating point because (i) it
preserves coverage on rounds the deterministic $m > f$
condition would abort, and (ii) the empirical
$\texttt{invalid\_hmaj}$ rate at $m = 1$ stays at or below
$0.04$ across all MVR-50 and MVR-100 benchmarks
(Section~\ref{sec:eval-setup}).

\subsection{Aggregation, quantisation, and digest binding}
\label{sec:protocol-aggregate-digest}
\label{subsec:phase3} 

\textbf{Semantic path.} The protocol aggregates the selected
within-verdict core $\core$ by Euclidean geometric
median~\cite{pillutla2022,blanchard2017,yin2018median}:
\begin{equation*}
\ystar \;=\; \arg\min_{\mathbf{y}}\,\sum_{i\in\core}
\|\mathbf{y} - \mathbf{e}_i\|_2,
\end{equation*}
solved iteratively to within numerical tolerance and re-
normalised to the unit sphere. The aggregate $\ystar$ is
quantised by $Q_\eta$ to a deterministic lattice point
$\widetilde{\mathbf{y}} = Q_\eta(\ystar)$. The semantic digest
is
\begin{equation*}
\hstar_{\textsc{sem}} \;=\; H\bigl(\,\texttt{"semantic\_commit"}\,\Vert\,\widetilde{\mathbf{y}}\,\Vert\, h_\pi \,\Vert\, r \,\Vert\, v^\star\,\bigr),
\end{equation*}
where $h_\pi$ is the parameter digest binding the parameter
tuple $\pi$ and the leading byte string
\texttt{"semantic\_commit"} is an explicit \emph{domain
separator} that prevents cross-type digest collision with the
verdict branch below.

\textbf{Verdict path.} The protocol assembles a verdict-only
payload
$\textsc{VerdictPayload}(v^\star, |G_{v^\star}|, \textit{verdict\_margin}, n, f, r)$
as a canonical serialisation. The verdict digest is
\begin{equation*}
\hstar_{\textsc{ver}} \;=\; H\bigl(\,\texttt{"verdict\_commit"}\,\Vert\,\textsc{VerdictPayload} \,\Vert\, h_\pi \,\Vert\, r\,\bigr).
\end{equation*}
A \textsc{verdict\_commit} explicitly does not carry an
embedding aggregate $\widetilde{\mathbf{y}}$; the
$\mathbf{no\_semantic\_aggregate} = \textsc{true}$ flag is
present on the commit object.

\subsection{Quorum certificate}
\label{subsec:phase4} 

For both commit types, the protocol collects a certificate
$\cert{\hstar} = \{\sigma_j\}$ of $2f{+}1$ distinct signers
drawn from the relevant signer set. For the semantic path,
signers may be any $2f{+}1$ distinct valid signers from
$\core$ (the within-verdict admissible core). For the verdict
path, signers may be any $2f{+}1$ distinct valid signers from
$G_{v^\star}$ (the candidate verdict group); the protocol does
not restrict the signer set to the first $2f{+}1$ agents by
id, so commit liveness does not depend on Byzantine
cooperation when $\Sigma_{\textit{src}}$ contains at least
$2f{+}1$ honest agents
(Theorem~\ref{thm:hcsc-termination}). Each $\sigma_j$ is an authenticator over the
appropriate digest under agent $j$'s identity. We use a
\emph{logical} signature simulator that preserves per-signer
uniqueness and the $2f{+}1$ distinct-signer threshold;
production threshold-signature schemes (e.g.\ BLS, FROST) are a
standard substitution and are out of scope for this paper. If
fewer than $2f{+}1$ distinct signatures arrive, the protocol
aborts with typed reason \texttt{insufficient\_signers}
regardless of which path produced the candidate digest.

\subsection{Commit/abort decision}

The per-round decision is a typed value
(Definition~\ref{def:csc-object}): a \textsc{semantic\_commit}
object, a \textsc{verdict\_commit} object, or an \emph{abort}
object with one of the closed-set typed reasons. The protocol
does not produce a ``soft'' commit, a partial commit, or a
commit on a verdict group below $2f{+}1$. Two honest verifiers
that recompute the typed-decision step from the published
commit object agree on the same typed outcome and the same
digest (Theorem~\ref{thm:agreement}). The
\textsc{semantic\_commit} branch is strictly stronger than the
\textsc{verdict\_commit} branch: a \textsc{semantic\_commit} carries
both an embedding-backed aggregate and a verdict-bound digest,
while a \textsc{verdict\_commit} carries only a verdict-level
payload digest with the explicit
$\mathbf{no\_semantic\_aggregate} = \textsc{true}$ flag.
Downstream consumers that require an embedding-backed semantic
representation must check the $\mathbf{commit\_type}$ field on
the returned object before treating it as a semantic
representation.

\subsection{Why H-CSC is not majority vote}

\sloppy
H-CSC differs from \texttt{majority\_vote} and from
\texttt{confidence\_weighted\_vote} on three axes. First,
\textbf{certificate emission}: every H-CSC commit (semantic or
verdict) carries a $2f{+}1$ distinct-signer certificate; the
verdict-only baselines do not. Second, \textbf{typed
finality}: the H-CSC commit object distinguishes
\textsc{semantic\_commit} from \textsc{verdict\_commit}, providing
an embedding-backed digest on rounds where the within-verdict
semantic core is admissible. Third, \textbf{safety floor under
fragile margins}: H-CSC requires $|G_{v^\star}| \ge 2f{+}1$ AND
$\textit{verdict\_margin} \ge 1$ for the verdict path;
\texttt{majority\_vote} commits regardless. The fair
certificate-emitting baseline against H-CSC is therefore not
\texttt{majority\_vote} but B3
(\texttt{certificate\_wrapped\_majority}), which emits the same
certificate envelope. Section~\ref{sec:eval-setup} reports the
B3 comparison explicitly: on the headline coverage and safety
axes H-CSC is numerically close to B3 (point-estimate gaps
$\le 0.02$ in commit rate and $\le 0.002$ in valid-commit
coverage), and the contribution attributable to the embedding
axis is the typed \textsc{semantic\_commit} outcome on $74\%$ /
$72\%$ of rounds.

\section{Correctness Analysis}
\label{sec:correctness}
\label{sec:theory} 

H-CSC does not have \emph{a} correctness guarantee. It has four,
they are of different strengths, they rest on different
assumptions, and --- the point this section is organised
around --- \textbf{they are not nested}. A round can satisfy the
first three and fail the fourth, and one of our own MVR-50 rounds
does exactly that. We therefore state the guarantees as a
\emph{lattice} (\S\ref{sec:guarantee-lattice}), prove the two
cross-level implications that do hold
(\S\ref{sec:level4}), exhibit the realised attack that the
shipped configuration admits (\S\ref{sec:margin-gate}), record a
mismatch between our admissibility \emph{definition} and our
admissibility \emph{implementation}
(\S\ref{sec:ball-vs-component}), and close by discharging the
encoder-determinism assumption that Level~1 rests on
(\S\ref{sec:encoder-determinism}).

The four levels map onto the five theorem statements of v1
(agreement, semantic validity, verdict validity, safe abort,
conditional termination), which are retained unchanged in
substance; safe abort and termination are properties of the
control flow rather than of a commit, and are stated separately
in \S\ref{sec:abort-termination}. The decomposition mirrors the
agreement / validity / termination split standard in classical
Byzantine
agreement~\cite{lamport1982byzantine,castro1999practical,dolev1986reaching,yin2019hotstuff},
refined by the typed commit object of
Definition~\ref{def:csc-object}.

\subsection{The four-level guarantee lattice}
\label{sec:guarantee-lattice}

\begin{table*}[t]
\centering
\caption{\textbf{The four-level guarantee lattice.} L1--L4 rise
in what they claim about a published commit --- form, then
witness, then geometry, then content --- but they are
\textbf{not} a chain: no level implies the next
(Remark~\ref{rem:not-a-chain}). The only cross-level implication
we can prove runs from a \emph{strengthened} L3 core condition to
L4 (Theorem~\ref{thm:det-semantic-validity} and
Lemma~\ref{lem:containment}). Assumption numbers refer to
\S\ref{sec:correctness-assumptions}.}
\label{tab:guarantee-lattice}
\footnotesize
\setlength{\tabcolsep}{4pt}
\setlength{\tabcolsep}{3pt}
\begin{tabular}{@{}>{\raggedright\arraybackslash\hspace{0pt}}p{0.12\linewidth} >{\raggedright\arraybackslash\hspace{0pt}}p{0.22\linewidth} >{\raggedright\arraybackslash\hspace{0pt}}p{0.21\linewidth} >{\raggedright\arraybackslash\hspace{0pt}}p{0.13\linewidth} >{\raggedright\arraybackslash\hspace{0pt}}p{0.13\linewidth}@{}}
\toprule
\textbf{Level} & \textbf{What is guaranteed} & \textbf{Under which assumption} & \textbf{Deterministic or empirical} & \textbf{Commit types covered} \\
\midrule
\textbf{L1} \newline digest-level agreement
& Two honest verifiers that publish in round $r$ publish the same
  $\mathbf{commit\_\allowbreak type}$ and the same digest $\hstar$; two honest
  nodes that abort agree on the typed reason
  (Theorem~\ref{thm:hcsc-agreement},
  Remark~\ref{rem:agreement-abort-abort}).
& Assumption~\ref{assn:auth-channels} (authenticated channels,
  honest sign-once); Assumption~\ref{assn:deterministic-cE}
  (byte-deterministic $\cE$, canonicalisation, $Q_\eta$);
  successful round (\S\ref{subsec:net}).
& Deterministic, \emph{conditional on
  Assumption~\ref{assn:deterministic-cE}} --- which is the weakest
  link (\S\ref{sec:encoder-determinism}).
& \textsc{semantic\_\allowbreak commit}, \textsc{verdict\_\allowbreak commit},
  \textsc{abort} \\
\addlinespace
\textbf{L2} \newline certificate validity
& $\cert{\hstar}$ binds $2f{+}1$ distinct identities, hence at
  least $f{+}1$ honest signatures, to $\hstar$
  (Lemma~\ref{lem:cert-uniqueness}).
& Assumption~\ref{assn:auth-channels} only. No encoder
  dependence, no admissibility conditioning.
& Deterministic. The strongest and cheapest level.
& \textsc{semantic\_\allowbreak commit}, \textsc{verdict\_\allowbreak commit} \\
\addlinespace
\textbf{L3} \newline semantic proximity
& $\angle(\ystar,\mathbf{c}_{\Honest}) \le
  \mathrm{diam}_{\angle}(\core) \le 2\thalpha$ for \emph{every}
  honest $\mathbf{c}_{\Honest}\in\core$, and
  $|\core\cap\Honest|\ge f{+}1$
  (Theorem~\ref{thm:hcsc-semantic-validity}).
& Assumption~\ref{assn:admissible-round} (semantic variant):
  $\core\subseteq G_{v^\star}$ is a $\thalpha$-\emph{ball} of size
  $\ge 2f{+}1$ (Definition~\ref{def:admissible-round}); $\agg$
  returns a point of $\mathrm{conv}(\core)$.
& Deterministic \emph{as specified}. The shipped implementation
  tests $\thalpha$-\emph{connectivity} instead
  (\S\ref{sec:ball-vs-component}).
& \textsc{semantic\_\allowbreak commit} only \\
\addlinespace
\textbf{L4} \newline verdict correctness
& $v^\star = \mathrm{majority}_{\Honest}(V)$, the
  evidence-bounded honest-reference verdict
  (Definition~\ref{def:two-validities}). Two sufficient
  conditions, of different strengths: $\textit{verdict\_\allowbreak margin}
  > f$ gives the honest-reference \emph{plurality}
  (Theorem~\ref{thm:hcsc-verdict-validity});
  $|\core| > (n{+}f)/2$ gives the \emph{strict honest majority}
  of the delivered view
  (Theorem~\ref{thm:det-semantic-validity}).
& One of the two conditions above. \textbf{Neither is enforced at
  the shipped operating point}, where the verdict branch gates at
  $\textit{verdict\_\allowbreak margin}\ge 1 \le f$ and the semantic branch
  applies no margin condition at all
  (\S\ref{sec:margin-gate}).
& \textbf{Empirical} at the shipped operating point; deterministic
  only under the conditions at left, at the coverage cost
  measured in Table~\ref{tab:margin-gate-cost}.
& \textsc{semantic\_\allowbreak commit}, \textsc{verdict\_\allowbreak commit} \\
\bottomrule
\end{tabular}
\end{table*}

L1 is about the \emph{form} of the outcome: honest parties do not
disagree about what was decided. L2 is about the \emph{witness}:
whatever was decided carries an unforgeable, honest-majority-backed
receipt. L3 is about the \emph{geometry} of the payload: the
committed aggregate is not far, in the encoder's metric, from the
honest proposals that back it. L4 is about the \emph{content}: the
verdict is the one the honest agents would have produced. Only
L4 is a statement a downstream consumer of the verdict actually
cares about, and only L4 is not deterministic in the shipped
configuration. That asymmetry is the honest summary of this
paper's formal contribution.

\begin{remark}[L1--L4 is not a chain]
\label{rem:not-a-chain}
None of $\text{L1}\Rightarrow\text{L2}$,
$\text{L2}\Rightarrow\text{L3}$,
$\text{L3}\Rightarrow\text{L4}$ holds, and neither does any
reverse implication.
\begin{itemize}
\item L1 and L2 hold on \emph{every} published commit, including
one whose committed verdict is wrong; they say nothing about
content. They are satisfied by construction on the counterexample
below.
\item L3 fails on every \textsc{verdict\_commit} by definition
(no semantic aggregate exists to bound), while L1, L2 and ---
empirically --- L4 all hold on such rounds. So
$\text{L4}\not\Rightarrow\text{L3}$.
\item L3 can hold while L4 fails. This is not hypothetical:
Example~\ref{ex:tiebreak-capture} is a round of our own MVR-50
evaluation that satisfies L1, L2 and L3 and violates L4.
\end{itemize}
The single cross-level implication that \emph{does} hold is the
strengthened one: a core large enough to satisfy
$|\core| > (n{+}f)/2$ forces both a strict honest majority
(Theorem~\ref{thm:det-semantic-validity}) and a verdict margin
$> f$ (Lemma~\ref{lem:containment}). At the shipped operating
point that condition is not required, and consequently no
implication is available.
\end{remark}

\begin{example}[Tie-break capture: L1 $\wedge$ L2 $\wedge$ L3
$\wedge \neg$ L4]
\label{ex:tiebreak-capture}
Task \texttt{cfever\_01019} of MVR-50, static mode, $n=10$,
$f=2$. The eight honest agents split
$5$ \textsc{insufficient} / $3$ \textsc{support}; the two
Byzantine agents both submit \textsc{support}. The delivered view
is therefore $|G_{\textsc{support}}| = 5$,
$|G_{\textsc{insufficient}}| = 5$,
$\textit{verdict\_margin} = 0$ --- an exact tie. Stage~2's
deterministic tie-break
($\textsc{support} > \textsc{refute} > \textsc{insufficient}$)
selects $v^\star = \textsc{support}$, and
$|G_{v^\star}| = 5 = 2f{+}1$ clears the verdict-quorum gate.
Stage~3 finds all five members of $G_{\textsc{support}}$ mutually
within $\thalpha$, so $\core = G_{\textsc{support}}$ with
$|\core| = 5 \ge 2f{+}1$; Stage~4a emits a
\textsc{semantic\_commit} on \textsc{support} with a Byzantine
share of $2/5 = 0.40$ inside the certified core.

L1 holds (the decision is a deterministic function of $V(r)$).
L2 holds ($5$ distinct signers, hence at least $3$ honest).
L3 holds ($\core$ is a $\thalpha$-ball, so
$\angle(\ystar,\mathbf{c}_{\Honest})\le 2\thalpha$ for each of the
three honest core members). L4 fails: the honest-reference
plurality is \textsc{insufficient} ($5$ versus $3$), and so is the
Climate-FEVER gold label. This single round is the entire
$\texttt{invalid\_hmaj} = 0.02$ that H-CSC reports on static
MVR-50 (\S\ref{sec:eval-setup}).

Two adversaries suffice because the adversary is not attacking the
verdict count --- it is attacking the \emph{tie-break}. Both
deterministic conditions of L4 correctly reject this round:
$\textit{verdict\_margin} = 0 \not> f$, and
$|\core| = 5 \not> (n{+}f)/2 = 6$. The vulnerability is that the
shipped protocol tests neither.
\end{example}

\subsection{Levels 1 and 2: agreement and certificate validity}
\label{subsec:liveness} 
\label{sec:level12}

\begin{lemma}[Deterministic typed digest consistency]
\label{lem:digest-consistency}
Given Assumption~\ref{assn:deterministic-cE} and a fixed
parameter tuple $\pi$, two honest signers that select the same
candidate verdict group $G_{v^\star}$ and the same
within-verdict core $\core$ (when the semantic path triggers)
compute identical
$\widetilde{\mathbf{y}} = Q_\eta(\agg(\core))$ and identical
$\hstar_{\textsc{sem}}$. Two honest signers that select the same
$G_{v^\star}$ and the same verdict-fallback payload (when the
verdict path triggers) compute identical
$\hstar_{\textsc{ver}}$.
\end{lemma}

\begin{proof}[Proof sketch]
For the semantic path: byte-determinism of $\cE$ and of the
canonical proposal text imply identical $\mathbf{e}_i$ for
$i \in \core$; the verdict grouping is a deterministic partition
of $V(r)$; the angular-component extraction with deterministic
tie-break yields identical $\core$; $\agg$ is a deterministic
function of an identical multiset under an identical execution
environment; $Q_\eta$ is a fixed-point round; $H$ is
deterministic on byte input; the round id $r$ and committed
verdict $v^\star$ are equal at both signers. For the verdict
path: the verdict payload is a canonical serialisation of
$(v^\star, |G_{v^\star}|, \textit{verdict\_margin}, n, f, r)$,
which is identical at both signers; the rest follows by
$h_\pi$ and $H$ determinism. Detail in
Appendix~\ref{app:proof-digest}. The clause ``under an identical
execution environment'' is not free, and is the subject of
\S\ref{sec:encoder-determinism}.
\end{proof}

\begin{lemma}[Quorum certificate uniqueness]
\label{lem:cert-uniqueness}
Under Assumption~\ref{assn:auth-channels}, no certificate
$\cert{\hstar}$ of size $\ge 2f{+}1$ can carry two valid
signatures from the same agent. In particular, $\cert{\hstar}$
witnesses at least $f{+}1$ honest signers.
\end{lemma}

\begin{proof}[Proof sketch]
The certificate verifier counts at most one signature per
signer identity: per-signer uniqueness in $\Sigma$ is enforced
by the runner. A Byzantine identity may sign multiple distinct
digests with its own key (signature unforgeability rules out
\emph{others} forging that identity, not the identity itself
equivocating), but each such identity contributes at most one
counted signature to any single certificate. Honest identities
satisfy the sign-once rule of
Assumption~\ref{assn:auth-channels} and produce a signature on
at most one digest per round. With $|\Sigma| \ge 2f{+}1$
distinct identities and at most $f$ Byzantine identities, the
pigeonhole principle gives at least
$(2f{+}1) - f = f{+}1$ honest signatures on the bound digest.
This is identical for \textsc{semantic\_commit} and
\textsc{verdict\_commit} certificates.
\end{proof}

\begin{theorem}[L1: typed digest-level agreement]
\label{thm:hcsc-agreement}
\label{thm:agreement} 
Under Assumptions~\ref{assn:auth-channels} and
\ref{assn:deterministic-cE}, in a \emph{successful round} as
defined in \S\ref{subsec:net} (all honest nodes complete RBC
with the identical delivered view $V(r)$ over $I_r$), if the
protocol publishes typed commit objects $\mathcal{O}_i$ and
$\mathcal{O}_j$ at two honest verifiers in round $r$, then both
objects share the same $\mathbf{commit\_type}$ and the same
digest $\hstar$, and both certificates bind $2f{+}1$ distinct
signer identities to $\hstar$.
\end{theorem}

\begin{proof}[Proof sketch]
The argument does not require quorum-intersection between
$\Sigma_i$ and $\Sigma_j$ (which can be empty for
$n > 3f{+}1$, including the experimental $n = 10$, $f = 2$).
Instead: RBC delivers an identical view $V(r)$ to every honest
node; H-CSC's typed decision on $V(r)$ and the frozen
parameter tuple $\pi$ is deterministic, hence produces a
canonical typed digest $D^\star(r)$ per round; an honest agent
in round $r$ signs at most one digest, namely $D^\star(r)$
(Stage~5 signing rule). Each of $\Sigma_i,\Sigma_j$ contains
at least $(2f{+}1) - f = f{+}1$ honest signers, and every
honest signer in either set signed $D^\star(r)$. Signature
unforgeability (Assumption~\ref{assn:auth-channels}) and the
typed-digest domain separators
(\S\ref{sec:protocol-aggregate-digest}) together force
$\hstar_i = \hstar_j = D^\star(r)$ and matching
$\mathbf{commit\_type}$. The remainder follows from
Lemmas~\ref{lem:digest-consistency} and
\ref{lem:cert-uniqueness}. Detailed proof in
Appendix~\ref{app:proof-agreement}.
\end{proof}

\begin{remark}
\label{rem:agreement-abort-abort}
By the same deterministic-decision argument, two honest nodes
in a successful round also agree on the abort outcome and the
typed reason code when both abort: the canonical typed digest
$D^\star(r)$ is $\bot$ for abort, and the typed reason code is
deterministic in $(V(r),\pi,r)$. The mixed case (one honest
commits, another aborts in the same successful round) cannot
arise: both honest nodes execute identical
Algorithm~\ref{alg:certified-commitment-round} on identical
$(V(r),\pi,r)$, so their commit-or-abort decisions are
identical, and a $2f{+}1$ certificate can only be assembled
when the local protocol returns a typed commit.
\end{remark}

\subsection{Level 3: semantic proximity}
\label{subsec:validity}  
\label{sec:level3}

\begin{theorem}[L3: semantic validity under a semantic-admissible round]
\label{thm:hcsc-semantic-validity}
\label{thm:validity} 
Under Assumption~\ref{assn:admissible-round} (semantic-admissible
variant), if the protocol publishes a \textsc{semantic\_commit}
object $\mathcal{O}^{\textsc{sem}}$ in round $r$ with selected
core $\core \subseteq G_{v^\star}$, then for any honest centre
$\mathbf{c}_{\Honest} \in \core \cap \Honest$,
\begin{equation*}
\angle(\ystar,\mathbf{c}_{\Honest}) \;\le\;
\mathrm{diam}_{\angle}(\core) \;\le\; 2\thalpha.
\end{equation*}
The committed verdict $\textsc{verdict}(\mathcal{O}^{\textsc{sem}}) = v^\star$
is the verdict held by every member of $\core$, and
$|\core \cap \Honest| \ge f{+}1$.
\end{theorem}

\begin{proof}[Proof sketch]
The within-verdict admissibility predicate guarantees a centre
node $h \in \core$ with
$\angle(\mathbf{e}_h, \mathbf{e}_j) \le \thalpha$ for all
$j \in \core$, hence
$\mathrm{diam}_{\angle}(\core) \le 2\thalpha$. The Euclidean
geometric median of points inside an angular ball of radius
$\thalpha$ lies inside the convex hull of those points, hence
inside the same ball after re-normalisation.
$|\core| \ge 2f{+}1$ and $|\Byz| \le f$ give
$|\core \cap \Honest| \ge f{+}1 \ge 1$. Every member of $\core$
shares the verdict $v^\star$ by construction (the core is
extracted within $G_{v^\star}$). Detailed proof in
Appendix~\ref{app:proof-validity}. The proof consumes the
\emph{ball} form of Definition~\ref{def:admissible-round}, which
\S\ref{sec:ball-vs-component} shows the implementation does not
enforce.
\end{proof}

\begin{remark}
Theorem~\ref{thm:hcsc-semantic-validity} is a \emph{relative}
bound: it bounds the deviation of the semantic commit from
honest centres inside the selected core. It does \emph{not}
state that the commit matches an external ground truth, and ---
by Remark~\ref{rem:not-a-chain} --- it does not imply L4. The
$2\thalpha$ ceiling is a \emph{worst-case admissibility bound},
not a tight empirical error bound; at the H-CSC main-path
operating point $\thalpha = 0.65$~rad ($\approx 37.2^\circ$)
this yields a ceiling of $\approx 74.5^\circ$, while the
empirical SCE values reported in
Section~\ref{sec:eval-retained}--\ref{sec:eval-setup} are
substantially smaller.
\end{remark}

\subsubsection{A mismatch between the predicate and its implementation}
\label{sec:ball-vs-component}

Definition~\ref{def:admissible-round} specifies a \emph{ball}:
$\core$ is admissible iff some centre $h\in\core$ satisfies
$\angle(\mathbf{e}_h,\mathbf{e}_j)\le\thalpha$ for every
$j\in\core$. Algorithm~\ref{alg:certified-commitment-round}
Stage~3 as printed also tests this, via
$r^\star=\min_{h\in\core}\max_{j\in\core}\angle(\mathbf{e}_h,\mathbf{e}_j)$.
Our reference implementation does not: it returns
$\textsc{LargestAngularComponent}(G_{v^\star},\thalpha)$ --- the
largest \emph{connected component} of the $\thalpha$-threshold
graph --- and does not re-test the ball predicate on it. We
report this because it changes what
Theorem~\ref{thm:hcsc-semantic-validity} actually certifies for
the numbers in \S\ref{sec:eval-setup}.

\begin{proposition}[Connectivity is a strictly weaker predicate]
\label{prop:ball-component}
Let $\core$ be connected in the $\thalpha$-threshold graph with
$|\core| = k$. Then
$\mathrm{diam}_{\angle}(\core) \le \min\{(k-1)\thalpha,\ \pi\}$,
and this bound is tight. Conversely, every $\thalpha$-ball is
$\thalpha$-connected but not vice versa; for $k\ge 3$ the
containment is strict.
\end{proposition}

\begin{proof}
A connected graph on $k$ vertices has a spanning tree, and any
two vertices are joined by a path of at most $k-1$ edges. Each
edge contributes at most $\thalpha$ to the angular distance, and
$\angle$ satisfies the triangle inequality on $\mathbb{S}^{d-1}$,
so summing along the path gives
$\mathrm{diam}_{\angle}(\core)\le(k-1)\thalpha$; the metric is
bounded by $\pi$. Tightness: place $k$ unit vectors on a common
great circle at consecutive angular spacing exactly $\thalpha$
(possible whenever $(k-1)\thalpha\le\pi$); the path graph is the
threshold graph and the endpoints realise $(k-1)\thalpha$. For
strictness of the containment, the same configuration with
$k = 3$ has no centre within $\thalpha$ of both endpoints once
$2\thalpha > \thalpha$, i.e.\ always. Proof detail and the
resulting restatement of
Theorem~\ref{thm:hcsc-semantic-validity} under the weakened
predicate are in Appendix~\ref{app:ball-vs-component}.
\end{proof}

\paragraph{What this costs, measured.}
Across the $73$ \textsc{semantic\_commit} rounds of our MVR-50
runs, $5$ selected cores ($6.8\%$) are $\thalpha$-connected but
are \emph{not} $\thalpha$-balls, so
Theorem~\ref{thm:hcsc-semantic-validity}'s $2\thalpha$ conclusion
does not formally apply to them. Empirically the bound is not
violated: the largest observed core diameter over all $73$
commits is $0.964$~rad, against $2\thalpha = 1.30$~rad at
$\thalpha=0.65$. Re-running the benchmark with the ball predicate
enforced leaves the headline commit and safety numbers
essentially unchanged.

\paragraph{Recommendation: tighten the implementation, not the
definition.}
The two repairs available are (a) enforce the ball predicate in
code --- exactly the $r^\star \le \thalpha$ test already printed
in Algorithm~\ref{alg:certified-commitment-round} --- or (b) weaken
Definition~\ref{def:admissible-round} to $\thalpha$-connectivity
and replace $2\thalpha$ by $(k-1)\thalpha$ throughout. We
recommend (a), for three reasons. First, it costs almost nothing:
$5/73$ rounds are affected and the headline numbers do not move.
Second, (b) yields a bound that is \emph{vacuous} at the operating
point: with $\thalpha = 0.65$ any core of size $k \ge 6$ gives
$(k-1)\thalpha \ge 3.25 > \pi$, and even the minimum core
$k = 2f{+}1 = 5$ gives $2.6$~rad $\approx 149^\circ$. Third, and
decisively, (b) has the wrong monotonicity: the guarantee would
\emph{weaken} as the core grows, so a round with more corroborating
proposals would carry a \emph{worse} certified bound. A
specification with that property is not worth keeping. Until (a)
ships, the correct reading of the reported semantic commits is:
$68/73$ carry the $2\thalpha$ guarantee, $5/73$ carry only the
measured diameter.

\subsection{Level 4: verdict correctness}
\label{subsec:agreement} 
\label{sec:level4}

This is the level a consumer of the committed verdict cares about,
and the only one that is not deterministic in the shipped
configuration. We give the two conditions under which it
\emph{is} deterministic, prove that one contains the other, and
draw the consequence: at matched deterministic guarantees, H-CSC
cannot cover more rounds than certificate-wrapped majority. That
negative result is the main formal contribution of this section,
because it tells the reader where a difference between the two
protocols \emph{cannot} exist, and therefore where to look
instead.

Throughout, $\mathrm{majority}_{\Honest}(V)$ is evaluated over
the honest agents \emph{in the delivered view},
$\Honest_V := \Honest \cap V(r)$, consistent with
Definition~\ref{def:two-validities} and with the offline
evaluation of \S\ref{sec:eval-setup}. We write
$h_v := |G_v \cap \Honest|$ and $b_v := |G_v \cap \Byz|$, so
$|G_v| = h_v + b_v$ and $\sum_v b_v = |\Byz \cap V(r)| \le f$.

\begin{theorem}[L4 on the verdict branch: honest-reference
plurality under margin $>f$]
\label{thm:hcsc-verdict-validity}
Fix a margin parameter $m := \textit{verdict\_margin}^{\min}$ and
suppose the protocol publishes a \textsc{verdict\_commit} object
$\mathcal{O}^{\textsc{ver}}$ in round $r$ with candidate group
$G_{v^\star}$ and $\textit{verdict\_margin} \ge m$.
\begin{enumerate}
\item \textbf{(Sufficiency.)} If $m > f$, then $v^\star$ is the
\emph{strict} honest-reference plurality: $h_{v^\star} > h_v$ for
every $v \ne v^\star$. Consequently
$v^\star = \mathrm{majority}_{\Honest}(V)$, and the
$\textsc{largest\_verdict}$ tie-break is never invoked in
computing $\mathrm{majority}_{\Honest}(V)$.
\item \textbf{(Tightness.)} For every $m \le f$ the conclusion
fails: there is a delivered view with $|G_{v^\star}|\ge 2f{+}1$,
$\textit{verdict\_margin} = m$, and
$\mathrm{majority}_{\Honest}(V) \ne v^\star$.
\end{enumerate}
For the shipped operating point $m = 1$ with $f = 2$, part (2)
applies and honest-reference validity is an empirical rate, not a
guarantee (\S\ref{sec:eval-setup}). The same statement holds
verbatim for the semantic branch once that branch is gated at
margin $m$ (\S\ref{sec:margin-gate}); the shipped semantic branch
is gated at no margin at all, so not even $m = 1$ applies to it.
\end{theorem}

\begin{proof}
Deferred to Appendix~\ref{app:proof-verdict-validity}.
\end{proof}

\begin{remark}[Correction to the previous version]
v1 stated this theorem for $m > f$ but proved only part~(2):
the ``worst-case allocation'' argument given there establishes
that $m > f$ is \emph{necessary}, not that it is
\emph{sufficient}, and the appendix of v1 stated that the
sufficiency argument was fully developed in the main text, which
it was not. Part~(1) above is the missing direction; the
corrected cross-reference is
Appendix~\ref{app:proof-verdict-validity}.
\end{remark}

\begin{theorem}[L4 on the semantic branch: strict honest majority
under a large core]
\label{thm:det-semantic-validity}
Suppose the protocol publishes a \textsc{semantic\_commit} object
in round $r$ with selected core $\core \subseteq G_{v^\star}$
satisfying
\begin{equation}
\label{eq:det-semantic-condition}
|\core| \;>\; \frac{n+f}{2}.
\end{equation}
Then $v^\star$ is the \emph{strict honest majority} of the
delivered view: strictly more than half of the agents in
$\Honest_V = \Honest \cap V(r)$ hold verdict $v^\star$. In
particular $v^\star = \mathrm{majority}_{\Honest}(V)$, and the
conclusion is strictly stronger than
Theorem~\ref{thm:hcsc-verdict-validity}'s plurality conclusion.
\end{theorem}

\begin{proof}
Deferred to Appendix~\ref{app:proof-det-semantic}.
\end{proof}

\paragraph{Coverage cost.} Condition~\eqref{eq:det-semantic-condition}
requires $|\core| \ge 7$ at $n = 10$, $f = 2$. Gating the semantic
branch on it reduces the \textsc{semantic\_commit} rate from
$0.74 / 0.72$ (static / rushing) to $\mathbf{0.34 / 0.28}$, at
$\texttt{invalid\_hmaj} = 0.00$ in both modes on MVR-50
($N = 50$ tasks). This is the price of turning L4 on the semantic
branch from empirical into deterministic: slightly more than half
of the semantic commits are given up.

\begin{lemma}[Containment]
\label{lem:containment}
If the selected core satisfies $|\core| > (n+f)/2$, then the
delivered view satisfies $\textit{verdict\_margin} > f$.
\end{lemma}

\begin{proof}
Since $\core \subseteq G_{v^\star}$ we have
$|G_{v^\star}| \ge |\core| > (n+f)/2$. The verdict groups
partition the delivered view, so the runner-up group satisfies
\[
\max_{v \ne v^\star} |G_v| \;\le\; |V(r)| - |G_{v^\star}|
\;\le\; n - |G_{v^\star}| \;<\; n - \frac{n+f}{2} \;=\; \frac{n-f}{2}.
\]
Therefore
\[
\textit{verdict\_margin}
= |G_{v^\star}| - \max_{v \ne v^\star} |G_v|
\;>\; \frac{n+f}{2} - \frac{n-f}{2} \;=\; f. \qedhere
\]
\end{proof}

\begin{corollary}[No coverage separation at matched deterministic
guarantees]
\label{cor:coverage-identity}
Define the two \emph{deterministic operating points}:
\begin{itemize}
\item $\textsc{H-CSC}^{\det}$: the H-CSC control flow of
Algorithm~\ref{alg:certified-commitment-round} with the semantic
branch gated on $|\core| > (n+f)/2$ and the verdict branch gated
on $\textit{verdict\_margin} > f$;
\item $\textsc{M3}^{\det}$: certificate-wrapped majority
(baseline B3) gated on $\textit{verdict\_margin} > f$.
\end{itemize}
Then on \emph{every} delivered view, $\textsc{H-CSC}^{\det}$
commits if and only if $\textsc{M3}^{\det}$ commits, and both
commit the same verdict $v^\star$. Hence their commit rates, their
honest-reference invalid-commit rates, and their per-task commit
sets are \emph{identical}, not merely statistically
indistinguishable.
\end{corollary}

\begin{proof}
Both protocols compute the same $v^\star$ from the same delivered
view by the same $\textsc{largest\_verdict}$ rule, so it suffices
to match the commit predicates.
($\Leftarrow$) If $\textsc{M3}^{\det}$ commits then
$\textit{verdict\_margin} > f$ and $|G_{v^\star}| \ge 2f{+}1$, so
$\textsc{H-CSC}^{\det}$'s verdict branch is enabled; whichever
branch H-CSC takes, it commits $v^\star$.
($\Rightarrow$) If $\textsc{H-CSC}^{\det}$ commits, it did so
either on the verdict branch, in which case
$\textit{verdict\_margin} > f$ directly, or on the semantic
branch, in which case $|\core| > (n+f)/2$ and
Lemma~\ref{lem:containment} again gives
$\textit{verdict\_margin} > f$. Either way $\textsc{M3}^{\det}$
commits. Thus the semantic-commit predicate is \emph{strictly
stronger} than the verdict-commit predicate at these operating
points, the semantic branch is redundant as a coverage
mechanism, and the union of the two branches equals the verdict
branch alone.
\end{proof}

\paragraph{Empirical confirmation.} We ran both deterministic
operating points on MVR-50 ($N = 50$ tasks, $n=10$, $f=2$, paired
static and rushing). The per-task commit sets are identical: the
McNemar off-diagonal cells are $b = c = 0$, i.e.\ \textbf{zero
discordant tasks} in either mode. Both protocols commit
$\mathbf{0.72}$ (static) $/\ \mathbf{0.64}$ (rushing) at
$\texttt{invalid\_hmaj} = 0.00$. The identity is exact, as
Corollary~\ref{cor:coverage-identity} requires.

\begin{remark}[Why this is the navigational result]
\label{rem:navigational}
Reviewers of the previous version observed that H-CSC and
certificate-wrapped majority were statistically indistinguishable
on MVR-50 and asked for a larger or better-powered experiment.
Corollary~\ref{cor:coverage-identity} shows that no experiment
will separate them \emph{on coverage at matched deterministic
guarantees}, because the separation is impossible by
construction, not merely undetectable at $N = 50$. This changes
what the empirical programme should be. It rules out the
coverage axis and the deterministic-safety axis, and leaves
exactly two places where a semantic layer could earn its keep:
(i) the \emph{payload} --- what the digest commits to, which the
verdict-only protocol cannot express at all; and (ii) the
\emph{admissibility instrument} --- what predicate decides which
proposals may back a commit, which is a question about the
instrument and not about the protocol wrapping it. Direction (i)
is addressed by the typed commit object of
Definition~\ref{def:csc-object} and evaluated in
\S\ref{sec:eval-retained}; direction (ii) is the subject of
\S\ref{sec:eval-r2}, and is where we find that the
learned encoder is not the right instrument.

At the \emph{shipped} operating point ($m = 1$ on the verdict
branch, no margin condition on the semantic branch)
Corollary~\ref{cor:coverage-identity} does not apply, and a
difference is possible in principle. Measured, it is one task:
the H-CSC and B3 decision sets differ on $1/50$ tasks in static
mode and $0/50$ in rushing mode (\S\ref{sec:eval-setup}). The one
task is Example~\ref{ex:tiebreak-capture}, and it is a difference
in H-CSC's disfavour.
\end{remark}

\subsection{The margin gate: a vulnerability in our own protocol}
\label{sec:margin-gate}

\paragraph{The defect.} The verdict-fallback branch of
Algorithm~\ref{alg:certified-commitment-round} is gated at
$\textit{verdict\_margin}\ge 1$, and \S\ref{sec:protocol} presents
that gate as part of H-CSC's ``safety floor under fragile
margins''. The semantic branch is not gated at all. Stage~2
aborts only when $|G_{v^\star}| < 2f{+}1$; Stage~3 tests only core
size and the admissibility radius; Stage~4a constructs the
semantic digest without ever consulting
$\textit{verdict\_margin}$. Because the semantic branch is
attempted \emph{first}, the protocol's principal path --- the one
that fires on $74\%/72\%$ of rounds --- effectively runs at
$m_{\min} = 0$. The margin gate that the paper advertises is
applied only to the fallback the paper describes as weaker.

\paragraph{The realised attack.}
Example~\ref{ex:tiebreak-capture} is the exploitation of exactly
this gap, found in our own evaluation rather than constructed for
it. Two Byzantine agents do not attempt to win the verdict count;
they engineer an exact tie ($\textit{verdict\_margin} = 0$) and
let the protocol's own deterministic tie-break hand them the
candidate group. Because the tie-break also determines which
group the semantic core is extracted \emph{within}, the Byzantine
pair is inside the candidate group by construction, and a core of
exactly $2f{+}1$ forms around them with a $0.40$ Byzantine share.
The commit is a \textsc{semantic\_commit}, not a
\textsc{verdict\_commit}; any description of this round as a
verdict-commit failure is incorrect. Note also that the attack is
\emph{cheaper} than a verdict flip: flipping the honest plurality
from $5$--$3$ to a Byzantine win would need three corrupted
agents, but capturing the tie-break needs two.

\paragraph{The fix.} Four lines in Stage~3, applying the margin
condition to both branches:

\begin{algorithmic}
\State \Comment{\emph{Stage 3 (fix): gate the semantic branch on the margin too.}}
\If{$\textit{verdict\_margin} < \textit{verdict\_margin}^{\min}$}
    \State $\textit{semantic\_path\_ok} \gets \textsc{false}$
    \State $\textit{semantic\_fail\_reason} \gets \texttt{margin\_below\_threshold}$
\EndIf
\end{algorithmic}

\noindent
placed immediately after the core-size and radius checks. The
closed abort-reason set of Definition~\ref{def:csc-object} gains
one element, \texttt{margin\_below\_threshold}; no other part of
the protocol changes, and Theorems~\ref{thm:hcsc-agreement},
\ref{thm:hcsc-semantic-validity}, \ref{thm:hcsc-safe-abort} and
\ref{thm:hcsc-termination} are unaffected, since the new gate is a
deterministic function of $\Psi(r)$
(Definition~\ref{def:finality-signal}) evaluated before any digest
is constructed. With the gate in place,
Theorem~\ref{thm:hcsc-verdict-validity} applies to \emph{both}
branches at the same $m$.

\paragraph{The cost, measured.}
Table~\ref{tab:margin-gate-cost} sweeps the shared margin
threshold on MVR-50 ($N = 50$ tasks, paired modes).

\begin{table}[t]
\centering
\caption{Cost of gating \emph{both} branches at
$\textit{verdict\_margin}\ge m$. MVR-50, $N=50$ tasks, $n=10$,
$f=2$, $\thalpha=0.65$. Row $m=0$ is the shipped configuration
(semantic branch ungated; verdict branch at $\ge 1$); rows
$m\ge 1$ gate both branches at $\ge m$. CR $=$ commit rate;
$\texttt{inv\_hmaj}$ is the honest-reference invalid-commit rate
(lower better). $m = f{+}1 = 3$ is the deterministic operating
point of Theorem~\ref{thm:hcsc-verdict-validity} and of
Corollary~\ref{cor:coverage-identity}.}
\label{tab:margin-gate-cost}
\footnotesize
\setlength{\tabcolsep}{5pt}
\begin{tabular}{@{}l c c c c@{}}
\toprule
& \multicolumn{2}{c}{\textbf{static}} & \multicolumn{2}{c}{\textbf{rushing}} \\
\cmidrule(lr){2-3}\cmidrule(lr){4-5}
$m$ & CR & inv\_hmaj & CR & inv\_hmaj \\
\midrule
$0$ (as shipped)      & $0.90$ & $0.02$ & $0.92$ & $0.00$ \\
$\mathbf{1}$ (\textbf{recommended}) & $\mathbf{0.88}$ & $\mathbf{0.00}$ & $\mathbf{0.92}$ & $\mathbf{0.00}$ \\
$2$                   & $0.86$ & $0.00$ & $0.84$ & $0.00$ \\
$3\ (=f{+}1)$         & $0.72$ & $0.00$ & $0.64$ & $0.00$ \\
\bottomrule
\end{tabular}
\end{table}

We adopt $m = 1$. It removes the only honest-reference-invalid
commit either mode produces, at a cost of $0.02$ commit rate in
static mode and nothing in rushing mode. It does \emph{not} buy a
deterministic L4 guarantee --- $m = 1 \le f$, so
Theorem~\ref{thm:hcsc-verdict-validity}(2) still applies --- and we
do not claim otherwise; it closes one specific, realised attack.
The deterministic point $m = 3$ is available and costs $0.18$
(static) and $0.28$ (rushing) of commit rate, at which point
Corollary~\ref{cor:coverage-identity} says the protocol has become
certificate-wrapped majority.

\paragraph{Honest accounting.} We report this as a defect found in
our own evaluation, not as a feature. Three things follow. First,
the $\texttt{invalid\_hmaj} = 0.02$ that v1 reported for static
H-CSC is attributable to an unenforced gate, not to encoder error
or to the fallback path. Second, the claim in v1 that the
verdict-fallback margin condition is what defends against
same-verdict Byzantine proposals
(Remark~\ref{rem:rep-boundary-defence}) was wrong in the direction
that mattered, because the branch that actually fires does not
apply it. Third, a rate of $0/50$ in rushing mode should not be
read as evidence that rushing is safe: the Clopper--Pearson $95\%$
upper bound on a $0/50$ observation is $0.058$, which is wider
than the entire effect being discussed.

\section{Evaluation}
\label{sec:evaluation}

Version~1 of this work asked whether H-CSC \emph{matches} the
strongest verdict-only certificate on safety and coverage. It
does --- and the containment lemma
(Lemma~\ref{lem:containment}) now explains why it must: at
matched deterministic guarantees, no coverage separation from
verdict-only certification is available to \emph{any} protocol
of this shape. That closes the question v1 was built around and
opens the one that actually matters. Typed finality needs an
\emph{admissibility instrument}: some computable predicate that
decides whether a set of proposals is close enough together to
be committed as one semantic object. H-CSC instantiates that
instrument with a learned encoder. This section asks whether
that was the right choice, and reports that on the evidence we
have it was not.

\paragraph{A standing note on sample size.}
The results of
\S\S\ref{sec:eval-r1}--\ref{sec:eval-r6} come from two pilot
studies run at small scale ($N=10$ tasks / $40$ attacks, and
$N=8$ tasks / $3$ output schemas respectively). We state $n$ in
every table caption and in the text of every claim. They are
reported because their \emph{direction} is consistent, large,
and mechanistically explained, not because they are
tight. Where an effect is smaller than what the design can
resolve we say so (\S\ref{sec:eval-negative}, H5--H6). This is
a preprint revision; we prefer a stated-$n$ preliminary finding
to a confident one we cannot support.

\subsection{Discharging the encoder-determinism assumption}
\label{sec:encoder-determinism}

Level~1 requires two honest nodes to derive byte-identical digests, so
Assumption~\ref{assn:deterministic-cE} has to hold exactly, not approximately. It
does not. Re-encoding the MVR-50 canonical texts with the same CRSE checkpoint in a
second execution environment reproduces the cached embeddings to
$3.6\times10^{-7}$, which is reproducibility, not the bit-exactness that a hash
input demands; coarsening the quantiser $\eta$ narrows the disagreement probability
without eliminating it.

A lexical admissibility instrument removes the problem rather than mitigating it.
Token-set membership is exact, so two honest nodes computing the same sketch over
the same signed proposals agree bit-for-bit by construction, and
Assumption~\ref{assn:deterministic-cE} is discharged instead of assumed. This is
the theoretical payoff of the empirical result in \S\ref{sec:eval-r2}: the
instrument that wins on detection is also the one that repairs the foundation the
digest guarantees stand on. The full argument, including the collision arithmetic
for the floating-point instrument and the per-round disagreement bound under
quantiser coarsening, is in Appendix~\ref{app:moved}.

\subsection{Benchmarks, Operating Point, and Reproducibility}
\label{sec:eval-setup}

\textbf{We now read this number differently than v1 did.}
$\mathrm{AUC}=0.9946$ is a \emph{cross-verdict} separability
measurement: BCS-v1's poisons target the anchor's verdict.
\S\ref{sec:eval-r2} measures the same encoder's
\emph{within-verdict} separability and obtains $0.62$--$0.82$.
The gap between those two numbers is the subject of this
section.

\begin{table}[t]
\centering
\caption{BCS-v1 commitment results per Byzantine ratio
(20 trials/bucket; CRSE encoder; geometric-median aggregator).
\texttt{strict-semantic CSC} is the angular-component ablation
of H-CSC (\S\ref{sec:hcsc-strict-csc}); \texttt{all\_nodes} is
the no-filter baseline. SCE: angular distance from committed
aggregate to honest reference centre (deg). Beyond-BFT buckets
($n=30 < 3f{+}1$) are diagnostic-only.}
\label{tab:bcs}
\small
\begin{tabular}{c c c c c c}
\toprule
& & \multicolumn{2}{c}{\texttt{strict-semantic CSC}} & \multicolumn{2}{c}{\texttt{all\_nodes}} \\
\cmidrule(lr){3-4}\cmidrule(lr){5-6}
byz\_ratio & $n$ & commit & SCE$^\circ$ & commit & SCE$^\circ$ \\
\midrule
0.0 & 50 & 1.00 & 2.04 & 1.00 & 0.00 \\
0.1 & 50 & 1.00 & 0.81 & 1.00 & 0.93 \\
0.2 & 50 & 1.00 & 0.56 & 1.00 & 1.85 \\
0.3 & 50 & 0.80 & 0.31 & 0.80 & 3.60 \\
0.4 & 30 & 0.00 & --   & 0.00 & --   \\
0.5 & 30 & 0.00 & --   & 0.00 & --   \\
\bottomrule
\end{tabular}
\end{table}

\paragraph{MVR-50 and MVR-100 (real LLM agents).}
MVR-50 consists of 50 Climate-FEVER claim-verification tasks
(the schema and loaders are dataset-agnostic; SciFact
instantiation remains future work), each run with $n=10$ agents
and $f=2$. Honest agents propose through 5 prompt profiles
(\textit{evidence-focused}, \textit{skeptical},
\textit{concise}, \textit{uncertainty-aware},
\textit{domain-specialist}); Byzantine agents propose through
the 4 attack types of \S\ref{sec:eval-attack-families},
Family~I, filling $100$ Byzantine slots. Every slot is paired
across \emph{static} (Byzantine does not see honest broadcasts)
and \emph{rushing} (it does) modes. All proposals are generated
with \texttt{gpt-4o-mini}~\cite{openai2024gpt4o} under
task-aware schema validation.
MVR-100 is the same pipeline at 100 tasks (800 honest $+$ 200
static $+$ 200 rushing Byzantine proposals per agent LLM), used
for the cross-model check in \S\ref{sec:eval-retained}.
The H-CSC main path uses $\thalpha = 0.65$~rad,
$\textit{candidate\_rule}=\textsc{largest\_verdict}$,
$\textit{fallback}=\textsc{margin\_1}$, and
$\textit{topology\_gate}=\bot$. As
\S\ref{sec:margin-gate} records, the shipped semantic
branch does \emph{not} apply the margin condition, so the
principal configuration runs the semantic path at
$\marginmin = 0$; \S\ref{sec:eval-retained} reports the cost of
closing that gap.

\paragraph{Metrics.}
$\CR$: commit rate. $\SEM$: \textsc{semantic\_commit} coverage, the
fraction of \emph{all} rounds committed with an
embedding-backed digest (verdict-only methods give $\SEM = 0$ by
construction; verdict-commit coverage is $\CR - \SEM$, abort is
$1-\CR$). $\INVgold$ and $\INVhmaj$: the two validity-mismatch
rates of Definition~\ref{def:two-validities}, lower better;
$\INVhmaj$ (against the evidence-bounded honest-reference
plurality) is the protocol-side safety metric, $\INVgold$ a
dataset-quality diagnostic. The two references diverge on
$14/50$ tasks ($28\%$). $\byzcore$: mean Byzantine fraction in
the selected core. $\VCC$: valid-commit coverage. For the new
instrument study we add AUROC (attack proposals scored against
honest proposals) and TPR at a $5\%$ honest false-positive rate
--- the fraction of attacks caught by a threshold set at the
$95$th percentile of the honest score distribution. TPR@5\% is
the operationally meaningful number: it is what a deployment
that refuses to discard more than $5\%$ of honest work can
actually buy.

\paragraph{Reproducibility, determinism, and contamination.}
Four checks, all of which a reader is entitled to demand of a
paper whose theorems assume a deterministic encoder
(Assumption~\ref{assn:deterministic-cE}).
\emph{(i) Encoder determinism.} Re-encoding the MVR-50 canonical
texts from the frozen CRSE checkpoint reproduces the cached
embeddings to a minimum cosine similarity of $0.99999990$ and a
maximum element-wise deviation of $3.6\times10^{-7}$. This is
reproducibility on \emph{one} machine and library stack; it is
not bit-exactness, and Assumption~\ref{assn:deterministic-cE}
remains an assumption for the encoder-based instrument (see
\S\ref{sec:eval-r2}).
\emph{(ii) Protocol-layer cost.} The full protocol-layer sweep
--- $6\,000$ (method $\times$ parameter $\times$ mode $\times$
task) configurations --- executes in $1.3$~s on CPU, so every
sweep in this section is cheap to reproduce.
\emph{(iii) Signing behaviour.} Byzantine agents in our
experiments \emph{do} sign candidate digests
(\texttt{signer\_policy = selected\_core}); the
\texttt{insufficient\_signers} abort fires $0$ times across all
$6\,000$ per-task records.
\emph{(iv) Encoder contamination.} All $488$ CRSE training
anchors are Wikipedia \texttt{science\_tech} passages; exactly
$1/488$ titles contains a climate term; the maximum
title-to-MVR-claim Jaccard overlap is $0.176$ (mean $0.059$)
and the title intersection with MVR is empty. The 20-task
calibration pilot used to select $\thalpha$ and the fallback
rule is disjoint from MVR-50 and MVR-100 (overlap $0$). We find
no evidence of train/test contamination.

\subsection{What v1's Attack Suite Could Not Reach}
\label{sec:eval-attack-audit}

Before adding experiments we must withdraw a claim. Auditing the
generator that produced v1's Byzantine proposals establishes
that the suite was structurally incapable of exercising the
component the paper is named after.

\paragraph{The semantic layer was unreachable by construction.}
The attack selector forces every Byzantine agent to target a
verdict \emph{different} from the honest plurality. The H-CSC
semantic core, however, is extracted \emph{inside} the candidate
verdict group (Algorithm~\ref{alg:hcsc-round}, Stage~3). A
Byzantine agent that flips the verdict is therefore assigned to
a different group and is \emph{structurally ineligible} for the
core, regardless of what its embedding looks like. Whatever the
angular predicate does or does not do, it was never the thing
keeping those agents out. Measured directly: under rushing,
$0/100$ Byzantine proposals preserve the honest-plurality
verdict; under static, $22/100$ do, and those are incidental
rather than designed.

\paragraph{Consequence.}
The v1 result ``$\byzcore = 0.000$ under rushing, and rushing
never produces an invalid commit'' is a property of the
\emph{attack policy}, not of the defence. We withdraw the
safety reading v1 attached to it. The corresponding cells of
Table~\ref{tab:mvr50} are retained because they are correct
measurements of what was run; they are no longer offered as
evidence that the semantic predicate filters anything.

\paragraph{Two further corrections to v1's setup description.}
\emph{(a) Attack balance.} v1 stated that the four attack types
were balanced $25/25/25/25$ across the $100$ Byzantine slots and
paired across modes. That is true of static only. The rushing
slots are $29/9/41/21$ (Table~\ref{tab:attack-composition}).
Consequently \textbf{every static-versus-rushing difference in
Table~\ref{tab:mvr50} confounds attack \emph{timing} with attack
\emph{composition}} and cannot be read as a rushing effect. We
have not re-run the generation to fix the imbalance; we flag it
instead.
\emph{(b) Unanimity.} On $40\%$ of MVR-50 tasks the honest
agents are unanimous. On those rounds every method under
comparison commits the same verdict for the same reason, and
the protocol's discriminating machinery is not exercised. The
effective task base for any method comparison is therefore
closer to $30$ than to $50$, which
\S\ref{sec:eval-negative} (H5) turns into a power statement.

\begin{table}[t]
\centering
\caption{Correction to the v1 setup description: realised
Byzantine attack-type composition over the $100$ Byzantine
slots of MVR-50. v1 reported $25/25/25/25$ for both modes.
Static is as reported; rushing is not.}
\label{tab:attack-composition}
\small
\setlength{\tabcolsep}{4pt}
\begin{tabular}{@{}l cccc@{}}
\toprule
mode & \textsc{B1} polarity & \textsc{B2} omission & \textsc{B3} causality & \textsc{B4} halluc. \\
\midrule
static  & 25 & 25 & 25 & 25 \\
rushing & 29 & \phantom{0}9 & 41 & 21 \\
\bottomrule
\end{tabular}
\end{table}

\subsection{Two Attack Families}
\label{sec:eval-attack-families}

\paragraph{Family I --- verdict-flipping (\textsc{B1}--\textsc{B4}).}
v1's suite, drawn from the textual adversarial-attack
literature~\cite{jin2020textfooler,alzantot2018,li2020bertattack,morris2020textattack}:
\textsc{B1} polarity flip, \textsc{B2} evidence omission,
\textsc{B3} false causality, \textsc{B4} on-topic
hallucination. Each targets a verdict other than the honest
plurality. By \S\ref{sec:eval-attack-audit}, this family
exercises the verdict layer and the certificate layer, and
cannot exercise the semantic layer.

\paragraph{Family II --- verdict-preserving evidential poisoning
(\textsc{B6}--\textsc{B7}).}
The family the architecture is supposed to defend against, and
the one Definition~\ref{def:semantic-poisoning} calls Mode~B: an
adversary that \emph{agrees with the honest plurality on the
verdict} and corrupts only the reasoning that supports it. Such
an agent is admitted to the candidate verdict group by
construction, so the semantic predicate is the only thing
standing between it and the committed payload. We instantiate
two generators.\footnote{The numbering continues
\textsc{B1}--\textsc{B4}; \textsc{B5} is left unused so that the
two families are separated by a gap in the index.}
\begin{description}[leftmargin=1.1em,itemsep=1pt,topsep=2pt]
\item[\textsc{B6} scope shift.] Same verdict, same \emph{real}
evidence ids, but the conclusion's scope is quietly widened or
narrowed --- time span, quantifier, or causal strength.
\item[\textsc{B7} reasoning substitution.] Same verdict, same
class of evidence, but the reasoning chain is replaced by one
the cited evidence does not license.
\end{description}
Each generator is run at two \emph{evasion levels}:
\emph{naive} (poison only) and \emph{evasive}, which
additionally instructs the adversary to mimic the honest
agents' vocabulary, sentence structure, length, and confidence
band. The evasive level is the $\thalpha$-aware adaptive
adversary: it is trying to stay inside the admissibility ball.

\paragraph{Pilot scale and compliance --- read this before the
numbers.}
Family~II is a \textbf{preliminary study}: $N=10$ Climate-FEVER
tasks (stratified by gold label: 4 \textsc{insufficient},
3 \textsc{refute}, 3 \textsc{support}), $4$ attacks per task,
$40$ attacks and $80$ honest proposals in total. Attacks were
generated with \texttt{gpt-4o-mini} at a cost of \$0.014, and
all $40/40$ are schema-compliant in the sense that matters
here: each emits the honest-plurality verdict and cites only
evidence ids that actually exist in the task. Generated attacks
are substituted into the two Byzantine slots of the real MVR-50
views, re-encoded with the frozen CRSE checkpoint, and run
through the real protocol implementation --- these are
generated attacks evaluated by executed code, not
counterfactual relabellings. Ten tasks is not enough to
estimate a rate to better than roughly $\pm 0.15$; it is enough
to establish the presence of an effect whose point estimates
run to $100\%$.

\subsection{Result 1: The Semantic Predicate Does Not Filter
Verdict-Preserving Attacks}
\label{sec:eval-r1}

Table~\ref{tab:b67-penetration} reports, for each of the four
Family~II configurations, the fraction of attacks that (a) land
inside the $\thalpha$ ball around the honest geometric median,
(b) are admitted to the H-CSC selected core, and (c) are
admitted to M3's certified verdict group.

\begin{table}[t]
\centering
\caption{Result 1: Family~II penetration.
\textbf{Preliminary; $n=10$ attacks per row, $40$ total, $N=10$
tasks.} Real CRSE encodings, real protocol code. ``in
$\thalpha$ ball'': within $0.65$~rad of the honest geometric
median. ``core'': admitted to the H-CSC selected core $\core$.
``M3 group'': admitted to the certified verdict group of the
verdict-only baseline.}
\label{tab:b67-penetration}
\small
\setlength{\tabcolsep}{5pt}
\begin{tabular}{@{}l l c c c@{}}
\toprule
evasion & attack & in $\thalpha$ ball & \textbf{into $\core$} & into M3 group \\
\midrule
naive   & \textsc{B6} scope shift & 60\% & \textbf{100\%} & 100\% \\
naive   & \textsc{B7} reasoning subst. & 30\% & \textbf{70\%} & 100\% \\
evasive & \textsc{B6} scope shift & 90\% & \textbf{100\%} & 100\% \\
evasive & \textsc{B7} reasoning subst. & 40\% & \textbf{80\%} & 100\% \\
\bottomrule
\end{tabular}
\end{table}

\paragraph{Round level.}
Under both evasion levels H-CSC commits on $10/10$ rounds and
\textbf{$100\%$ of committed payloads are poisoned --- exactly
the same rate as M3}. $\INVhmaj$ is $1$ for both methods and
carries no information about this attack, because the attack
does not touch the verdict. On the very attack family the semantic layer exists to stop, it buys nothing
over the verdict-only baseline.

\paragraph{The gap between ball and core.}
\textsc{B7} lands inside the $\thalpha$ ball on only $30\%$ /
$40\%$ of attempts but enters the core on $70\%$ / $80\%$. The
excess is the implementation's use of a largest \emph{connected
component} where Definition~\ref{def:admissible-round}
specifies a \emph{ball}: transitive chaining through honest
proposals pulls in points the ball criterion would have
excluded (\S\ref{sec:eval-retained}, ``core-extractor
fidelity''). Tightening the extractor to a strict ball would
recover part of \textsc{B7} but nothing of \textsc{B6}, which
lands inside the ball on $60\%$ / $90\%$ of attempts.

\paragraph{Evasion works in the intended direction.}
Asking the adversary to imitate honest surface form raises
\textsc{B6}'s in-ball rate from $60\%$ to $90\%$ and
\textsc{B7}'s from $30\%$ to $40\%$. The admissibility ball is
a surface-form-sensitive object, and an adversary told about it
in prose --- not even given gradient access --- moves inside it.

\subsection{Result 2: Which Instrument Should Decide
Admissibility?}
\label{sec:eval-r2}

Result~1 says the learned encoder fails. It does not say the
architecture fails, because the encoder is one choice of
instrument for one predicate in it. This subsection replaces
that choice with four alternatives and measures them on the same attacks: $200$ attacks and $400$ honest proposals over
$50$ tasks for the CRSE, lexical and evidence-id instruments, and the earlier
$10$-task subset ($40$ attacks, $80$ honest) for the LLM judge and NLI. To our knowledge
this is the first systematic comparison of admissibility
instruments for a semantic-agreement protocol.

\paragraph{The five instruments.}
All five are scored so that \emph{higher means more suspicious},
and all five are computable by an honest node from the
delivered view alone (no gold labels, no role labels).
\begin{description}[leftmargin=1.1em,itemsep=1pt,topsep=2pt]
\item[CRSE angular distance (the shipped instrument).] Angular
distance from the proposal's CRSE embedding to the geometric
median of the honest embeddings, leave-one-out for honest
proposals. Requires the $442$~MB fine-tuned checkpoint and a
forward pass per proposal.
\item[Lexical conformity.] Let the round's \emph{consensus
vocabulary} be the set of lower-cased whitespace tokens
appearing in more than half of the other proposals' rationale
and proposal text. The score is
$1 - \mathrm{Jaccard}(\text{tokens}(\rho_i),\ \text{consensus
vocabulary})$. No model, no parameters, no training, no
floating point.
\item[Evidence-id overlap.] The same construction over cited
evidence ids instead of tokens: $1 - \mathrm{Jaccard}$ against
the set of ids cited by more than half of the other proposals.
This is the cheap provenance anchor a verdict certificate could
bind for free.
\item[LLM judge.] \texttt{gpt-4o}, role-blind, shown the claim,
the cited evidence, the verdict and the rationale, asked to
score how well every substantive assertion is supported
($0$--$10$); the score is $10$ minus that value. One API call
per proposal.
\item[Off-the-shelf NLI.] A \texttt{deberta-v3} cross-encoder
NLI model; premise $=$ concatenated cited evidence, hypothesis
$=$ each rationale sentence; the score is the maximum
$P(\text{contradiction})$ over sentences.
\end{description}

\begin{table*}[t]
\centering
\caption{\textbf{Result 2 (centrepiece): the admissibility
instrument decides everything.} Cells are AUROC / TPR at a
$5\%$ honest false-positive rate, attacks versus honest
proposals. \textbf{$50$ tasks, $50$ attacks per cell, $200$ attacks total, against
$400$ honest proposals, for the CRSE, lexical and evidence-id rows.
The LLM-judge and NLI rows are the earlier $10$-task subset
($n=10$ per cell, $40$ attacks, $80$ honest; $40$ honest for the judge)
and are marked accordingly.}
The NLI row was run per attack type pooled over both evasion
levels ($n=20$ each), and is reported spanning the two columns
it covers. ``Det.'': deterministic and reproducible across
nodes without a shared model artefact.
Best AUROC per column in \textbf{bold}.}
\label{tab:instrument}
\small
\setlength{\tabcolsep}{5pt}
\begin{tabularx}{\linewidth}{@{}>{\raggedright\arraybackslash}X >{\raggedright\arraybackslash}X c cc cc@{}}
\toprule
& & & \multicolumn{2}{c}{\textsc{B6} scope shift} & \multicolumn{2}{c}{\textsc{B7} reasoning substitution} \\
\cmidrule(lr){4-5}\cmidrule(lr){6-7}
instrument & signal & Det. & naive & evasive & naive & evasive \\
\midrule
CRSE angular distance \emph{(shipped)} & learned embedding geometry & assumed$^\dagger$
  & 0.631 / 0\% & 0.621 / 2\% & 0.744 / 6\% & 0.724 / 8\% \\
\textbf{Lexical conformity} & token-set Jaccard & \textbf{yes}
  & \textbf{0.900} / \textbf{38\%} & \textbf{0.865} / \textbf{40\%} & \textbf{0.982} / \textbf{80\%} & \textbf{0.963} / \textbf{68\%} \\
Evidence-id overlap & cited-id Jaccard & yes
  & 0.565 / 0\% & 0.569 / 0\% & 0.614 / 0\% & 0.552 / 0\% \\
LLM judge (\texttt{gpt-4o}) & role-blind groundedness & no
  & 0.807 / 30\% & 0.775 / 10\% & 0.877 / 40\% & 0.859 / 50\% \\
Off-the-shelf NLI (\texttt{deberta-v3}) & max $P(\text{contradiction})$ & assumed$^\dagger$
  & \multicolumn{2}{c}{0.654 / 25\%} & \multicolumn{2}{c}{0.651 / 5\%} \\
\bottomrule
\end{tabularx}
\smallskip\par
\noindent\footnotesize $^\dagger$Deterministic only under
Assumption~\ref{assn:deterministic-cE}: identical results
require every honest node to run the same checkpoint, the same
library versions, and compatible floating-point hardware. Our
re-encoding check (\S\ref{sec:eval-setup}) reproduces cached
embeddings to $3.6\times10^{-7}$ on one stack, which is
reproducibility, not the bit-exactness the digest binding
assumes.
\end{table*}

\paragraph{Finding.}
\textbf{Lexical conformity dominates the learned encoder in
every one of the four cells}, by $0.25$--$0.27$ AUROC on
\textsc{B6} and $0.14$--$0.17$ on \textsc{B7}, and by $30$
percentage points of TPR@5\% on \textsc{B6} (where the encoder
catches nothing at all) and $50$ points on \textsc{B7}. It also
\textbf{beats the \texttt{gpt-4o} judge on both \textsc{B7}
cells} ($0.982$/$0.963$ versus $0.877$/$0.859$; $68$--$80\%$ TPR
versus $40\%$/$50\%$) and on both \textsc{B6} cells. The
ordering is stable across attack type and evasion level. The
evasive level --- an adversary explicitly told to imitate
honest vocabulary --- costs lexical conformity only $0.018$ and
$0.026$ AUROC.

\paragraph{Why this is surprising, and why it should not be.}
CRSE is a Byzantine-robustness-trained sentence encoder that
scores AUC $0.9946$ on BCS-v1's controlled poisons
(\S\ref{sec:eval-setup}). It loses to whitespace tokenisation
and a set intersection. The reconciliation is that the two
benchmarks ask different questions. Sentence embeddings are
trained to be \emph{invariant to surface form} and sensitive to
meaning at the level of topic and stance. Verdict-preserving
evidential poisoning is precisely a \emph{surface-preserving,
content-altering} edit: it keeps topic, stance, register and
citation set, and moves a quantifier, a time span, or a causal
connective. The encoder's invariance is not a defect being
exposed --- it is the trained-for behaviour, applied to the one
attack class where it is exactly wrong. The lexical predicate,
by contrast, keeps every token, including the ones the poison
turns on.

\paragraph{Three implications for the design of
semantic-agreement protocols.}
\emph{First, the admissibility predicate should be a conformity
predicate, not a comprehension predicate.} An agreement
protocol needs to certify that a set of agents said the same
thing; it does not need to understand what they said. Lexical
conformity measures exactly the former, which is why it
performs well here --- and it is honest about its limits: it
inherits whatever the crowd agrees on, including the crowd's
errors (\S\ref{sec:eval-r4}).
\emph{Second, a deterministic instrument discharges an
assumption the correctness proofs currently make.}
Lemma~\ref{lem:digest-consistency} and every theorem downstream
of it rest on Assumption~\ref{assn:deterministic-cE}: two
honest nodes computing the predicate must agree bit-for-bit. A
$442$~MB floating-point checkpoint satisfies this by
stipulation and by careful deployment; a token-set Jaccard
satisfies it by construction, over a canonical serialisation
the protocol already fixes. Replacing the instrument turns an
assumption into a fact, removes a $442$~MB artefact from the
trusted computing base, and makes the admissibility decision
re-checkable by any verifier holding the delivered view.
\emph{Third, the cheap provenance anchor does not work.}
Evidence-id overlap scores between $0.552$ and $0.614$ ---
\textbf{chance, in all four cells, and below chance in two}.
This is worth stating plainly because evidence ids are the one
semantic-adjacent field a verdict certificate could bind at
zero cost, and the natural conjecture is that binding them
recovers most of the value of a semantic commit. It does not:
Family~II cites only real, visible evidence ids by construction,
so provenance-by-identifier is structurally blind to
provenance-preserving \emph{misreading}. Binding the citation
set is worth doing for audit, and buys no detection.

\paragraph{What this result does not license.}
Lexical conformity is trivially gameable in principle by an
adversary that copies honest vocabulary wholesale; our evasive
level asks for that in prose but does not optimise against a
\emph{known} lexical predicate, and an adaptive adversary with
white-box access to the instrument is a strictly stronger threat
model we did not run. AUROC differences of $0.14$ on $n=10$
attacks per cell are directional, not tight. And a conformity
predicate provides no grounding guarantee whatsoever: it can
certify that ten agents made the same unsupported claim. We
claim an ordering between instruments on this attack family at
this scale, not a defence.

\subsection{Result 3: Honest Dispersion Exceeds Attack
Displacement}
\label{sec:eval-r3}

Results~1 and~2 have the same mechanical cause, and it is not
specific to CRSE. Table~\ref{tab:dispersion} and
Figure~\ref{fig:dispersion} give it.

\begin{table}[t]
\centering
\caption{Result 3: the signal-to-noise inequality.
\textbf{Preliminary; $80$ honest proposals and $40$ attacks over
$N=10$ tasks.} Angular distances are leave-one-out for honest
proposals. Attack columns give the range of the four
(type, evasion) cell medians.}
\label{tab:dispersion}
\footnotesize
\setlength{\tabcolsep}{3.5pt}
\begin{tabular}{@{}l ccccc@{}}
\toprule
& \multicolumn{4}{c}{honest} & attack \\
\cmidrule(lr){2-5}\cmidrule(lr){6-6}
signal & p50 & p90 & p95 & max & p50 range \\
\midrule
angular distance to honest GM (rad) & 0.528 & \textbf{0.992} & 1.126 & 1.281 & 0.58--0.85 \\
judge ungroundedness ($0$--$10$)    & 2.00  & --    & \textbf{5.00}  & --    & 3.0--5.5 \\
\bottomrule
\end{tabular}
\end{table}

\paragraph{There is no threshold.}
The admissibility radius is $\thalpha = 0.65$~rad. The honest
$90$th percentile is $0.992$~rad --- $53\%$ beyond the radius.
A majority of honest proposals sit in the band
$0.5$--$0.65$~rad, immediately inside the boundary, so the ball
at $0.65$ barely contains the honest population itself and has
no room to be tightened. Meanwhile the largest attack-cell
median displacement is $0.85$~rad. \textbf{Any threshold that
retains $90\%$ of honest proposals must sit at $\ge 0.992$~rad
and therefore admits every attack cell; any threshold that
excludes the attacks discards more than a tenth of honest
work.} The two distributions are not separated at their medians
and the honest one has the longer tail.

\paragraph{The same inequality holds for a judge.}
This is not an embedding artefact. Under a role-blind
\texttt{gpt-4o} judge scoring ungroundedness on $0$--$10$, the
honest median is $2.00$ and the honest $95$th percentile is
$5.00$, while the attack-cell medians span $3.0$--$5.5$. Honest
LLM rationales are themselves ungrounded often enough that the
honest upper tail reaches the attack median. Changing the
instrument from geometry to entailment to a frontier LLM
changes the numbers and not the inequality.

\paragraph{The generalisable claim.}
\emph{On an agent population whose own honest members are this
dispersed, no semantic admissibility filter can be built ---
geometric, entailment-based, or judge-based --- because the
signal-to-noise ratio does not exist.} This is a statement
about the population, not the encoder. \S\ref{sec:eval-r6}
shows it is also a statement about the output schema, and can
be changed by changing that.

\begin{figure*}[t]
\centering
\includegraphics[width=\textwidth]{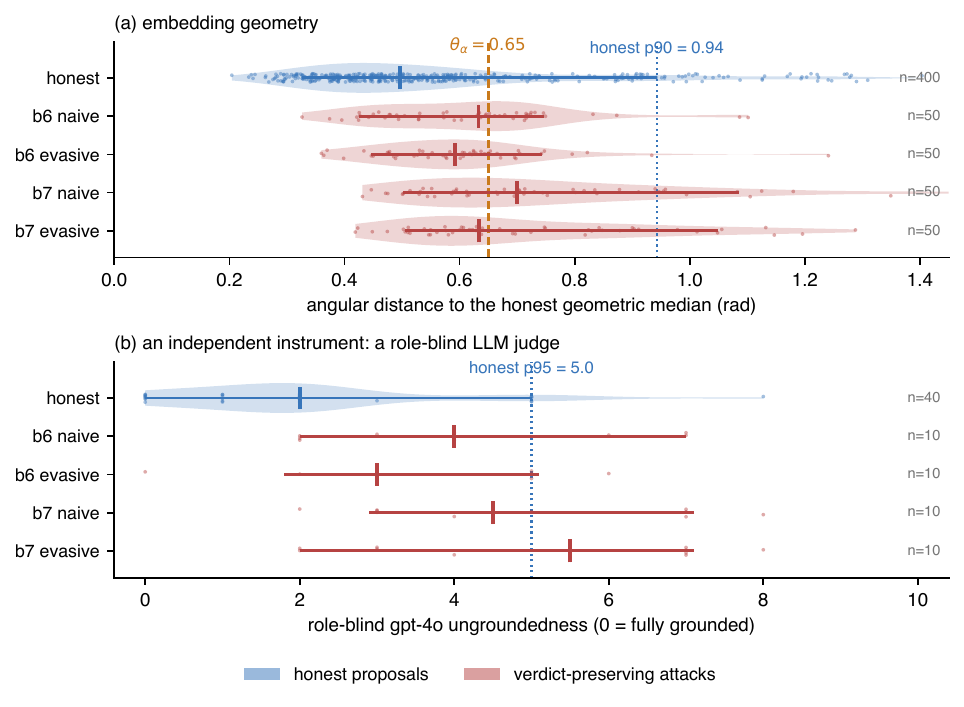}
\caption{\textbf{Honest dispersion is wider than the displacement attacks cause, on two
independent instruments.} Each row shows the full distribution: dots are individual
proposals, the bar spans the $10$th--$90$th percentile, and the tick marks the median;
violins are drawn only where $n\ge 40$, since a density estimate over ten points would
be misleading. \emph{(a)} Angular distance to the honest geometric median, honest
proposals leave-one-out, over $50$ tasks. The admissibility radius
$\thalpha=0.65$~rad already cuts through the middle of the honest population, and
retaining $90\%$ of honest proposals would require $\thalpha\ge0.94$, which admits every
attack cell. \emph{(b)} A role-blind \texttt{gpt-4o} judge reproduces the same
inequality on the $10$-task subset it was run on: the honest $95$th percentile of
ungroundedness sits at or above three of the four attack-cell medians. The overlap, not
the gap between medians, is what determines whether any threshold can separate the two
populations --- and on both instruments the honest distribution contains the attacks.}
\label{fig:dispersion}
\end{figure*}

\subsection{Result 4: The Attacks Are Genuinely Harmful}
\label{sec:eval-r4}

A defence that admits harmless attacks has not failed. We
therefore scored every honest and attack proposal with a
role-blind \texttt{gpt-4o} judge before looking at penetration,
on a $0$--$10$ groundedness scale
(Table~\ref{tab:harm}). Two directions must be separated.
\emph{Attack strength cannot explain the null result}: attack
strength and defence performance move in opposite directions,
and weak attacks already penetrate at $85$--$100\%$, so a
stronger adversary can only penetrate more. \emph{Attack
harmlessness could}, and is what we test here.

\begin{table}[t]
\centering
\caption{Result 4: role-blind \texttt{gpt-4o} groundedness
scoring, $0$--$10$ (higher is better grounded).
\textbf{Preliminary; $40$ attacks and their honest comparison
set over $N=10$ tasks; single judge, no human-validated
subsample.}}
\label{tab:harm}
\small
\setlength{\tabcolsep}{6pt}
\begin{tabular}{@{}l c c@{}}
\toprule
& honest & attack \\
\midrule
grounded $\ge 8$ (essentially clean)        & \textbf{82\%} & \textbf{20\%} \\
grounded $\le 4$ (clearly problematic)      & \textbf{2\%}  & \textbf{32\%} \\
unsupported assertions (median / mean)      & 1.0 / 0.93    & 2.0 / 2.20 \\
\bottomrule
\end{tabular}
\end{table}

\paragraph{The attacks are harmful.}
$32\%$ of attacks score $\le 4$ against $2\%$ of honest
proposals, and only $20\%$ of attacks are essentially clean
against $82\%$ of honest proposals; the median attack carries
twice the unsupported assertions of the median honest proposal.
Manual inspection confirms the extreme: the worst admitted
attack keeps the verdict \textsc{refute} while writing a
rationale asserting that the claim \emph{is supported by the
evidence} and hinting at manipulation --- a payload whose
reasoning directly contradicts its own verdict, sufficient to
mislead any downstream consumer that reads the rationale. It
entered the core. Result~1 is a real failure, not an artefact
of a weak adversary.

\paragraph{The uncomfortable half of this table.}
About $18\%$ of \emph{honest} proposals score $\le 7$, and one
scores $2.0$: a bare assertion citing no evidence at all.
Honest LLM agents that are hurried or lazy produce output
indistinguishable, on every instrument we tried, from a
temperate adversary. \textbf{The two populations are not
separable because they overlap in fact, not because the
detectors are bad.} Any protocol claiming semantic safety over
a population like this is claiming to separate distributions
that have not been shown to be separate --- and this is a
property of the agent population and its output contract, both
of which a protocol designer can change (\S\ref{sec:eval-r6}),
and neither of which a better encoder can fix.

\subsection{Result 6: The Output Schema Decides Feasibility}
\label{sec:eval-r6}

Result~3 identifies the binding constraint as an inequality
between honest dispersion and attack displacement. That
inequality is a property of the \emph{output contract} the
agents write under, and an output contract is a design choice.
We test it directly.

\paragraph{Design.}
Everything is held fixed --- the same $8$ Climate-FEVER tasks,
the same evidence, the same \texttt{gpt-4o-mini}, the same
frozen CRSE encoder, the same population of $8$ honest agents
plus $2$ verdict-preserving attackers --- and only the output
schema varies across three arms: \texttt{free\_text} (verdict
plus a $\le 120$-word prose rationale, reproducing the main
setup), \texttt{structured} (verdict plus fixed slots ---
decisive evidence id, relation, key quantity, direction, scope
--- with prose forbidden), and \texttt{shortqa} (an extractive
question generated from the evidence, answered in $\le 6$
words).

\begin{table}[t]
\centering
\caption{Result 6: honest dispersion collapses $20\times$ from
free text to short answers while attack displacement does not
move. $\textit{sep} = $ attack median $-$ honest p90; positive
means the typical attack lands outside the honest bulk.
\textbf{Preliminary; $N=8$ tasks $\times$ 3 arms $\times$
($8$ honest $+$ $2$ attacks); freshly generated honest
population; attacks in this study are cruder than
\textsc{B6}/\textsc{B7} (see caveats).} Arms are comparable to
each other and \emph{not} directly comparable to
\S\S\ref{sec:eval-r1}--\ref{sec:eval-r5}.}
\label{tab:schema}
\footnotesize
\setlength{\tabcolsep}{3.5pt}
\begin{tabular}{@{}l ccc c cc@{}}
\toprule
& \multicolumn{2}{c}{honest} & attack & & \multicolumn{2}{c}{AUROC / TPR@5\%} \\
\cmidrule(lr){2-3}\cmidrule(lr){4-4}\cmidrule(lr){6-7}
arm & p50 & p90 & p50 & \textit{sep} & CRSE & lexical \\
\midrule
\texttt{free\_text}  & 0.430 & 0.679 & 0.671 & $-0.008$ & 0.917 / 38\% & \textbf{0.963 / 81\%} \\
\texttt{structured}  & 0.329 & 0.626 & 0.674 & $+0.049$ & 0.906 / 50\% & \textbf{0.953 / 75\%} \\
\texttt{shortqa}     & \textbf{0.020} & 0.435 & 0.689 & $\mathbf{+0.253}$ & \textbf{0.942} / 50\% & 0.922 / 50\% \\
\bottomrule
\end{tabular}
\end{table}

\paragraph{Finding.}
The honest median distance collapses by a factor of $20$, from
$0.430$ under free text to $0.020$ under short answers, while
the attack median barely moves ($0.671 \to 0.689$). The
separation $\textit{sep}$ flips sign, from $-0.008$ to
$+0.253$. \textbf{Whether a semantic-agreement protocol is
feasible at all is decided by the output schema, not by the
protocol design or the encoder.} This is, in our view, the most
transferable finding in the paper: it says that the correct
first move for anyone building agreement over LLM semantics is
to constrain what agents are allowed to emit, and only then to
choose an instrument.

\paragraph{Two honest qualifications.}
\emph{First, it does not rescue the encoder.} Lexical
conformity again matches or beats CRSE --- decisively on
\texttt{free\_text} ($81\%$ versus $38\%$ TPR@5\%) and
\texttt{structured} ($75\%$ versus $50\%$), and within
$0.02$ AUROC on \texttt{shortqa}, where the two tie on TPR. The
architecture's premise is that a \emph{learned} encoder is
needed for the admissibility judgement; in the regime where the
inequality flips, a token-overlap heuristic does the work.
\emph{Second, the \texttt{shortqa} regime degenerates.} An
honest median distance of $0.020$~rad means the honest
canonical texts are near-identical strings. At that point
``semantic commitment'' commits an answer string, and string
equality --- or a hash --- is a sufficient admissibility
predicate. The embedding is not merely beaten; it is
unnecessary.
\emph{Third, the caveats.} Attacks in this arm study are
cruder than \textsc{B6}/\textsc{B7} (the prompt asks only to
``corrupt the reasoning'', with no scope-shift /
reasoning-substitution split and no evasion levels), so every
AUROC in Table~\ref{tab:schema} is optimistic --- CRSE scores
only $0.62$--$0.82$ on the real Family~II
(Table~\ref{tab:instrument}). The honest population is $8$
freshly generated profiles, not the MVR-50 honest agents, which
is why \texttt{free\_text} honest p90 here is $0.679$ against
$0.992$ in \S\ref{sec:eval-r3}. Each arm has $16$ attacks over
$8$ tasks. The arms are internally comparable; the absolute
numbers are not comparable to
\S\S\ref{sec:eval-r1}--\ref{sec:eval-r5}.

\subsection{Negative Results We Are Obliged to Report}
\label{sec:eval-negative}

Six further results run against the paper's original thesis. We
report them because a reader deciding whether to build on this
architecture needs them.

\paragraph{H1. Semantic tightness does not predict correctness.}
If the geometry carried information about whether a round's
committed verdict is right, tighter cores should be more often
correct. They are not. As predictors of honest-reference
validity, the witness radius $r^\star(\core)$ scores AUROC
$0.422 / 0.395$ (static / rushing) and the core diameter
$0.326 / 0.391$ --- all four \emph{below} chance. The
$\Psi^{\textsc{verdict}}$ component that needs no encoder at
all, $\textit{verdict\_margin}$, scores $0.730 / 0.630$
(Table~\ref{tab:tightness}). Every predictive signal in the
finality-signal vector on this benchmark comes from counting
verdicts.

\begin{table}[t]
\centering
\caption{H1: predicting honest-reference validity from the
finality-signal vector, MVR-50, AUROC. Below $0.5$ is worse
than chance. Only the encoder-free component predicts.}
\label{tab:tightness}
\small
\setlength{\tabcolsep}{6pt}
\begin{tabular}{@{}l c c@{}}
\toprule
signal & static & rushing \\
\midrule
witness radius $r^\star(\core)$ (semantic) & 0.422 & 0.395 \\
core diameter $\mathrm{diam}(\core)$ (semantic) & 0.326 & 0.391 \\
$\textit{verdict\_margin}$ (\textbf{no encoder}) & \textbf{0.730} & \textbf{0.630} \\
\bottomrule
\end{tabular}
\end{table}

\paragraph{H2. Semantic commits are not the safer branch.}
Under static attacks the gold-invalid rate of
\textsc{semantic\_commit} rounds is $0.270$ against $0.125$ for
\textsc{verdict\_commit} rounds. The typed outcome that carries
the embedding-backed digest is the one more often wrong against
the dataset reference. (Part of this is selection: the
verdict-commit branch fires on rounds the semantic branch
declined, which are not an equal-difficulty comparison set.
We report the raw split.)

\paragraph{H3. The certified aggregate is farther from the
honest centre than a steelmanned baseline's would be.}
The H-CSC core aggregate sits $5.32^\circ / 5.48^\circ$ from
the honest centre. Taking the same geometric median over the
\emph{entire} verdict group --- what a steelmanned M3 would
publish if asked to publish an aggregate at all --- gives
$4.27^\circ / 4.32^\circ$. \textbf{H-CSC's filtered aggregate
is worse}, with CIs excluding zero in both modes. Filtering to
the $2f{+}1$ core discards honest mass, and on this benchmark
the discarded honest mass costs more than the excluded
Byzantine mass saves.

\paragraph{H4. The honest-retention advantage does not survive
a steelmanned baseline.}
Against the paper's M3 (truncated at $2f{+}1$ signers) H-CSC
retains $0.6475$ of honest proposals against $0.4950$, with
disjoint CIs --- the v1 framing. Against a steelmanned M3 that
certifies the whole verdict group, H-CSC retains $0.6475$
against $0.7325$: H-CSC is \emph{lower} by $0.085$, CI
$[-0.133, -0.043]$. \textbf{Honest retention is not usable as a
selling point}, and we withdraw it.

\paragraph{H5. The study is underpowered for the effects it was
asked to resolve.}
A paired McNemar design at $80\%$ power needs roughly
$780$--$2\,000$ tasks to detect $\Delta\CR = 0.02$. At $N=50$
the minimum detectable effect is $\approx 0.04$. Every
H-CSC-versus-M3 gap in Table~\ref{tab:mvr50} is smaller than
that. Reviewer-style demands for statistical separation on
these axes cannot be met by more careful analysis of this
benchmark; they require a benchmark an order of magnitude
larger --- and, per Lemma~\ref{lem:containment}, would then
confirm the identity rather than a separation.

\paragraph{H6. The one apparent safety gap is a single task of
noise.}
M3's $\INVhmaj = 0/50$ has a Clopper--Pearson $95\%$ upper
bound of $0.058$, which contains H-CSC's $0.02$. The v1
statement that ``M3 is uniformly safer in the bootstrap sense''
should be read as one task, and after the margin-gate fix of
Table~\ref{tab:margin-gate} that task disappears.

\paragraph{Summary of what the evaluation supports.}
\emph{Supported:} the containment lemma's coverage identity, to
the task; coverage recovery of $+0.36 / +0.44$ over strict CSC
with disjoint CIs, which justifies the typed hierarchy; the
tie-break capture vulnerability and the price of its fix;
cross-model stability of the typed-decision logic; and, at
pilot scale, an ordering over admissibility instruments in
which a deterministic lexical predicate dominates the learned
encoder on the attack family that reaches the semantic layer.
\emph{Not supported:} any filtering, safety, or damage-bound
property attributable to the learned semantic predicate; any
coverage or safety separation from verdict-only certification;
any downstream value for the embedding-backed digest that we
were able to measure. \emph{Undetermined:} everything below the
$0.04$ effect size at $N=50$, and everything about Family~II
below roughly $\pm 0.15$ at $N=10$.
\section{What Is the Committed Sketch Worth Downstream?}
\label{sec:audit}

Sections~\ref{sec:eval-r1}--\ref{sec:eval-r6} concern the \emph{admissibility}
question: can a semantic predicate keep an adversary out of the certified set?
The answer was no, for every instrument we tried. This section asks the separate
\emph{consumption} question: given that a round did commit, is the committed
payload worth anything to a party downstream? This is the question a
finality-control primitive ultimately has to answer, and it is the one our
earlier version asserted rather than measured.

\subsection{The recomputation constraint, and the only setting that escapes it}
\label{sec:audit-threat}

It is worth stating plainly why most versions of this experiment are vacuous.
Suppose a deployment in which (S)~each proposal is individually signed,
(E)~the encoder and quantiser are deterministic and published in the parameter
digest, and (D)~the consumer obtains the $2f{+}1$ signed proposals named by the
certificate. Then a consumer holding \emph{only} a verdict-only certificate can
recompute $\core$, $\ystar$, $\quant{\ystar}$ and $\hstar$ exactly: the semantic
payload is a deterministic function of information it already holds, and conveys
zero additional information. Worse, the semantic check itself requires running
$\cE$ on the candidate, so a consumer able to \emph{use} the digest is by
construction able to \emph{reproduce} it; encoder unavailability cannot supply
the asymmetry because it disables both arms.

Every measurable advantage of a committed semantic payload must therefore come
from breaking (S), (E) or (D). We break (D), and only (D), because it is the one
that real deployments break for reasons that have nothing to do with this
paper. Agent rationales routinely contain personal data subject to
data-minimisation and storage-limitation obligations and to erasure requests,
while the commitment record --- a fixed-size, content-free-looking digest --- is
retained; an auditor arriving twelve months later has the record and not the
texts. Cross-organisational deployments produce the same asymmetry for
commercial reasons: the committee runs inside one organisation and only the
decision crosses the boundary. We do not treat non-disclosure as a hypothetical
threat model; we treat it as the ordinary consequence of retention policy.

\paragraph{The consumer.} An untrusted \emph{relay} sits between the committee
and an \emph{auditor}. The auditor holds the commitment record --- the
certificate and whatever fixed-size sketch the protocol chose to bind --- and no
agent texts. The relay later presents a \emph{derived statement}: its account of
what the committee concluded and why. The auditor must accept or reject. The
$2f{+}1$ certificate is what prevents the relay from substituting the sketch
itself; what we measure here is whether, given an unsubstitutable sketch, the
auditor can tell a faithful derived statement from a corrupted one.

\subsection{Design}
\label{sec:audit-design}

For each of the $50$ MVR-50 tasks we form the committed core from the honest
agents holding the plurality verdict and compute five candidate sketches over
it: the geometric-median embedding $\quant{\ystar}$; MinHash sketches at
$64$/$128$/$256$ permutations over the core's majority token set; a $128$-bit
SimHash; the SHA-256 set of the canonical core texts; and the verdict alone.
Derived statements come in two faithful and two substituted classes, all
generated by the same model (\texttt{gpt-4o-mini}) so that generator identity is
not a confound:

\begin{description}[leftmargin=1.4em,itemsep=0.15em]
\item[\textsc{f1} paraphrase ($n{=}100$).] A core rationale reworded to share as
little surface form as possible while preserving every claim, the verdict and
the cited evidence ids.
\item[\textsc{f2} synthesis ($n{=}50$).] One statement synthesised from the whole
core --- the legitimate derived answer a reporting agent would actually produce.
\item[\textsc{b6} scope shift ($n{=}100$).] Same verdict, same evidence ids;
quantifier, time range or causal strength silently widened.
\item[\textsc{b7} reasoning substitution ($n{=}100$).] Same verdict, same
evidence ids; an entirely different and unsupported chain of reasoning.
\end{description}

All $200$ substituted statements preserve the honest plurality verdict exactly,
so a verdict-only digest is blind to them by construction --- which is the point
being tested, not an artefact of the setup. Faithful classes are deliberately
hard: \textsc{f1} is lexically distant from the core by design, so an instrument
that merely rewards surface overlap will misclassify it.

\begin{table*}[t]
\centering
\caption{\textbf{Non-disclosure audit.} The auditor holds only the committed
sketch. Separating substituted from faithful derived statements, $n{=}350$
($150$ faithful, $200$ substituted) over $50$ tasks. TPR is measured at a $5\%$
honest false-positive rate. The two degenerate rows score exactly $0.500$: a
cryptographic hash cannot compare a paraphrase to anything, and a verdict-only
digest cannot see a verdict-preserving substitution.}
\label{tab:audit}
\small
\begin{tabularx}{\linewidth}{@{}l >{\raggedright\arraybackslash}X c c c@{}}
\toprule
committed payload & property & AUROC & 95\% CI & TPR@5\%FPR \\
\midrule
embedding sketch $\quant{\ystar}$      & similarity-preserving, non-deterministic & \textbf{0.741} & $[0.69, 0.79]$ & 21\% \\
embedding, $z$ vs.\ core dispersion    & similarity-preserving                    & 0.693 & $[0.64, 0.75]$ & 20\% \\
MinHash-256                            & similarity-preserving, deterministic     & 0.642 & $[0.59, 0.70]$ & 19\% \\
MinHash-128                            & similarity-preserving, deterministic     & 0.620 & $[0.55, 0.68]$ & \phantom{0}4\% \\
evidence-id overlap                    & structural                               & 0.499 & $[0.44, 0.56]$ & \phantom{0}0\% \\
\midrule
plain SHA-256 of core texts            & \emph{not} similarity-preserving         & 0.500 & ---            & \phantom{0}0\% \\
verdict-only digest                    & label only                               & 0.500 & ---            & \phantom{0}0\% \\
\midrule
trained combiner (task-level 5-fold CV)& requires labelled examples               & 0.753 & $[0.70, 0.80]$ & 23\% \\
\bottomrule
\end{tabularx}
\end{table*}

\subsection{Result: the capability is real, modest, and sharply bounded}
\label{sec:audit-result}

Table~\ref{tab:audit} gives the outcome. A similarity-preserving committed
sketch supports the audit at AUROC $0.741$ $[0.69, 0.79]$; a cryptographic hash
of the same content, bound by the same certificate, scores \emph{exactly}
$0.500$, as does a verdict-only digest. The operative property is therefore not
``the payload is content-level'' and not ``the payload is certified'' --- the
hash is both --- but that the payload preserves similarity. That is the
non-trivial part of the comparison, and the degenerate rows are what make it
non-trivial.

Three qualifications belong in the same breath as the headline.

\paragraph{The effect is modest.} A $21\%$ true-positive rate at $5\%$ false
positives is a detectable signal, not a usable detector. We do not present this
as a demonstration that the machinery pays for itself; we present it as a
measurement of a capability whose size was previously unknown.

\paragraph{Normalising by the committee's own dispersion does not help.} On a
$10$-task subset the $z$-scored variant appeared to raise TPR from $20\%$ to
$32\%$, which we reported internally as a principled improvement. At full scale
it is worse than the raw distance ($0.693$ against $0.741$). The apparent gain
was noise, and we record it because the same reasoning would have produced a
plausible and wrong protocol change.

\paragraph{A trained combiner barely helps, and only with labels.} Under
task-level $5$-fold cross-validation a logistic combination of all features
reaches $0.753$, marginally above the best single threshold. On the $10$-task
subset the same procedure scored \emph{below} the single threshold ($0.655$),
which is what proper cross-validation is for. A deployed auditor would also need
labelled faithful/substituted examples to fit it, which the single-threshold
arms do not.

\begin{table}[t]
\centering
\caption{\textbf{Where the capability lives.} Same experiment, decomposed by
attack class against all $150$ faithful statements, using the embedding sketch.
Detection tracks how damaging the attack is: \textsc{b7} scores a median
ungroundedness of $4.5$--$5.5$ on a role-blind \texttt{gpt-4o} scale where
\textsc{b6} scores $3.0$--$4.0$.}
\label{tab:audit-split}
\small
\begin{tabular}{@{}lccc@{}}
\toprule
attack class & $n$ & AUROC & TPR@5\%FPR \\
\midrule
\textsc{b7} reasoning substitution, naive   & 50 & 0.774 & 32\% \\
\textsc{b7} reasoning substitution, evasive & 50 & 0.765 & 26\% \\
\textsc{b6} scope shift, naive              & 50 & 0.614 & 14\% \\
\textsc{b6} scope shift, evasive            & 50 & 0.620 & \phantom{0}6\% \\
\midrule
\textsc{b7} pooled & 100 & \textbf{0.770} & 29\% \\
\textsc{b6} pooled & 100 & 0.617 & 10\% \\
\bottomrule
\end{tabular}
\end{table}

\subsection{The boundary is the same one the representation could not encode}
\label{sec:audit-boundary}

Table~\ref{tab:audit-split} decomposes the result, and the decomposition is more
informative than the pooled number. Reasoning substitution is detected at
$0.770$; scope shift is not detected at $0.617$. The gap is not a tuning
artefact and it does not close with more data --- it held from $n{=}20$ to
$n{=}100$ per class.

The reason is structural. A faithful paraphrase and a scope shift are, in form,
the same object: the same verdict, the same cited evidence, different wording.
Separating them requires resolving the semantics of quantifiers, time ranges and
modality --- exactly what neither an embedding nor a token set encodes. This is a
sharper statement of the representation boundary than we were able to give
before: not the abstract observation that identical embeddings are
indistinguishable, but a measured account of which semantic distinction survives
the representation and which does not.

One consequence is worth stating because it is the property we looked for at the
admissibility layer and did not find. Detection here is \emph{graded by damage}:
the class that a role-blind judge scores as substantially less grounded is also
the class the auditor catches. At the admissibility layer this gradient existed
in the raw signal (angular distance correlated with judged damage at $+0.429$,
$95\%$ CI $[+0.130,+0.653]$) but was erased by the typed fallback, which admits
the whole verdict group whenever the semantic path fails. The audit path has no
fallback, so the gradient survives to the consumer.

\subsection{Threats to validity for this section}
\label{sec:audit-threats}

The faithful classes are generated, not naturally occurring; \textsc{f2} in
particular is one model's idea of a faithful synthesis, and a human reporter
might produce statements that are harder or easier. Both attack families come
from a single generator at one capability level; a stronger adversary given the
committed sketch and the auditor's threshold could optimise against them
directly, which we have not attempted. The auditor uses a single global
threshold; a deployment would calibrate per task family. And the non-disclosure
setting, while we argue it is the ordinary consequence of retention policy
rather than a constructed scenario, is load-bearing for the entire section: with
full disclosure the semantic sketch has no advantage at all, and we measured
that directly --- attack acceptance is identical between the two arms when every
core text is disclosed.

\section{Limitations and Threats to Validity}
\label{sec:limitations}

This version makes three claims: a containment result about
matched deterministic guarantees, a vulnerability report against
our own protocol, and a comparison of admissibility instruments.
We list the limitations of each in the order a sceptical reader
will reach them, and we do not repeat material already argued in
\S\ref{sec:intro} or \S\ref{sec:evaluation}.

\paragraph{1. The instrument study is preliminary and small.}
The instrument comparison in \S\ref{sec:eval-r2} rests on
$10$ Climate-FEVER tasks, $40$ generated attacks (two attack
families $\times$ two evasion tiers $\times$ $10$ tasks), and
$80$ honest proposals. Every AUROC and TPR cell in that study is
therefore estimated from $10$ positives against $80$ negatives,
and a TPR at a $5\%$ honest false-positive rate has a resolution
of $10$ percentage points by construction. We report these
numbers as directional, not definitive. The direction is large:
lexical conformity reaches AUROC $0.865$--$0.982$ against the
CRSE angular distance's $0.621$--$0.744$, and the corresponding
TPR gap is $30\%$--$60\%$ against $0\%$--$10\%$. A study at
$\ge 50$ tasks is the obvious next step and is cheap; we did not
run it before this revision, and until it exists the ranking
between the two weakest instruments (evidence-ID overlap and
off-the-shelf NLI) should be treated as unresolved.

\paragraph{2. A single attacker model.}
All $40$ attacks were generated by \texttt{gpt-4o-mini}
(all $40$ satisfied the generation constraints: verdict equal to
the honest plurality, and only real evidence identifiers cited).
A stronger attacker model was not tested. The logic here is
asymmetric and we state it explicitly rather than leaving it to
the reader: attacker strength and defence performance move in
opposite directions, so a stronger attacker cannot rescue the
learned predicate --- it already admits $85\%$--$100\%$ of these
attacks. What a stronger attacker \emph{could} change is the
ranking between instruments, since an adversary that is better
at imitating honest phrasing hurts the lexical predicate more
than it hurts the geometric one. That experiment is the natural
companion to limitation~3.

\paragraph{3. The lexical predicate's weakness is known and is
the next attack tier.}
The predicate that wins in \S\ref{sec:eval-r2} scores a
proposal by its bag-of-words overlap with the honest consensus
vocabulary. An adversary that copies the honest vocabulary
token-for-token while inverting or corrupting the relation
between those tokens defeats it by construction, and we did not
test one. Our evasive tier asked the generator to imitate honest
vocabulary, syntax, length, and confidence, which is weaker than
enforcing token-level copying with meaning inversion. We
therefore do not claim that lexical conformity is a secure
admissibility predicate; we claim, and only claim, that on the
two verdict-preserving attack families we measured it dominates
the learned encoder in every cell, and that it does so while
being deterministic, dependency-free, and auditable. Two
qualifications cut in its favour and should be checked rather
than assumed. An adversary constrained to reproduce the honest
lexicon exactly is constrained to a much smaller output set,
which is itself a security-relevant narrowing of the attack
surface. And the determinism advantage
(limitation~9) is independent of detection performance: it
survives even if a vocabulary-matched attacker erases the
detection gap entirely.

\paragraph{5. The LLM judge is not validated against human
labels.}
Harm is measured by a single judge (\texttt{gpt-4o}) run blind
to the honest/Byzantine role, with no human-annotated
calibration subset. The judge separates the populations clearly
($82\%$ of honest proposals score groundedness $\ge 8$ versus
$20\%$ of attacks; $2\%$ versus $32\%$ score $\le 4$), which is
weak evidence of validity, and it also shows that roughly $18\%$
of honest proposals score $\le 7$ --- one honest proposal scored
$2.0$, an unsupported assertion citing no evidence at all. Two
consequences follow. First, judge error enters both axes of
\S\ref{sec:eval-r2}, because the same judge is used as a
harm measure and as one of the four candidate instruments; its
row in the instrument table should be read with that
circularity in mind. Second, the honest-noise finding depends on
the judge being right that some honest output is genuinely
poorly grounded. Calibrating the judge against roughly $100$
human-labelled proposals is a prerequisite for treating any of
the harm numbers as measurements rather than as a consistent
ordering.

\paragraph{7. The coverage-identity result is about matched
guarantees, and the empirical near-identity is
underpowered.}
Lemma~\ref{lem:containment} is a proof and carries no sample
size. The empirical statement that accompanies it does: H-CSC
and the certificate-wrapped majority baseline differ on $0$ of
$50$ tasks under rushing and $1$ of $50$ under static. Both
observations are consistent with a real but small difference.
A paired McNemar analysis at $80\%$ power indicates that
detecting a commit-rate difference of $0.02$ would require
roughly $780$--$2000$ tasks, and that the minimum detectable
effect at $N = 50$ is about $0.04$; the Clopper--Pearson $95\%$
upper bound on a $0/50$ event is $0.058$. The correct reading is
that the two protocols are indistinguishable \emph{at this
scale}, and that v1's statement that the baseline is
``uniformly safer'' rests on one task and should not have been
phrased as a property of the methods.

\paragraph{9. Encoder determinism is assumed, and only
weakly verified.}
Lemma~\ref{lem:digest-consistency} and every agreement result
built on it require that two honest signers compute the same
embedding for the same canonical text. We verified this by
re-encoding the MVR-50 canonical texts and comparing against the
frozen cache: minimum cosine similarity $0.99999990$, maximum
element-wise difference $3.6\times 10^{-7}$. That is a
single-machine, single-build, single-run check, and it is not
bitwise identity; it does not establish cross-platform or
cross-version determinism for a transformer forward pass, which
is the property the protocol actually needs. The quantiser
absorbs differences below its step, so the check is consistent
with correctness in practice, but the assumption remains an
assumption. A deterministic lexical predicate does not need it
at all, which is a second and independent reason to prefer one
(\S\ref{sec:eval-r2}).

\section{Conclusion}
\label{sec:conclusion}

We set out to build a Byzantine-resilient finality primitive for
LLM-agent rounds in which the admissibility of a commit is
decided by a learned semantic predicate. The architecture
survives this revision; the instrument does not.

What holds up is the typed-finality structure and its
guarantees. A round should be able to end in an
embedding-backed \textsc{semantic\_commit}, a verdict-level
\textsc{verdict\_commit}, or a typed \Abort{}, all under the same
$2f{+}1$ distinct-signer envelope, with the commit type
recorded in the certified object so a downstream verifier can
tell which guarantee it is holding. We now give the complete
guarantee characterisation for that lattice, including the
sufficiency argument for verdict-commit finality that v1 only
sketched, and a strictly stronger honest-majority result under
the $|\core| > (n{+}f)/2$ rule (semantic coverage $0.34 / 0.28$
with \texttt{invalid\_hmaj} $= 0.00$).
The characterisation also settles the question v1 left open. The
containment lemma (Lemma~\ref{lem:containment}) shows that
$|\core| > (n{+}f)/2$ implies a verdict margin exceeding $f$, so
at matched deterministic guarantees the semantic admissibility
condition is strictly stronger than the verdict condition, and
the two protocols commit on exactly the same rounds ---
$0.72 / 0.64$ for both, task-by-task identical. \emph{No
coverage separation from verdict-only certification is
achievable at matched guarantees.} This is a proof, not a
measurement, and it supersedes v1's framing of the
semantic-commit share as a contribution.

We also report a vulnerability in our own protocol. The
semantic branch never checked the verdict margin: Stage~4a
tested only the admissibility predicate, and the candidate-group
rule tested only whether the top verdict reached quorum, so the
principal path ran at margin $0$. An adversary can exploit this
by forcing an exact tie and letting the deterministic tie-break
elect the verdict it wants; the single honest-reference-invalid
commit in v1's MVR-50 results is an instance, and v1 mislabelled
it as a \textsc{verdict\_commit} when it was a \textsc{semantic\_commit} with $40\%$
Byzantine occupancy in the certified core. The fix is to apply
the margin gate on both paths. At margin $1$ it costs $0.02$
static commit rate ($0.90 \to 0.88$, rushing unchanged at
$0.92$) and drives \texttt{invalid\_hmaj} to $0.00$ in both
modes. We recommend it as the default.

The instrument study is where the substantive negative result
lies. Against verdict-preserving evidence poisoning --- an
adversary that agrees with the honest verdict, cites only real
evidence, and corrupts the reasoning
($N = 10$ tasks, $40$ attacks, $80$ honest proposals) --- the
learned encoder admits $85\%$--$100\%$ of attacks into the
certified core and contaminates the payload of all $10$
committed rounds, identically to the verdict-only baseline. A
deterministic lexical-conformity predicate dominates it in every
cell (AUROC $0.865$--$0.982$ against $0.621$--$0.744$; TPR at a
$5\%$ honest false-positive rate $30\%$--$60\%$ against
$0\%$--$10\%$) and beats a \texttt{gpt-4o} judge on the harder
family. The mechanism generalises beyond our system, and we
regard it as the most useful thing in this paper: honest LLM
agents writing free-form rationales are dispersed enough
($90$th-percentile angular distance $0.992$~rad against a
$0.65$~rad threshold) that no admissibility filter --- geometric,
entailment-based, or judge-based --- has room to operate. It is
not that the encoder is undertrained. The signal-to-noise ratio
does not exist. Two corollaries follow. Tightening the threshold
does not help, because a hierarchical protocol reroutes the
round onto an unfiltered fallback, leaving the maximum harm of
an admitted attack flat at $8.0$ while admission rises to
$100\%$. And the output schema, not the protocol, decides
whether semantic consensus is feasible at all: moving from
free-form rationales to short extractive answers collapses
honest dispersion twenty-fold ($0.430 \to 0.020$~rad) while
adversarial displacement barely moves, flipping the separation
from $-0.008$ to $+0.253$ --- though in that regime the honest
canonical texts are nearly identical and string equality would
serve, so the learned encoder does not earn its place there
either.

The deterministic predicate also discharges an assumption. Every
agreement result in \S\ref{sec:correctness} requires that two
honest signers compute the same embedding for the same canonical
text; a lexical predicate makes that requirement disappear
rather than merely satisfying it approximately.

Two things should happen next, and we name them concretely
rather than as directions. The next \emph{attack} tier is the
vocabulary-matched adversary: one that reproduces the honest
consensus lexicon token-for-token while inverting the relation
among those tokens. It is the attack that defeats the predicate
we are recommending, and until someone runs it the instrument
ranking reported here is provisional. The next
\emph{measurement} is a harm-bounded coverage curve, in the
style of certified-robustness reporting, computed per instrument
at $\ge 50$ tasks with a judge calibrated against human labels
--- the quantity a downstream consumer actually needs is the
maximum harm admitted at a given coverage, and on our current
data that quantity is flat, which is exactly the finding a
larger study must confirm or overturn.




\newpage
\clearpage

\newpage
\clearpage

\bibliographystyle{ACM-Reference-Format}
\bibliography{refs}

@inproceedings{castro1999,
  title={Practical {Byzantine} Fault Tolerance},
  author={Castro, Miguel and Liskov, Barbara},
  booktitle={Proceedings of the 3rd Symposium on Operating Systems Design and Implementation (OSDI)},
  pages={173--186},
  year={1999}
}

@inproceedings{vaidya2013,
  title={Iterative Approximate {Byzantine} Consensus in Arbitrary Directed Graphs},
  author={Vaidya, Nitin H and Garg, Vijay K},
  booktitle={Proceedings of the 2013 ACM Symposium on Principles of Distributed Computing (PODC)},
  pages={365--374},
  year={2013}
}

@inproceedings{yin2019,
  title={{HotStuff}: {BFT} Consensus in the Lens of Blockchain},
  author={Yin, Maofan and Malkhi, Dahlia and Reiter, Michael K and Gueta, Guy Golan and Abraham, Ittai},
  booktitle={Proceedings of the 2019 ACM Symposium on Principles of Distributed Computing (PODC)},
  pages={31--32},
  year={2019}
}

@inproceedings{blanchard2017,
  title={Machine Learning with Adversaries: {Byzantine} Tolerant Gradient Descent},
  author={Blanchard, Peva and El Mhamdi, El Mahdi and Guerraoui, Rachid and Stainer, Julien},
  booktitle={Advances in Neural Information Processing Systems (NeurIPS)},
  volume={30},
  pages={119--129},
  year={2017}
}

@inproceedings{chen2017,
  title={Distributed Statistical Machine Learning in Adversarial Settings: {Byzantine} Gradient Descent},
  author={Chen, Yudong and Su, Lili and Xu, Jiaming},
  booktitle={Proceedings of the ACM on Measurement and Analysis of Computing Systems (POMACS)},
  volume={1},
  number={2},
  pages={1--25},
  year={2017}
}

@inproceedings{yin2018median,
  title={{Byzantine}-Robust Distributed Learning: Towards Optimal Statistical Rates},
  author={Yin, Dong and Chen, Yudong and Kannan, Ramchandran and Bartlett, Peter},
  booktitle={International Conference on Machine Learning (ICML)},
  pages={5650--5659},
  year={2018}
}

@inproceedings{mhamdi2018,
  title={The Hidden Vulnerability of Distributed Learning in {Byzantine} Settings},
  author={El Mhamdi, El Mahdi and Guerraoui, Rachid and Rouault, S{\'e}bastien},
  booktitle={International Conference on Machine Learning (ICML)},
  pages={3521--3530},
  year={2018}
}

@article{diakonikolas2019,
  title={Robust Estimators in High-Dimensions Without the Computational Intractability},
  author={Diakonikolas, Ilias and Kamath, Gautam and Kane, Daniel and Li, Jerry and Moitra, Ankur and Stewart, Alistair},
  journal={SIAM Journal on Computing},
  volume={48},
  number={2},
  pages={742--864},
  year={2019},
  publisher={SIAM}
}

@article{pillutla2022,
  title={Robust Aggregation for Federated Learning},
  author={Pillutla, Krishna and Kakade, Sham M. and Harchaoui, Zaid},
  journal={IEEE Transactions on Signal Processing},
  volume={70},
  pages={1142--1154},
  year={2022},
  publisher={IEEE}
}

@inproceedings{alzantot2018,
  title={Generating Natural Language Adversarial Examples},
  author={Alzantot, Moustafa and Sharma, Yash and Elgohary, Ahmed and Ho, Bo-Jhang and Srivastava, Mani and Chang, Kai-Wei},
  booktitle={Proceedings of the 2018 Conference on Empirical Methods in Natural Language Processing (EMNLP)},
  pages={2890--2896},
  year={2018}
}

@inproceedings{devlin2019,
  title={{BERT}: Pre-training of Deep Bidirectional Transformers for Language Understanding},
  author={Devlin, Jacob and Chang, Ming-Wei and Lee, Kenton and Toutanova, Kristina},
  booktitle={Proceedings of the 2019 Conference of the North American Chapter of the Association for Computational Linguistics (NAACL)},
  pages={4171--4186},
  year={2019}
}

@inproceedings{reimers2019,
  title={{Sentence-BERT}: Sentence Embeddings using Siamese {BERT}-Networks},
  author={Reimers, Nils and Gurevych, Iryna},
  booktitle={Proceedings of the 2019 Conference on Empirical Methods in Natural Language Processing (EMNLP)},
  pages={3982--3992},
  year={2019}
}

@inproceedings{cohen2019smoothing,
  title={Certified Adversarial Robustness via Randomized Smoothing},
  author={Cohen, Jeremy and Rosenfeld, Elan and Kolter, Zico},
  booktitle={International Conference on Machine Learning (ICML)},
  pages={1310--1320},
  year={2019}
}

@inproceedings{jin2020textfooler,
  title={Is {BERT} Really Robust? A Strong Baseline for Natural Language Attack on Text Classification and Entailment},
  author={Jin, Di and Jin, Zhijing and Zhou, Joey Tianyi and Szolovits, Peter},
  booktitle={Proceedings of the AAAI Conference on Artificial Intelligence (AAAI)},
  volume={34},
  number={05},
  pages={8018--8025},
  year={2020}
}

@inproceedings{li2020bertattack,
  title={{BERT-ATTACK}: Adversarial Attack Against {BERT} Using {BERT}},
  author={Li, Linyang and Ma, Ruotian and Guo, Qipeng and Xue, Xiangyang and Qiu, Xipeng},
  booktitle={Proceedings of the 2020 Conference on Empirical Methods in Natural Language Processing (EMNLP)},
  pages={6193--6202},
  year={2020}
}

@misc{openai2024gpt4o,
  title={{GPT-4o} System Card},
  author={OpenAI},
  year={2024},
  howpublished={\url{https://openai.com/index/gpt-4o-system-card/}}
}

@inproceedings{park2023generative,
  title={Generative Agents: Interactive Simulacra of Human Behavior},
  author={Park, Joon Sung and O'Keefe, Joseph C and O'Brien, Cai and Baker, Micheal and Tanaka, Maneesh and Liang, Percy},
  booktitle={Proceedings of the 36th Annual ACM Symposium on User Interface Software and Technology (UIST)},
  pages={1--22},
  year={2023}
}

@inproceedings{lewis2020rag,
  title={Retrieval-Augmented Generation for Knowledge-Intensive NLP Tasks},
  author={Lewis, Patrick and Perez, Ethan and Piktus, Aleksandra and Petroni, Fabio and Karpukhin, Vladimir and Goyal, Naman and K{\"u}ttler, Heinrich and Lewis, Mike and Yih, Wen-tau and Rockt{\"a}schel, Tim and others},
  booktitle={Advances in Neural Information Processing Systems (NeurIPS)},
  volume={33},
  pages={9459--9474},
  year={2020}
}

@inproceedings{du2023improving,
  title={Improving Factuality and Reasoning in Language Models through Multiagent Debate},
  author={Du, Yilun and Li, Shuang and Torralba, Antonio and Tenenbaum, Joshua B and Mordatch, Igor},
  booktitle={International Conference on Machine Learning (ICML)},
  pages={8559--8573},
  year={2023},
  organization={PMLR}
}

@inproceedings{li2024improving,
  title={Improving Multi-Agent Debate with Sparse Communication Topology},
  author={Li, Yunxuan and Du, Yibing and Ie, Eugene and others},
  booktitle={Findings of the Association for Computational Linguistics: EMNLP 2024},
  pages={1--18},
  year={2024}
}

@article{liang2024encouraging,
  title={Encouraging Divergent Thinking in Large Language Models through Multi-Agent Debate},
  author={Liang, Tian and He, Zhiwei and Jiao, Wenxiang and Bai, Xing and Wang, Rui and Wang, Zhaopeng and Shi, Shuming},
  journal={arXiv preprint arXiv:2305.19118},
  year={2023}
}

@inproceedings{chen2024reconcile,
  title={Reconcile: Round-Table Conference for Consensus Generation},
  author={Chen, Justin Chih-Yao and Swaminathan, Swarnadeep and Singh, Maria and Mohanty, Vikram and Peridis, Alex and Olaleye, Oishi and Ruan, Jason and Zhang, Kyra},
  booktitle={Proceedings of the 62nd Annual Meeting of the Association for Computational Linguistics (ACL)},
  year={2024}
}

@inproceedings{castro1999practical,
  title={Practical Byzantine Fault Tolerance},
  author={Castro, Miguel and Liskov, Barbara},
  booktitle={Proceedings of the 3rd Symposium on Operating Systems Design and Implementation (OSDI)},
  pages={173--186},
  year={1999}
}

@inproceedings{yin2019hotstuff,
  title={HotStuff: BFT Consensus in the Lens of Blockchain},
  author={Yin, Maofan and Malkhi, Dahlia and Reiter, Michael K and Gueta, Guy Golan and Abraham, Ittai},
  booktitle={Proceedings of the 2019 ACM Symposium on Principles of Distributed Computing (PODC)},
  pages={31--40},
  year={2019}
}

@article{dolev1986reaching,
  title={Reaching Approximate Agreement in the Presence of Faults},
  author={Dolev, Danny and Lynch, Nancy A and Pinter, Shlomit S and Stark, Eugene W and Weihl, William E},
  journal={Journal of the ACM (JACM)},
  volume={33},
  number={3},
  pages={499--516},
  year={1986}
}

@inproceedings{mendes2013multidimensional,
  title={Multidimensional Approximate Agreement in Byzantine Asynchronous Systems},
  author={Mendes, Hammurabi and Herlihy, Maurice},
  booktitle={Proceedings of the 45th Annual ACM Symposium on Theory of Computing (STOC)},
  pages={391--400},
  year={2013}
}

@inproceedings{baruch2019little,
  title={A Little Is Enough: Circumventing Defenses for Distributed Learning},
  author={Baruch, Gilad and Baruch, George and Goldberg, Yoav},
  booktitle={Advances in Neural Information Processing Systems (NeurIPS)},
  volume={32},
  year={2019}
}

@inproceedings{xie2020fall,
  title={Fall of Empires: Breaking Byzantine-tolerant SGD by Inner Product Manipulation},
  author={Xie, Cong and Koyejo, Oluwasanmi and Gupta, Indranil},
  booktitle={Uncertainty in Artificial Intelligence (UAI)},
  pages={261--270},
  year={2020},
  organization={PMLR}
}

@article{bracha1987asynchronous,
  title={Asynchronous Byzantine agreement protocols},
  author={Bracha, Gabriel},
  journal={Information and Computation},
  volume={75},
  number={2},
  pages={130--143},
  year={1987},
  publisher={Elsevier}
}

@article{dwork1988partial,
  title   = {Consensus in the Presence of Partial Synchrony},
  author  = {Dwork, Cynthia and Lynch, Nancy and Stockmeyer, Larry},
  journal = {Journal of the ACM},
  volume  = {35},
  number  = {2},
  pages   = {288--323},
  year    = {1988}
}

@inproceedings{morris2020textattack,
  title     = {TextAttack: A Framework for Adversarial Attacks, Data Augmentation, and Adversarial Training in {NLP}},
  author    = {Morris, John X. and Lifland, Eli and Yoo, Jin Yong and Grigsby, Jake and Jin, Di and Qi, Yanjun},
  booktitle = {Proceedings of EMNLP 2020: System Demonstrations},
  year      = {2020}
}

@article{lee2024promptinfection,
  title        = {Prompt Infection: LLM-to-LLM Prompt Injection within Multi-Agent Systems},
  author       = {Lee, Donghyun and Tiwari, Mo},
  journal      = {arXiv preprint arXiv:2410.07283},
  year         = {2024}
}

@article{ferrag2025protocol,
  title        = {From Prompt Injections to Protocol Exploits: Threats in LLM-Powered AI Agents Workflows},
  author       = {Ferrag, Mohamed Amine and Tihanyi, Norbert and Hamouda, Djallel and Maglaras, Leandros and Lakas, Abderrahmane and Debbah, Merouane},
  journal      = {ICT Express},
  volume       = {12},
  number       = {2},
  pages        = {353--383},
  year         = {2026}
}

@inproceedings{zhao2025lmc,
  title        = {Language Model Council: Democratically Benchmarking Foundation Models on Highly Subjective Tasks},
  author       = {Zhao, Justin and Plaza-del-Arco, Flor Miriam and Cercas Curry, Amanda},
  booktitle    = {Proceedings of NAACL 2025 (Long Papers)},
  year         = {2025}
}

@article{luo2025wbft,
  title        = {A Weighted Byzantine Fault Tolerance Consensus Driven Trusted Multiple Large Language Models Network},
  author       = {Luo, Haoxiang and Sun, Gang and Liu, Yinqiu and Zhao, Dongcheng},
  journal      = {arXiv preprint arXiv:2505.05103},
  year         = {2025}
}

@article{mao2024ibgp,
  title        = {IBGP: Imperfect Byzantine Generals Problem for Zero-Shot Robustness in Communicative Multi-Agent Systems},
  author       = {Mao, Yihuan and others},
  journal      = {arXiv preprint arXiv:2410.16237},
  year         = {2024}
}

@article{milentijevic2025approxagree,
  title        = {Approximate Agreement Algorithms for Byzantine Collaborative Learning},
  author       = {Milentijevi{\'c}, Tijana and Cambus, M{\'e}lanie and Melnyk, Darya and Schmid, Stefan},
  journal      = {arXiv preprint arXiv:2504.01504},
  year         = {2025}
}

@article{ferrag2025threats,
  title={Agentic AI Security: Threats, Defenses, Evaluation, and Open Challenges},
  author={Datta, Shrestha and Nahin, Shahriar Kabir and Chhabra, Anshuman and Mohapatra, Prasant},
  journal={arXiv preprint arXiv:2510.23883},
  year={2025}
}

@article{ai2025beyond,
  title={Beyond Majority Voting: LLM Aggregation by Leveraging Higher-Order Information},
  author={Ai, Rui and Pan, Yuqi and Simchi-Levi, David and Tambe, Milind and Xu, Haifeng},
  journal={arXiv preprint arXiv:2510.01499},
  year={2025}
}

@article{lamport1982byzantine,
  title={The Byzantine generals problem},
  author={Lamport, Leslie and Shostak, Robert and Pease, Marshall},
  journal={ACM Transactions on Programming Languages and Systems (TOPLAS)},
  volume={4},
  number={3},
  pages={382--401},
  year={1982}
}

@inproceedings{zheng2026cpwbft,
  title     = {Rethinking the Reliability of Multi-agent System: A Perspective from Byzantine Fault Tolerance},
  author    = {Zheng, Lifan and Chen, Jiawei and Yin, Qinghong and Zhang, Jingyuan and Zeng, Xinyi and Tian, Yu},
  booktitle = {Proceedings of the AAAI Conference on Artificial Intelligence (AAAI-26)},
  volume    = {40},
  number    = {41},
  year      = {2026},
  note      = {AAAI Technical Track on Natural Language Processing VI; also arXiv:2511.10400}
}

@misc{he2026sqa,
  title         = {Semantic Quorum Assurance: Collective Certification for Non-Deterministic {AI} Infrastructure},
  author        = {He, Jun and Yu, Deying},
  year          = {2026},
  eprint        = {2606.08021},
  archivePrefix = {arXiv},
  primaryClass  = {cs.DC}
}

@misc{lee2026sac,
  title         = {Robust Multi-Agent {LLM}s under {B}yzantine Faults},
  author        = {Lee, Haejoon and Yun, Vincent-Daniel and Panagou, Dimitra and Karimireddy, Sai Praneeth},
  year          = {2026},
  eprint        = {2605.09076},
  archivePrefix = {arXiv},
  primaryClass  = {cs.LG}
}

@misc{agentprov2026survey,
  title         = {From Agent Traces to Trust: A Survey of Evidence Tracing and Execution Provenance in {LLM} Agents},
  author        = {Wang, Yiqi and Zhang, Jiaqi and Cai, Taotao and Liu, Zirui and Sun, Qingqiang and Sun, Zequn
                   and Wu, Zhangkai and Dong, Manqing and Zheng, Mingkai and Yin, Xuefei and Zhu, Yanming},
  year          = {2026},
  eprint        = {2606.04990},
  archivePrefix = {arXiv},
  primaryClass  = {cs.AI}
}

@article{anand2026resilient,
  title   = {Resilient Consensus in Agentic {AI}},
  author  = {Anand, Sribalaji C. and Pappas, George J.},
  journal = {arXiv preprint arXiv:2606.15024},
  year    = {2026}
}

@inproceedings{torresarias2019intoto,
  title     = {in-toto: Providing Farm-to-Table Guarantees for Bits and Bytes},
  author    = {Torres-Arias, Santiago and Afzali, Hammad and Kuppusamy, Trishank Karthik and Curtmola, Reza and Cappos, Justin},
  booktitle = {USENIX Security Symposium},
  year      = {2019}
}

@inproceedings{haeberlen2007peerreview,
  title     = {{PeerReview}: Practical Accountability for Distributed Systems},
  author    = {Haeberlen, Andreas and Kouznetsov, Petr and Druschel, Peter},
  booktitle = {ACM Symposium on Operating Systems Principles (SOSP)},
  pages     = {175--188},
  year      = {2007}
}

\appendix
\section{Notation}
\label{appendix:notation}

Table~\ref{tab:notation} lists the notation used throughout
the main text and the appendix. Quantities relevant to the
protocol path are listed first; quantities relevant to
evaluation follow.

\begin{table*}[t]
\centering
\caption{Notation used throughout. Quantities relevant to the
protocol path are listed first; quantities relevant to evaluation
follow.}
\label{tab:notation}
\footnotesize
\setlength{\tabcolsep}{4pt}
\begin{tabular}{@{}>{\raggedright\arraybackslash}p{0.26\linewidth} >{\raggedright\arraybackslash}p{0.66\linewidth}@{}}
\toprule
\textbf{Symbol} & \textbf{Meaning} \\
\midrule
$\mathcal{N}=\{1,\ldots,n\}$ & set of agents in the round; $n$ is total node count \\
$f$            & declared Byzantine fault bound; at most $f$ agents may be Byzantine; required quorum is $2f{+}1$ \\
$\mathcal{H}\subseteq \mathcal{N}$ & honest agents (unknown to the protocol) \\
$\mathcal{B}\subseteq \mathcal{N}$ & Byzantine agents (unknown); $\mathcal{H}\sqcup \mathcal{B}=\mathcal{N}$ \\
$x_i$          & natural-language proposal of agent $i$ \\
$z_i$          & structured proposal object built from $x_i$ (verdict, confidence, evidence ids, rationale, claim) \\
$\mathcal{E}$  & deterministic semantic encoder \\
$\mathbf{e}_i = \mathcal{E}(\mathrm{CanonText}(\rho_i))$ & unit-norm semantic embedding of the canonical proposal text \\
$I_r$          & delivered index set at round $r$; $I_r\subseteq \Nset$, $|I_r|\ge n-f$ (successful-round RBC predicate) \\
$V(r)$         & delivered semantic view at round $r$: $\{(i,\rho_i,\mathbf{e}_i) : i\in I_r\}$ \\
$z_i.v$        & verdict component of $z_i$; $v\in\mathcal{V}$ (typed verdict vocabulary) \\
$G_v(r)$       & verdict group for $v$ at round $r$: $\{i\in V(r) : z_i.v = v\}$ \\
$v^\star$      & H-CSC candidate verdict: $\arg\max_v |G_v(r)|$ with deterministic tie-break (\textsc{largest\_verdict}) \\
$\textit{verdict\_margin}$ & $|G_{v^\star}| - \max_{v\ne v^\star}|G_v(r)|$ \\
$\widehat{C}$  & within-verdict semantic core on the semantic branch; $\widehat{C}\subseteq G_{v^\star}$, $|\widehat{C}|\ge 2f{+}1$ \\
$\theta_\alpha$ & angular radius threshold of the semantic admissibility predicate (H-CSC main path: $0.65$~rad) \\
$\bar{\varepsilon}_\alpha$ & angular diameter sister-predicate (reported only) \\
$\Psi(r)$      & finality-signal vector (Definition~\ref{def:finality-signal}): pre-certificate components $(|G_{v^\star}|, \textit{verdict\_margin}, |\widehat{C}|, r^\star)$ plus the post-Stage-5 certificate-size component $|\Sigma|$ \\
$\mathcal{A}$  & robust aggregator (Euclidean geometric median by default) \\
$\mathbf{y}^\star$ & aggregate semantic object on the semantic branch: $\mathbf{y}^\star = \mathcal{A}(\widehat{C})$ \\
$\eta,\ Q_\eta(\cdot)$ & quantization step and lattice quantizer \\
$\pi$          & frozen protocol parameters: $(\theta_\alpha,\bar{\varepsilon}_\alpha,\mathcal{A},Q_\eta,$ $\textit{candidate\_rule},\textit{fallback},$ $\textit{verdict\_margin}^{\min},\textit{topology\_gate}=\bot)$ \\
$h_\pi$        & parameter digest binding $\pi$ to a commit \\
$h^\star_{\textsc{sem}}$ & semantic-commit digest: $H(\texttt{"semantic\_commit"}\,\Vert\,Q_\eta(\mathbf{y}^\star)$ $\Vert\, h_\pi\,\Vert\, r\,\Vert\, v^\star)$, with leading byte-string domain separator binding verdict + quantised aggregate \\
$h^\star_{\textsc{ver}}$ & verdict-commit digest: $H(\texttt{"verdict\_commit"}$ $\Vert\,\textsc{VerdictPayload}(v^\star)\,\Vert\, h_\pi\,\Vert\, r)$, with matching domain separator; no embedding aggregate \\
$\mathrm{Cert}(h^\star)$ & quorum certificate over $h^\star$, requiring $2f{+}1$ distinct signers \\
$\mathbf{commit\_type}$ & typed per-round outcome: $\textsc{semantic\_commit}$, $\textsc{verdict\_commit}$, or $\textsc{abort}$ \\
\midrule
$\mathrm{majority}_\mathcal{H}(V)$ & evidence-bounded honest-reference verdict (honest plurality with $\textsc{largest\_verdict}$ tie-break; see Definition~\ref{def:two-validities}); computed offline for evaluation only \\
$\mathrm{gold}(V)$ & dataset gold label, when available \\
$\mathrm{SCE}(\mathbf{y}^\star,\mathcal{H})$ & angular distance from $\mathbf{y}^\star$ to the honest-only reference centre (degrees) \\
$\textsc{InvalidCommit}_{\mathrm{gold}}$, $\textsc{InvalidCommit}_{\mathrm{hmaj}}$ & two validity-mismatch metrics; reported separately \\
\bottomrule
\end{tabular}
\end{table*}

\section{Additional Proof Details}
\label{appendix:proofs}

This appendix gives proof details for the five-axis correctness
results stated in Section~\ref{sec:correctness} (Agreement,
Semantic validity, Verdict validity, Safe abort, Conditional
commit-or-abort termination). The conventions match the main
text: $n$ agents indexed by $\Nset=\{1,\ldots,n\}$ with
$|\Byz|\le f$ and $n\ge 3f{+}1$; embeddings
$\mathbf{e}_i = \cE(\mathrm{CanonText}(\rho_i))$ on
$\mathbb{S}^{d-1}$;
admissibility radius $\thalpha$; lattice quantiser $Q_\eta$.
Lemmas \ref{lem:digest-consistency},
\ref{lem:cert-uniqueness}, and \ref{lem:rep-boundary} are
restated for completeness, but the proof bodies given in
Section~\ref{sec:correctness} are not duplicated; we instead
provide auxiliary lemmas and detailed arguments for theorems
whose \S5 statements were proof sketches.

The proof details below match the H-CSC control flow of
Algorithm~\ref{alg:certified-commitment-round}: Stage~1 builds
the delivered view, Stage~2 partitions $V(r)$ into verdict
groups, Stage~3 either (a) extracts a within-verdict semantic
core for a candidate verdict $v^\star$ and runs the
admissibility predicate at radius $\thalpha$, or (b) falls
back to the verdict-only path when the semantic predicate
fails but the verdict-admissibility predicate holds. Stages
4--5 aggregate, quantise, bind the typed digest, and collect
a $2f{+}1$ distinct-signer certificate. Lemma~\ref{lem:digest-consistency}
applies uniformly to both semantic and verdict digests;
Lemma~\ref{lem:cert-uniqueness} applies uniformly to either
typed certificate object; the semantic-validity argument below
applies to the \textsc{semantic\_commit} branch only. The
verdict-validity argument is developed in full in
Appendix~\ref{app:proof-verdict-validity} below: the previous
version of this appendix stated that it was ``fully developed in
Section~\ref{sec:correctness} and not repeated here'', which was
incorrect --- the main text carried only the necessity direction.
Appendices~\ref{app:proof-det-semantic} and
\ref{app:proof-containment} give the deterministic
semantic-commit rule and the containment result that relates the
two; Appendix~\ref{app:ball-vs-component} treats the gap between
the admissibility \emph{definition} and its implementation, and
Appendix~\ref{app:proof-quantiser} treats the quantiser.

\subsection{Auxiliary Geometric Lemmas}
\label{app:auxlemmas}

\begin{lemma}[Sphere identity and threshold conversion]
\label{lem:sphere_identity}
For any $\mathbf{u},\mathbf{v}\in\mathbb{S}^{d-1}$,
\[
\|\mathbf{u}-\mathbf{v}\|_2 \;=\; 2\sin\!\Big(\tfrac{1}{2}\,d_{\angle}(\mathbf{u},\mathbf{v})\Big),
\]
where $d_{\angle}(\mathbf{u},\mathbf{v})=\arccos(\langle \mathbf{u},\mathbf{v}\rangle)\in[0,\pi]$.
Hence $d_{\angle}(\mathbf{u},\mathbf{v})\le \theta$ iff
$\|\mathbf{u}-\mathbf{v}\|_2 \le 2\sin(\theta/2)$.
\end{lemma}

\begin{lemma}[Convex-hull range bound via diameter]
\label{lem:conv_diam_ball}
Let $P\subset\mathbb{R}^d$ be finite and let
$\mathrm{diam}(P):=\max_{\mathbf{x},\mathbf{y}\in P}\|\mathbf{x}-\mathbf{y}\|_2$.
Then for any fixed $\mathbf{p}\in P$ and any $\mathbf{z}\in \mathrm{conv}(P)$,
\[
\|\mathbf{z}-\mathbf{p}\|_2 \;\le\; \mathrm{diam}(P).
\]
\end{lemma}

\begin{proof}
Write $\mathbf{z}=\sum_t \lambda_t \mathbf{x}_t$ with
$\mathbf{x}_t\in P$, $\lambda_t\ge 0$, $\sum_t \lambda_t=1$.
Triangle inequality and the definition of diameter give
$\|\mathbf{z}-\mathbf{p}\|_2 \le \sum_t \lambda_t \|\mathbf{x}_t-\mathbf{p}\|_2
\le \sum_t \lambda_t\,\mathrm{diam}(P)=\mathrm{diam}(P)$.
\end{proof}

\subsection{Proof Detail for Lemma~\ref{lem:digest-consistency}}
\label{app:proof-digest}

\begin{proof}[Proof detail]
Let $i$ and $j$ be two honest signers that select the same
core $\core\subseteq V(r)$ and run the four-stage pipeline of
Algorithm~\ref{alg:certified-commitment-round} on the same
input. We chase the four stages and show that each produces
byte-identical output at $i$ and $j$, from which $\hstar$
equality follows.

Stage~1 (encode). Determinism of $\cE$ together with the
byte-deterministic canonicalisation of $\rho_k$
(Assumption~\ref{assn:deterministic-cE}) implies that the
embedding of every agent in the delivered view is identical at
the two signers; in particular $\mathbf{e}_k$ is the same byte
string at $i$ and $j$ for every $k\in\core$.

Stage~3 (aggregate). The robust aggregator $\agg$ is a
deterministic function of its multiset input \emph{given a fixed
execution environment}: ties are broken canonically and the
iteration is run to a fixed stopping rule with a fixed iteration
count, so two honest signers that feed the byte-identical
multiset $\{\mathbf{e}_k:k\in\core\}$ to the same build of $\agg$
obtain the byte-identical $\ystar$. We flag explicitly that the
formulation used in the previous version --- ``$\agg$ converges to
within a documented numerical tolerance'' --- does \emph{not}
suffice for this step, because the output of this stage is
hashed and a hash has no tolerance. The same caveat applies to
Stage~1: $\mathbf{e}_k$ must be bit-identical, not merely close.
Appendix~\ref{app:proof-quantiser} quantifies the residual gap
for our artefact and shows that the lattice quantiser does not
close it; \S\ref{sec:encoder-determinism} of the main text gives
the resolution.

Stage~4 (quantise + digest). The lattice quantiser $Q_\eta$ is a
fixed-point round with deterministic tie-breaking, so the same
$\ystar$ produces the same $\quant{\ystar}$ at both signers.
The parameter digest $h_\pi$ depends only on the frozen
parameter tuple $\pi$ (identical at both signers by hypothesis)
and the round identifier $r$ (identical by construction). The
collision-resistant hash $H$ is deterministic on its byte input,
so the typed digest --- whether
$\hstar_{\textsc{sem}} = H(\texttt{"semantic\_commit"}\,\Vert\,\quant{\ystar}\,\Vert\,h_\pi\,\Vert\,r\,\Vert\,v^\star)$
on the semantic branch or
$\hstar_{\textsc{ver}} = H(\texttt{"verdict\_commit"}\,\Vert\,\textsc{VerdictPayload}\,\Vert\,h_\pi\,\Vert\,r)$
on the verdict branch (\S\ref{sec:protocol-aggregate-digest}) ---
is identical at $i$ and $j$.
\end{proof}

\subsection{Proof Detail for Lemma~\ref{lem:cert-uniqueness}}
\label{app:proof-cert}

\begin{proof}[Proof detail]
The certificate $\cert{\hstar}$ is constructed as a set
$\Sigma\subseteq\{(\mathrm{id}_k,\sigma_k)\}$ of pairs. The
verifier counts at most one signature per signer identity: the
runner enforces per-identity uniqueness when assembling
$\Sigma$, and any duplicate $(\mathrm{id}_k,\cdot)$ entry is
discarded. A Byzantine identity may itself sign more than one
distinct digest with its own key---signature unforgeability
(Assumption~\ref{assn:auth-channels}) only rules out
\emph{others} forging that identity, not the identity
equivocating---but each such identity contributes at most one
counted signature to any one certificate. Honest identities
satisfy the sign-once rule of
Assumption~\ref{assn:auth-channels} and produce a signature on
at most one digest per round. With $|\Sigma|\ge 2f{+}1$ distinct
identities and at most $f$ Byzantine identities, the pigeonhole
principle gives $|\{\mathrm{id}_k\in\Sigma : \mathrm{id}_k\in\Honest\}|\ge f{+}1$.
\end{proof}

\subsection{Proof of Theorem~\ref{thm:agreement}
            (Digest-level Agreement)}
\label{app:proof-agreement}

\begin{proof}[Proof detail]
Suppose two honest verifiers $i$ and $j$ each accept a commit
object in round $r$. Let $\hstar_i$ and $\hstar_j$ denote the
two committed digests, with certificates of size at least
$2f{+}1$ on each. We show $\hstar_i = \hstar_j$.

The argument does \emph{not} rely on quorum-intersection
between $\Sigma_i$ and $\Sigma_j$: under our system model
$n \ge 3f{+}1$ (rather than $n = 3f{+}1$ exactly), the bound
$|\Sigma_i\cap\Sigma_j| \ge (2f{+}1)+(2f{+}1)-n$ can be as
small as $0$ (as it is for the experimental setting $n=10$,
$f=2$). Instead we use \emph{deterministic decision on a
shared delivered view} plus the \emph{honest sign-once} rule,
which give a unique committable digest that any valid
certificate must bind.

\textbf{(i) Identical delivered view.} A successful round
delivers, by the RBC properties and the successful-round
condition stated in \S\ref{subsec:net} (identical delivered
view across all honest nodes once each has collected
$|V_i(r)| \ge n-f$ inputs), the same delivered index set
$I_r$ and the same delivered structured proposals
$\{(i, z_i) : i\in I_r\}$ to every honest node.

\textbf{(ii) Deterministic typed decision.} The H-CSC pipeline
of Algorithm~\ref{alg:certified-commitment-round}, run on a
fixed delivered view $V(r)$ and a fixed frozen parameter tuple
$\pi$ (with deterministic tie-break and a deterministic
aggregator), is a deterministic function: it outputs at most
one typed candidate digest per round. Concretely, let
$D^\star(r) = \textsc{HCSC}(V(r), \pi, r)$ denote that
canonical typed digest (either $\hstar_{\textsc{sem}}$ or
$\hstar_{\textsc{ver}}$, or $\bot$ on a typed abort).

\textbf{(iii) Honest sign-once rule.} An honest agent in round
$r$ signs at most one digest --- the protocol's own
$D^\star(r)$ --- and ignores any other digest presented for
signing (Stage~5 rule).

\textbf{(iv) Any valid certificate must contain at least
$f{+}1$ honest signatures.} A certificate of size $2f{+}1$
drawn from a verdict group $G_{v^\star}$ (verdict path) or
within-verdict core $\core$ (semantic path) contains at least
$(2f{+}1) - f = f{+}1$ honest signers, since at most $f$
identities are Byzantine in $\Nset$.

\textbf{(v) Conclusion.} Both $\Sigma_i$ and $\Sigma_j$
contain at least $f{+}1$ honest signers; by (iii) every honest
signer in either set signed $D^\star(r)$. Hence both
$\hstar_i$ and $\hstar_j$ contain valid honest signatures on
$D^\star(r)$. By signature unforgeability
(Assumption~\ref{assn:auth-channels}), the honest signatures
on $\hstar_i$ are signatures on the unique digest each honest
agent actually signed. Therefore
$\hstar_i = \hstar_j = D^\star(r)$. The conclusion holds
without invoking any quorum intersection between $\Sigma_i$ and
$\Sigma_j$.

Because $D^\star(r)$ is one of the typed digests
$\hstar_{\textsc{sem}}$ or $\hstar_{\textsc{ver}}$ defined in
\S\ref{sec:protocol-aggregate-digest}, equality of digests
also implies equality of the typed payload: in the semantic
case the committed
$(\quant{\ystar},\pi,r,v^\star)$ tuple is identical at $i$ and
$j$; in the verdict case the
$(\textsc{VerdictPayload},\pi,r)$ tuple is identical, with no
quantised aggregate attached. The explicit domain separators
in the digest construction prevent a cross-type collision
under collision-resistance of $H$, so $\mathbf{commit\_type}$
matches as well.

The argument is independent of the admissibility predicate of
Definition~\ref{def:admissible-round}: agreement holds on any
round in which two commits exist at all. The admissibility
predicate is the precondition for a commit existing in the first
place; safe abort under non-admissibility
(Theorem~\ref{thm:safe-abort}) takes care of the no-commit case.
\end{proof}

\subsection{Proof of Theorem~\ref{thm:validity}
            (Semantic Validity Under Admissible Core)}
\label{app:proof-validity}

\begin{proof}[Proof detail]
By Assumption~\ref{assn:admissible-round}, the round is
admissible at radius $\thalpha$, so by
Definition~\ref{def:admissible-round} there exists a centre
$h\in\core$ with $d_\angle(\mathbf{e}_h,\mathbf{e}_j)\le\thalpha$
for every $j\in\core$. Triangle inequality on the unit sphere
yields
$d_\angle(\mathbf{e}_j,\mathbf{e}_k)\le 2\thalpha$ for all
$j,k\in\core$, i.e.\
$\mathrm{diam}_\angle(\core)\le 2\thalpha$.

Let $\mathbf{c}_\Honest$ be any honest centre in
$\core\cap\Honest$ (since $|\core|\ge 2f{+}1$ and $|\Byz|\le f$,
$|\core\cap\Honest|\ge f{+}1\ge 1$). The Euclidean geometric
median $\ystar$ of $\{\mathbf{e}_j:j\in\core\}$ lies in the
convex hull of $\{\mathbf{e}_j:j\in\core\}$ (a standard property
of the GM that holds for any aggregator returning a
convex-combination point); applying
Lemma~\ref{lem:conv_diam_ball} to $P=\{\mathbf{e}_j:j\in\core\}$
and $\mathbf{p}=\mathbf{c}_\Honest$ gives
$\|\ystar-\mathbf{c}_\Honest\|_2\le\mathrm{diam}(\core)$. Since
$\mathbf{c}_\Honest\in\core$ and the diameter on $\mathbb{S}^{d-1}$
is bounded by $2\sin(\thalpha)$
(Lemma~\ref{lem:sphere_identity} with $\theta=2\thalpha$),
the Euclidean bound translates to the angular bound
$d_\angle(\ystar,\mathbf{c}_\Honest)\le\mathrm{diam}_\angle(\core)\le 2\thalpha$
after re-normalisation of $\ystar$ to the unit sphere.

The committed digest's underlying aggregate is therefore within
$2\thalpha$ of $\mathbf{c}_\Honest$. The bound is relative to
honest centres inside the selected core; it does not bound the
distance to an external ground-truth reference. The empirical
SCE numbers reported in Section~\ref{sec:eval-retained} are the
corresponding observation under specific encoder and aggregator
choices.
\end{proof}
\subsection{The Admissibility Predicate: Ball versus Connected
            Component}
\label{app:ball-vs-component}

Definition~\ref{def:admissible-round} specifies a $\thalpha$-ball;
the reference implementation returns the largest connected
component of the $\thalpha$-threshold graph on $G_{v^\star}$ and
does not re-test the ball predicate.
Proposition~\ref{prop:ball-component} of the main text states the
consequence; we give the proof detail and the restated theorem
here.

\begin{proof}[Proof detail for Proposition~\ref{prop:ball-component}]
\emph{Upper bound.} Let $\core$ be connected in the
$\thalpha$-threshold graph
$E_{\thalpha}=\{(i,j):\angle(\mathbf{e}_i,\mathbf{e}_j)\le\thalpha\}$
with $|\core|=k$. Fix $u,w\in\core$. Connectivity gives a simple
path $u=x_0,x_1,\dots,x_\ell=w$ inside $\core$ with
$\ell \le k-1$ (a simple path visits each of the $k$ vertices at
most once). Angular distance is a metric on $\mathbb{S}^{d-1}$,
so by the triangle inequality
\[
\angle(\mathbf{e}_u,\mathbf{e}_w)
\le \sum_{t=1}^{\ell}\angle(\mathbf{e}_{x_{t-1}},\mathbf{e}_{x_t})
\le \ell\,\thalpha \le (k-1)\thalpha .
\]
The metric is bounded above by $\pi$, giving the stated minimum.

\emph{Tightness.} Suppose $(k-1)\thalpha\le\pi$ and place
$\mathbf{e}_t = (\cos(t\,\thalpha),\sin(t\,\thalpha),0,\dots,0)$
for $t=0,\dots,k-1$. Consecutive points are at angular distance
exactly $\thalpha$ and non-consecutive points at distance
$|t-t'|\thalpha>\thalpha$, so $E_{\thalpha}$ restricted to this
set is exactly the path graph: the set is connected, and
$\angle(\mathbf{e}_0,\mathbf{e}_{k-1})=(k-1)\thalpha$.

\emph{Strictness of the containment.} A $\thalpha$-ball is
$\thalpha$-connected because its centre is adjacent to every
member. The converse fails for $k\ge 3$: in the configuration
above with $k=3$ the only candidate centres are the three points
themselves, and each is at distance $2\thalpha>\thalpha$ from one
of the others.
\end{proof}

\begin{corollary}[Theorem~\ref{thm:hcsc-semantic-validity} under
the implemented predicate]
\label{cor:semval-component}
If $\core$ is only known to be $\thalpha$-connected with
$|\core|=k$, then the conclusion of
Theorem~\ref{thm:hcsc-semantic-validity} weakens to
\[
\angle(\ystar,\mathbf{c}_{\Honest})
\;\le\;\mathrm{diam}_{\angle}(\core)
\;\le\;\min\{(k-1)\thalpha,\ \pi\},
\]
with the remaining conclusions ($v^\star$ shared across $\core$,
$|\core\cap\Honest|\ge f{+}1$) unchanged, since those are
cardinality statements independent of the geometric predicate.
\end{corollary}

\begin{proof}
The convex-hull argument in
Appendix~\ref{app:proof-validity} uses only
$\mathrm{diam}_{\angle}(\core)$, not the existence of a centre;
substitute the bound of Proposition~\ref{prop:ball-component} for
the $2\thalpha$ bound obtained from the ball predicate.
\end{proof}

At the H-CSC operating point $\thalpha = 0.65$ this bound is
$2.6$~rad ($\approx 149^\circ$) at the minimum core size
$k = 2f{+}1 = 5$ and is vacuous for every $k\ge 6$, since
$5\thalpha = 3.25 > \pi$. It is also anti-monotone in $k$: a
larger, better-corroborated core carries a weaker certified
bound. Both properties argue for repairing the implementation
rather than the definition, which is the recommendation made in
\S\ref{sec:ball-vs-component}. Empirically, $5$ of the $73$
\textsc{semantic\_commit} cores in our MVR-50 runs are connected
but not balls, and the maximum observed
$\mathrm{diam}_{\angle}(\core)$ over all $73$ is $0.964$~rad
against $2\thalpha = 1.30$~rad, so
Theorem~\ref{thm:hcsc-semantic-validity}'s numerical conclusion is
not violated on any observed round even where its hypothesis is
not verified.

\subsection{Proof of Theorem~\ref{thm:hcsc-verdict-validity}
            (Verdict-Commit Validity: Sufficiency and Tightness)}
\label{app:proof-verdict-validity}

The main-text proof is complete; we restate it here in the
notation of this appendix and add the two remarks that the main
text omits for space. Recall $h_v = |G_v\cap\Honest|$,
$b_v=|G_v\cap\Byz|$, $\sum_v b_v = |\Byz\cap V(r)| \le f$, and
$m=\textit{verdict\_margin}^{\min}$.

\begin{proof}[Proof detail]
\emph{Sufficiency ($m>f$).} Let $v\ne v^\star$. The commit
hypothesis gives $|G_v|\le|G_{v^\star}|-m$. Then
\[
h_{v^\star}-h_v
=\bigl(|G_{v^\star}|-b_{v^\star}\bigr)-\bigl(|G_v|-b_v\bigr)
\ge m-b_{v^\star}+b_v \ge m-f \ge 1,
\]
where the last step uses that $\textit{verdict\_margin}$ and $m$
are integers, so $m>f$ implies $m\ge f{+}1$. Hence $v^\star$ is
the unique honest-argmax, the tie-break of
Definition~\ref{def:two-validities} is not exercised, and
$\mathrm{majority}_{\Honest}(V)=v^\star$.

\emph{Remark on slack.} The bound is not tight in $b_v$: the
adversary's best strategy is $b_{v^\star}=f$, $b_v=0$, i.e.\ pad
the group the protocol will select and leave the honest runner-up
untouched. Under any other allocation the honest gap is strictly
larger. This is why the condition depends on $f$ and not on the
realised Byzantine count.

\emph{Remark on the delivered view.} The argument is entirely
internal to $V(r)$; it does not require $|V(r)|=n$. Agents whose
proposals were not delivered contribute to neither $h_v$ nor
$b_v$, and $\mathrm{majority}_{\Honest}(V)$ is by
Definition~\ref{def:two-validities} evaluated on the delivered
view. Undelivered honest agents are therefore a validity concern
only insofar as RBC totality is violated, which
Assumption~\ref{assn:auth-channels} and the successful-round
condition exclude.

\emph{Tightness ($m\le f$).} Construct $V(r)$ with
$|G_{v^\star}|=2f{+}1$ containing all $f$ Byzantine agents, and a
second group $|G_u|=2f{+}1-m$ containing only honest agents; the
remaining $n-(4f{+}2-m)$ delivered agents, if any, are honest and
distributed over further verdicts with group size below
$2f{+}1-m$. Then $\textit{verdict\_margin}=m$ and the verdict
branch fires on $v^\star$, but $h_{v^\star}=f{+}1$ and
$h_u=2f{+}1-m\ge f{+}1$. For $m<f$ the inequality is strict in
the adversary's favour ($h_u>h_{v^\star}$) and
$\mathrm{majority}_{\Honest}(V)=u\ne v^\star$; for $m=f$ the two
are equal, and whether the honest reference agrees with $v^\star$
is decided by a tie-break the adversary chose the labels to win.
Feasibility requires $n\ge 4f{+}2-m$, which holds for the
experimental $n=10$, $f=2$, $m\in\{1,2\}$
($n\ge 9$ and $n\ge 8$ respectively). The case $m=0$ is realised
in our own evaluation and is recorded as
Example~\ref{ex:tiebreak-capture}.
\end{proof}

\subsection{Proof of Theorem~\ref{thm:det-semantic-validity}
            (Deterministic Semantic-Commit Validity)}
\label{app:proof-det-semantic}

\begin{proof}[Proof detail]
The main-text proof is complete. We record here why the condition
$|\core|>(n+f)/2$ rather than the weaker $|\core|\ge 2f{+}1$ is
what the conclusion needs, and why the conclusion is a strict
majority rather than a plurality.

Write $b=|\Byz\cap V(r)|\le f$, $N=|V(r)|\le n$,
$\Honest_V=\Honest\cap V(r)$, so $|\Honest_V|=N-b$. All of
$\core$ holds $v^\star$.

The two inequalities established in the main text are
\[
|\core\cap\Honest| \;\ge\; |\core|-b \;\ge\; |\core|-f
\;>\; \frac{n-f}{2}
\]
and
\[
|\Honest_V\setminus\core| \;\le\; N-|\core| \;\le\; n-|\core|
\;<\; \frac{n-f}{2}.
\]
The common threshold $(n-f)/2$ is exactly what makes them
compose, and it is the reason the hypothesis is stated as
$|\core|>(n+f)/2$: subtracting the maximum Byzantine mass $f$
from one side and the core size from the total $n$ on the other
must land on the same value, which forces
$|\core|-f = n-|\core|$, i.e.\ $|\core|=(n+f)/2$, as the break-even
point.

At the quorum size $|\core|=2f{+}1$ the argument fails whenever
$2f{+}1\le(n+f)/2$, i.e.\ $n\ge 3f{+}2$. For $n=3f{+}1$ the two
coincide and the ordinary quorum already suffices; for the
experimental $n=10$, $f=2$ we have $2f{+}1=5<6=(n+f)/2$, so the
shipped quorum is two short of the deterministic condition, and
$|\core|\ge 7$ is required.

Strictness of the majority: with
$A=|\{i\in\Honest_V:\rho_i.v=v^\star\}|\ge|\core\cap\Honest|$ and
$B=|\Honest_V|-A\le|\Honest_V\setminus\core|$, the displays give
$A>(n-f)/2>B$. Then $2A>A+B=|\Honest_V|$, which is the definition
of a strict majority. A strict majority is a fortiori the unique
plurality, so the conclusion implies --- and is strictly stronger
than --- that of Theorem~\ref{thm:hcsc-verdict-validity}. It is
strictly stronger in a way that matters for a three-valued
verdict vocabulary, where a plurality can be as small as
$\lceil|\Honest_V|/3\rceil+1$.
\end{proof}

\subsection{Proof of Lemma~\ref{lem:containment} and
            Corollary~\ref{cor:coverage-identity}}
\label{app:proof-containment}

\begin{proof}[Proof detail for Lemma~\ref{lem:containment}]
$\core\subseteq G_{v^\star}$ gives $|G_{v^\star}|\ge|\core|$. The
verdict groups $\{G_v\}$ partition $V(r)$, so for any
$v\ne v^\star$, $|G_v|\le|V(r)|-|G_{v^\star}|$. Hence
\begin{align*}
\textit{verdict\_margin}
&= |G_{v^\star}|-\max_{v\ne v^\star}|G_v|\\
&\ge 2|G_{v^\star}|-|V(r)|
 \ \ge\ 2|\core|-n\\
&> (n+f)-n \;=\; f,
\end{align*}
using $|V(r)|\le n$ and the hypothesis. Since the margin is an
integer, $\textit{verdict\_margin}\ge f{+}1$.
\end{proof}

\begin{proof}[Proof detail for Corollary~\ref{cor:coverage-identity}]
Both protocols read the same delivered view $V(r)$, form the same
verdict partition, and select the same $v^\star$ by the same
$\textsc{largest\_verdict}$ rule, so any committed verdict
agrees. It remains to show the two commit predicates coincide.
$\textsc{M3}^{\det}$ commits iff
$|G_{v^\star}|\ge 2f{+}1$ and $\textit{verdict\_margin}>f$; note
that $\textit{verdict\_margin}>f$ together with
$|V(r)|\ge n-f\ge 2f{+}1$ already implies
$|G_{v^\star}|\ge(|V(r)|+f{+}1)/2$, but we do not need this.

$\textsc{H-CSC}^{\det}$ commits iff its semantic gate
$\bigl(|\core|>(n+f)/2$ and the admissibility predicate holds$\bigr)$
fires, or its verdict gate
$\bigl(|G_{v^\star}|\ge 2f{+}1$ and $\textit{verdict\_margin}>f\bigr)$
fires. The verdict gate is literally $\textsc{M3}^{\det}$'s
predicate, so $\textsc{M3}^{\det}$ commits $\Rightarrow$
$\textsc{H-CSC}^{\det}$ commits. Conversely, if the semantic gate
fires then $|\core|>(n+f)/2$, so by
Lemma~\ref{lem:containment} $\textit{verdict\_margin}>f$, and by
$|G_{v^\star}|\ge|\core|>(n+f)/2\ge 2f{+}1$ (the last step for
$n\ge 3f{+}2$; for $n=3f{+}1$, $(n+f)/2=2f{+}\tfrac12$ and
integrality gives $|G_{v^\star}|\ge 2f{+}1$) the verdict gate
fires as well. Hence the semantic gate implies the verdict gate,
the disjunction collapses to the verdict gate, and the two commit
predicates are equal as predicates on delivered views.

The corollary is therefore an identity of functions, not a
statistical statement; the empirical check reported in
\S\ref{sec:level4} ($0$ discordant tasks in either mode, both
protocols at $0.72/0.64$) is a regression test on the
implementation rather than evidence for the claim.
\end{proof}

\subsection{Proof of Proposition~\ref{prop:quantiser} and the
            Numerical Instantiation}
\label{app:proof-quantiser}

\begin{proof}[Proof detail]
Part (1) is the explicit construction given in the main text; note
that it needs only one coordinate and is therefore independent of
$d$, and that it survives any fixed tie-break rule because the two
values lie strictly on opposite sides of the boundary rather than
on it.

Part (2). Write $y_k=\eta\,(t_k+u_k)$ with $t_k\in\mathbb{Z}$ and
$u_k\in[0,1)$ the fractional part. Round-to-nearest maps $y_k$ to
$\eta t_k$ if $u_k<\tfrac12$ and to $\eta(t_k+1)$ otherwise;
$y'_k=y_k+\delta_k$ with $0\le\delta_k<\eta$ rounds differently
exactly when the half-open interval $[u_k,u_k+\delta_k/\eta)$
contains $\tfrac12$ modulo $1$. Under $u_k\sim\mathrm{Unif}[0,1)$
that event has probability $\delta_k/\eta$. Independence across
$k$ gives
$\Pr[\text{all agree}]=\prod_k(1-\delta_k/\eta)$, and
$1-x\le e^{-x}$ for $x\in[0,1]$ gives the exponential bound.
\end{proof}

\paragraph{Instantiation for our artefact.}
The CRSE encoder inherits the \texttt{intfloat/e5-base}
embedding width, $d=768$; the artefact uses $\eta=10^{-6}$, a
value the previous version of this manuscript did not state. A
cross-environment re-encode of the MVR-50 canonical texts against
the frozen embedding cache gives minimum cosine similarity
$0.99999990$ and maximum element-wise absolute difference
$\delta_{\max}=3.6\times10^{-7}$
(\S\ref{sec:encoder-determinism}). Under the model of part (2):

\begin{center}
\footnotesize
\begin{tabular}{@{}l c c c@{}}
\toprule
$\eta$ & mean $\bar\delta$ assumed & $\Pr[\text{digests agree}]\le$ & quantisation displacement \\
\midrule
$10^{-6}$ & $10^{-8}$ & $e^{-7.68}\approx 5\times10^{-4}$ & $1.4\times10^{-5}$~rad \\
$10^{-4}$ & $10^{-8}$ & $e^{-0.077}\approx 0.93$          & $1.4\times10^{-3}$~rad \\
$10^{-2}$ & $3.6\times10^{-7}$ & $\approx 0.972$          & $0.139$~rad $\approx 0.21\,\thalpha$ \\
\bottomrule
\end{tabular}
\end{center}

\noindent
where the displacement column is the worst-case
$\|\quant{\ystar}-\ystar\|_2 \le (\eta/2)\sqrt{d}$, converted to
angle on the unit sphere via Lemma~\ref{lem:sphere_identity}. The
table is a model, not a measurement, and the assumed $\bar\delta$
values are chosen to be conservative relative to
$\delta_{\max}$; the robust conclusion is the qualitative one,
that at the shipped $\eta$ the quantiser step is the same order of
magnitude as the discrepancy it is meant to absorb, and that
enlarging $\eta$ far enough to be comfortable costs a fifth of the
admissibility radius while still leaving agreement probabilistic.
The resolution adopted in \S\ref{sec:encoder-determinism} is to
change the admissibility instrument, under which the hashed object
becomes an exactly reproducible byte string and
Proposition~\ref{prop:quantiser} no longer bites.

\subsection{Proof of Theorem~\ref{thm:safe-abort}
            (Safe Abort Under Non-admissibility or Insufficient Certificate)}
\label{app:proof-safe-abort}

\begin{proof}[Proof detail]
The proof is by direct inspection of the control flow of
Algorithm~\ref{alg:certified-commitment-round}. We trace each
abort branch of the H-CSC pipeline in order; in every case the
algorithm returns \Abort\ \emph{before} any commit object is
constructed.

Stage~2 (verdict grouping) partitions $V(r)$ into verdict
groups $\{G_v\}_{v\in\mathcal{V}}$ and selects $v^\star$
deterministically via the \textsc{largest\_verdict} tie-broken
rule (Section~\ref{sec:eval-setup}); if no verdict reaches
$|G_{v}|\ge 2f{+}1$ the algorithm returns
$\Abort(\texttt{verdict\_below\_quorum})$. Stage~3a extracts
the largest angular-threshold component
$\core\gets\textsc{LargestAngularComponent}(G_{v^\star},\thalpha)$;
the semantic branch is unavailable and control passes to
Stage~3b if either $|\core|<2f{+}1$ or the radius predicate
$r^\star\gets\min_{h\in\core}\max_{j\in\core}\angle(\mathbf{e}_h,\mathbf{e}_j)$
exceeds $\thalpha$. Stage~3b is taken iff Stage~3a returned no
admissible core. The Stage~2 gate has already guaranteed
$|G_{v^\star}|\ge 2f{+}1$ (otherwise the round would have
aborted with \texttt{verdict\_below\_quorum}), so Stage~3b's
remaining check is the margin predicate
$\textit{verdict\_margin}\ge\textit{verdict\_margin}^{\min}$;
if it fails the algorithm returns
\Abort(\texttt{v2\_both\_paths\_failed:}\discretionary{}{}{}%
\texttt{semantic\_core\_failed:}\discretionary{}{}{}%
\textit{semantic\_fail\_reason}),
matching Algorithm~\ref{alg:certified-commitment-round}. Stage~4 (semantic branch only) computes
$\ystar\gets\agg(\{\mathbf{e}_i:i\in\core\})$ and returns
$\Abort(\texttt{aggregation\_failed})$ on a degenerate result
(zero norm or numerical failure). Stage~5 collects valid
signatures on the candidate typed digest $\hstar$---either
$\hstar_{\textsc{sem}}=H(\texttt{"semantic\_commit"}$ $\Vert\,Q_\eta(\ystar)\,\Vert\,h_\pi\,\Vert\,r\,\Vert\,v^\star)$
on the semantic branch or
$\hstar_{\textsc{ver}}=H(\texttt{"verdict\_commit"}$ $\Vert\,\textsc{VerdictPayload}(v^\star)$ $\Vert\,h_\pi\,\Vert\,r)$
on the verdict branch (\S\ref{sec:protocol-aggregate-digest});
if fewer than $2f{+}1$ distinct signatures arrive within
$\Delta$ after GST it returns
$\Abort(\texttt{insufficient\_signers},|\Sigma|,2f{+}1)$.

These branches are exhaustive: the algorithm returns either
$\Commit(\hstar,\pi,\cert{\hstar},\core,\mathbf{commit\_type})$
on the semantic or verdict success path, or one of the four
typed \Abort\ values. The control flow never publishes a
\textsc{semantic\_commit} when Stage~3a's admissibility check
fails (Stage~3a guards), never publishes a \textsc{verdict\_commit}
when the verdict-admissibility predicate fails (Stage~3b
guards), and never publishes either commit type when fewer
than $2f{+}1$ distinct signatures arrive (Stage~5 guards
combined with Lemma~\ref{lem:cert-uniqueness}).
\end{proof}

\subsection{Proof of Theorem~\ref{thm:termination}
            (Conditional Termination Under Admissible Round and
             Synchrony)}
\label{app:proof-termination}

\begin{proof}[Proof detail]
We prove two complementary statements. \emph{(Unconditional
termination)} The protocol always terminates with either
$\Commit$ or $\Abort$ within one round of communication after
GST. Stages 1--4 of
Algorithm~\ref{alg:certified-commitment-round} are local
computations on the delivered view $V(r)$ and do not block.
Stage 5 collects valid signatures over the candidate typed
digest $\hstar$ within the bounded delay
$\Delta$ guaranteed by
Assumption~\ref{assn:partial-synchrony}; if fewer than $2f{+}1$
distinct valid signatures arrive, the algorithm returns
$\Abort(\texttt{insufficient\_signers})$. The protocol therefore
exits via $\Commit$ or a typed $\Abort$ in finite time and
makes no liveness assumption on Byzantine behaviour.

\emph{(Conditional commit liveness)} A round terminates with
$\Commit$ if the relevant signer set
$\Sigma_{\textit{src}}$ used at Stage 5 contains at least
$2f{+}1$ honest agents. This is the strict reading of
Theorem~\ref{thm:hcsc-termination} and is strictly stronger
than Assumption~\ref{assn:admissible-round} (which only
requires a $2f{+}1$ admissible core, possibly mixed with
Byzantines). The reason the stronger condition is needed is
that liveness cannot depend on Byzantine signatures: an
adversary may withhold any subset of its signatures, so only
honest agents in $\Sigma_{\textit{src}}$ are guaranteed to
produce a valid signature within $\Delta$ after GST. Under
Assumption~\ref{assn:auth-channels} (authenticated channels)
and Assumption~\ref{assn:partial-synchrony} (partial synchrony),
the at-least-$2f{+}1$ honest agents in $\Sigma_{\textit{src}}$
broadcast and have their signatures delivered within $\Delta$
after GST, so $|\Sigma|\ge 2f{+}1$ without any Byzantine
contribution.

The semantic branch and the verdict branch differ only in
which set is taken as $\Sigma_{\textit{src}}$. On the semantic
branch $\Sigma_{\textit{src}}=\core$ (the within-verdict
admissible core selected at Stage~3a) and the commit-liveness
condition is: $\core$ contains at least $2f{+}1$ honest
agents. On the verdict branch
$\Sigma_{\textit{src}}=G_{v^\star}$ (the entire candidate
verdict group; any $2f{+}1$ distinct valid signatures suffice)
and the commit-liveness condition is: $G_{v^\star}$ contains
at least $2f{+}1$ honest agents. In both cases the condition
matches the authors' reference implementation, which is
planned to accompany a future reproducibility package, and matches
Theorem~\ref{thm:hcsc-termination}'s strengthened
admissibility statement. Together with the unconditional
termination clause this gives commit-or-abort termination
within one round after GST, with $\Commit$ reached when the
honest count in $\Sigma_{\textit{src}}$ meets the quorum and a
typed $\Abort$ otherwise. The conclusion that
Byzantine signatures are not required for liveness matches
Theorem~\ref{thm:hcsc-termination}'s strengthened-admissibility
statement in the main text.
\end{proof}

\subsection{Discussion: Lemma~\ref{lem:rep-boundary}
            (Representation Boundary)}
\label{app:rep-boundary-discussion}

The proof of Lemma~\ref{lem:rep-boundary} given in
Section~\ref{sec:correctness} is direct (two proposals with
identical embeddings cannot be distinguished by any function of
the embeddings) and does not require additional details. We use
this appendix instead to correct an empirical claim made at this
point in the previous version, and to record the measurement that
replaces it.

\paragraph{The corrected claim.}
v1 stated that when a Byzantine proposal $\rho_b$ is engineered so
that $\angle(\cE(\rho_h),\cE(\rho_b))$ is small but nonzero, ``the
latter outcome dominates: Byzantine embeddings that approach an
honest cluster's boundary produce \emph{abort}, not commit,'' and
concluded that the typed safe-abort branch is the defence. Direct
measurement contradicts this. Against verdict-preserving,
rationale-corrupting adversaries --- Mode~B of
Definition~\ref{def:semantic-poisoning}, generated as two attack
families (scope drift and reasoning substitution) at two evasion
levels, $n=10$ tasks, $40$ attacks, $80$ honest proposals,
re-encoded with the same CRSE checkpoint and run through the
unmodified protocol code --- the attack enters the selected core on
$100\%$, $70\%$, $100\%$ and $80\%$ of rounds respectively, and
H-CSC commits on $10/10$ rounds with a contaminated payload in
every case. That is the same round-level outcome as the
verdict-only baseline, which admits $100\%$ by construction.
Abort does not dominate; abort does not occur.

\paragraph{Why the safe-abort branch cannot serve as the defence.}
The natural repair --- tighten $\thalpha$ until the attacks fall
outside the ball --- fails for an architectural reason rather than
a calibration one. Tightening $\thalpha$ causes the semantic path
to fail, which routes the round to the verdict-commit fallback,
which certifies the entire verdict group unfiltered. Sweeping
$\thalpha\in\{0.85,0.75,0.65,0.55,0.45\}$ leaves the maximum harm
of an admitted attack constant at $8.0$ on a blind $0$--$10$
grounding scale, while the attack pass rate rises monotonically
from $85\%$ to $100\%$ and honest retention falls from $0.738$ to
$0.375$. At $\thalpha=0.45$ the semantic-commit rate is $0.00$ and
H-CSC has degenerated exactly into certificate-wrapped majority.
Splitting the admitted attacks by commit type does not rescue the
picture: attacks admitted through the semantic path have maximum
harm $7.0$--$8.0$ (mean $3.79$--$4.18$), against $8.0$ (mean
$4.36$--$4.50$) for the verdict path. The semantic path carries no
harm bound either.

\paragraph{What the boundary actually is.}
The mechanism is a signal-to-noise failure in the honest
population, not a shortfall in encoder training. The honest
proposals' own angular dispersion about their geometric median has
median $0.528$~rad, $90$th percentile $0.992$~rad and $95$th
percentile $1.126$~rad, while the attack displacement is
$0.58$--$0.85$~rad. At $\thalpha=0.65$ the admissibility ball only
just contains the honest population, so it cannot be tightened
without discarding honest agents faster than adversaries. The same
inequality defeats an off-the-shelf entailment model and a
frontier LLM judge on the same data
(\S\ref{sec:eval-r2}), which is why we read the boundary
as a property of the \emph{setting} --- free-text rationales from a
noisy honest population --- rather than of the encoder. The two
constructive consequences we can support are: change the
admissibility instrument (\S\ref{sec:eval-r2}), or change
the output schema so that the honest dispersion collapses
(\S\ref{sec:eval-r6}). Strengthening the encoder is
\emph{not} on the list, and v1's closing sentence recommending it
is withdrawn.

\section{Topology Design-Space Analysis}
\label{appendix:topology}

The original PODC submission of this work made a Jaccard
support-overlap graph and $k$-core peeling the central filtering
mechanism for Byzantine agents. After the post-PODC audit, the
topology mechanism is retained in the codebase as a design-space
option but is \emph{not} part of the main protocol path. We
describe the original motivation, the corrected finding, and the
disposition here for transparency.

\subsection{Original motivation}

Two sources of intuition motivated the topology mechanism.
First, two honest agents that base their proposals on overlapping
evidence cite overlapping \texttt{evidence\_id} sets, while a
Byzantine agent that fabricates citations does not. Building a
Jaccard graph $G_J=(V,E)$ with edges $(i,j)\in E$ iff
\[
\frac{|\rho_i.\textit{evidence\_ids}\cap\rho_j.\textit{evidence\_ids}|}
     {|\rho_i.\textit{evidence\_ids}\cup\rho_j.\textit{evidence\_ids}|}
\;\ge\; \tau_J
\]
should therefore tend to keep honest agents in a single dense
component while pushing Byzantines to isolated nodes or
disconnected fringes. Second, $k$-core peeling of $G_J$ at
$k=f{+}1$ should remove fringe Byzantines while preserving the
honest core.

\subsection{Corrected experimental finding}

The post-PODC audit ran the topology mechanism alongside the
main path on every benchmark. The summary observation, in three
parts, is summarised in the text below; the underlying
topology design-space outputs are recorded in the authors'
internal audit logs, which are planned to accompany a future
reproducibility package.

\begin{enumerate}
\item \textbf{No commit/abort decision change at the operating
point.} On BCS\_v1 across 120 episodes, on the 20-task pilot
across both attack modes, on MVR-50 across both attack modes,
and at every $\thalpha$ in the threshold sensitivity sweep, the
topology gate produces commit/abort decisions \emph{identical}
to the main path on every task / every grid point. The selected
core size shrinks slightly (the gate trims a small number of
nodes), but the round-level decision does not change.

\item \textbf{Topology aborts more in the loose region.} For
$\thalpha\ge 0.65$ on the 20-task pilot, the topology gate
aborts $5$--$20$ percentage points more rounds than the main
path, achieving ``safety by extra aborts''. This does not
reduce \texttt{invalid\_hmaj} below the main path's
in any measured cell except a single MVR-50 cell where the gate
lowers \texttt{invalid\_hmaj} from $0.02$ to $0.00$
at the cost of an extra $4$ aborts.

\item \textbf{$k$-core peeling is empirically a no-op.} The
post-PODC audit reported zero Hard-bin Byzantines on BCS\_v1 ---
Byzantines that are admitted by the angular filter but excluded
by the $k$-core filter --- under the published threshold
settings. We therefore removed $k$-core from the main protocol
path (decision-log entry \texttt{D-11}).
\end{enumerate}

\subsection{Disposition}

The topology mechanism is retained in the reference
implementation behind a parameter flag with values in
$\{$\texttt{None}, \texttt{jaccard}, \texttt{jaccard\_kcore}$\}$.
The main path of this paper uses
\texttt{topology\_gate=None}. Once a future reproducibility
package is made available, the design-space ablation can be reproduced
by setting \texttt{topology\_gate="jaccard"} and re-running the
commitment evaluation scripts; the corresponding CSVs are
regenerated by the paper-assets pipeline included in that
package.

\subsection{What this appendix is not}

This appendix is \emph{not} a competing main-path proposal.
It does \emph{not} support claims of the form ``topology is
necessary'', ``topology improves accuracy'', or ``topology is a
safety gate that catches what the angular filter misses''. The
corrected finding is the opposite: at the chosen operating
point, the topology gate is decision-redundant with the main
path; outside the operating point, it trades commit rate for
extra aborts without improving honest-reference validity. We
report the design-space ablation here for transparency; the
corrected-decision rationale is recorded in the authors'
internal decision log, which is planned to accompany a future
reproducibility package.

\section{Additional Experimental and Implementation Details}
\label{appendix:additional-experiments}

This appendix collects implementation specifics, an audit note,
and a supplementary table that did not fit the \S6 main-text
budget.

\subsection{CRSE Training Configuration}
\label{app:crse-training}

The Certified Robust Semantic Encoder (CRSE) is fine-tuned from
the publicly-released \texttt{intfloat/e5-base} sentence
encoder.
\begin{itemize}
    \item \textbf{Contrastive loss:} a supervised contrastive
    objective with an annealed temperature schedule (start
    $\tau \approx 0.20$, target $\tau \approx 0.12$); the exact
    schedule and the per-stage $\tau$ values are recorded in
    the authors' training configuration, which is planned to
    accompany a future reproducibility package.
    \item \textbf{Hard negatives:} GPT-4o-generated semantic
    poisons (see \S\ref{app:prompts}), plus a per-anchor
    Byzantine-quota cap to prevent over-representation of any
    one attack instance.
    \item \textbf{Optimisation:} AdamW with a small base
    learning rate (~$5\times 10^{-6}$ at the projector and
    sub-$10^{-5}$ at the backbone, with a Byzantine-introduction
    schedule that further reduces the backbone rate in later
    stages), batch size $32$.
    \item \textbf{Infrastructure:} a single training device
    (Apple Silicon / MPS backend) for the trained checkpoint;
    AMP / mixed precision is config-gated.
\end{itemize}
The encoder checkpoint used in Section~\ref{sec:evaluation} is
the single trained CRSE used in all reported evaluations; the
authoritative training hyperparameters are recorded in our
experiment logs, and the resolved per-run config is written
alongside each experiment's outputs. The training
configuration and checkpoint are retained by the authors and are
planned to accompany a future reproducibility package.
Multi-encoder ablation is future work.

\subsection{Semantic Attack Generation Prompts}
\label{app:prompts}

We use GPT-4o (for BCS\_v1) and \texttt{gpt-4o-mini} (for
MVR-50) to generate Byzantine proposals. The system prompt below
is the canonical version used for the four primary attack types
B1--B4 reported in Section~\ref{sec:eval-setup}.

\begin{tcolorbox}[colback=gray!5,colframe=gray!50,title=System
Prompt for Semantic Attacker]
\small
\textbf{Role:} You are a Semantic Byzantine Adversary. Your goal
is to generate a ``poisoned'' version of the input text based on
a specific \textbf{Attack Strategy ID}.

\textbf{General Constraints:}
\begin{itemize}[leftmargin=*]
    \item \textbf{High similarity:} Maximise lexical overlap
    (cosine similarity $>0.9$). Use the same sentence structure
    and entities where possible.
    \item \textbf{Length:} Keep output length within $\pm 5\%$
    of the original.
    \item \textbf{Stealth:} The output must read fluently and
    appear authoritative.
\end{itemize}

\textbf{Attack Strategy Definitions:}
\begin{itemize}[leftmargin=*]
    \item \textbf{B1: Polarity flip.} Invert the truth value or
    stance of the central claim
    (e.g.\ ``X increases Y'' $\to$ ``X reduces Y'').
    \item \textbf{B2: Evidence omission.} Selectively remove
    qualifiers, conditions, or limitations
    (e.g.\ delete ``in most cases'') to make the claim appear
    absolute or biased.
    \item \textbf{B3: Fake causality.} Distort the causal
    structure. Convert correlations into direct causal
    statements or introduce spurious causal links between
    existing entities without changing the topic.
    \item \textbf{B4: On-topic hallucination.} Fabricate
    specific, verifiable details (dates, numbers, institution
    names) that support the topic but are factually unsupported
    by the original text.
\end{itemize}

\textbf{Example (Strategy B1):} \\
\textbf{Input:} ``The encryption protocol guarantees forward
secrecy.'' \\
\textbf{Output:} ``The encryption protocol fails to guarantee
forward secrecy.''
\end{tcolorbox}

\subsection{Bootstrap Confidence Intervals}
\label{app:bootstrap}

The MVR-50 confidence intervals reported in
Tables~\ref{tab:mvr50}, \ref{tab:mvr50-ci} and
Section~\ref{sec:eval-setup} are computed by task-level
non-parametric bootstrap with $10\,000$ resamples and
seed~$42$. Each bootstrap sample is obtained by sampling $50$
tasks (with replacement) from the original $50$-task list and
recomputing the per-method, per-mode metric on the resampled
task set; the reported $95\%$ interval is the
$2.5$th--$97.5$th percentile of the metric distribution. The
bootstrap is reproducible CPU-only; the script and CI report
are retained by the authors and are planned to accompany a future
reproducibility package.

\subsection{Strict-Validation Audit and Evidence-ID Hallucination}
\label{app:audit}

A strict-validation audit applied a task-aware schema check to
every Byzantine record at write time, after the discovery that
the original lenient validator did not catch
\texttt{evidence\_id} fields hallucinated by the LLM. The audit
found $3$ of $100$ rushing Byzantine records on MVR-50 that
cited an \texttt{evidence\_id} not present in the parent task
(none in the static-attack set). All $3$ records were traced
through the commitment evaluation: in each case the certified
pipeline aborted the affected task under both attack modes
(the \texttt{certified\_absorbs\_attack} paired class), so the
hallucinations did not change any commit-level result. The
runner-side validator was hardened post-hoc to reject any
record citing a nonexistent \texttt{evidence\_id} at write time;
the authors' internal audit log records the affected record IDs and
is planned to accompany a future reproducibility package.

\subsection{Logical Certificate, Not Threshold-Signature
            Cryptography}
\label{app:logical-cert}

The certificate $\cert{\hstar}$ in
Algorithm~\ref{alg:certified-commitment-round} is implemented by
a logical signature simulator that preserves per-signer
uniqueness and the $2f{+}1$ distinct-signer threshold. We do
\emph{not} implement, evaluate, or analyse production
threshold-signature primitives; production schemes
(e.g.\ BLS, FROST, threshold ECDSA) are a standard substitution
that does not change the protocol logic but does change the
signature payload size and verification cost. Cryptographic
performance benchmarking is explicit future work.

\section{Reproducibility and Corrections}
\label{appendix:repro}

This appendix does two things that v1 did not. It supplies the
mapping from the labels used in this manuscript to the
identifiers that actually appear in the code and result files,
and it lists the errors we found in v1 while preparing this
revision. We think an errata section makes a preprint more
useful, not less: every item below was found by re-running the
released pipeline against the manuscript, and a reader who
repeats that exercise should not have to rediscover them.

\subsection{Paper labels and code identifiers}
\label{app:decoder}

None of the method labels used in this manuscript
(H-CSC, CSC, M1--M3, V1) appear anywhere in the implementation
or in the released result files, which use a different naming
scheme. Table~\ref{tab:label-decoder} is the decoder.

\begin{table*}[t]
\centering
\caption{Paper label to code-artifact decoder. The
\texttt{method} and \texttt{param\_\allowbreak setting} columns are the
literal values in the released
\texttt{coverage\_\allowbreak recovery\_\allowbreak *.csv} files; a reader reproducing
Table~\ref{tab:hcsc-main} must filter on the pair.}
\label{tab:label-decoder}
\small
\begin{tabularx}{\linewidth}{@{}>{\raggedright\arraybackslash}X >{\raggedright\arraybackslash}X >{\raggedright\arraybackslash}X >{\raggedright\arraybackslash}X@{}}
\toprule
\textbf{Label in this paper} & \textbf{CSV \texttt{method}} &
\textbf{CSV \texttt{param\_\allowbreak setting}} & \textbf{Effective parameters} \\
\midrule
H-CSC (main path)
  & \texttt{V2\_\allowbreak hierarchical}
  & \texttt{theta\_\allowbreak 0.65\_\allowbreak \_\allowbreak fb\_\allowbreak margin\_\allowbreak 1}
  & $\thalpha = 0.65$~rad, fallback $=$ \textsc{margin\_\allowbreak 1}, $\eta = 10^{-6}$, $\eps = 0.70$ \\
CSC (strict-semantic ablation)
  & \texttt{B0\_\allowbreak certified\_\allowbreak main}
  & \texttt{default}
  & $\thalpha = 0.55$~rad, no fallback \\
M1 (abstaining majority)
  & \texttt{B1\_\allowbreak abstaining\_\allowbreak majority}
  & \texttt{default}
  & --- \\
M2 (margin majority)
  & \texttt{B2\_\allowbreak margin\_\allowbreak majority}
  & \texttt{margin\_\allowbreak 1}
  & --- \\
M3 (certificate $+$ majority)
  & \texttt{B3\_\allowbreak certificate\_\allowbreak wrapped\_\allowbreak majority}
  & \texttt{margin\_\allowbreak 1}
  & --- \\
V1 (verdict-conditioned semantic)
  & \texttt{V1\_\allowbreak verdict\_\allowbreak conditioned}
  & \texttt{theta\_\allowbreak 0.7\_\allowbreak \_\allowbreak cand\_\allowbreak largest\_\allowbreak verdict}
  & $\thalpha = 0.70$~rad \emph{(not the main-path $0.65$; see Erratum~E3)} \\
\bottomrule
\end{tabularx}
\end{table*}

\paragraph{A label collision the reader must know about.}
The tokens \texttt{B1}--\texttt{B4} are overloaded in v1 and,
for continuity, in this version. In
\S\ref{sec:eval-setup} and Table~\ref{tab:mvr50-ci} they denote
\emph{baseline methods} (\texttt{B0}--\texttt{B3} in the code
naming above). In the attack-suite description of
\S\ref{sec:eval-setup} and in Appendix~\ref{app:prompts} they
denote \emph{attack types}
(\textsc{B1}~polarity\_flip, \textsc{B2}~evidence\_omission,
\textsc{B3}~false\_causality,
\textsc{B4}~on\_topic\_hallucination). The new
verdict-preserving families introduced in
\S\ref{sec:eval-r2} continue the attack numbering as
\textsc{B6}~scope\_shift and
\textsc{B7}~reasoning\_substitution. Where the two could be
confused we now write ``method \texttt{B3}'' or ``attack
\textsc{B3}'' explicitly.

\paragraph{Encoder identity.}
The single CRSE checkpoint used for every embedding reported in
this paper is the file \texttt{500pt\_best\_model.pt}
($442$~MB), fine-tuned from a \texttt{intfloat/e5-base}-family
sentence encoder with a $512$-dimensional projection head
(Appendix~\ref{app:crse-training}). The projection dimension is
a useful integrity check on any reproduction: the released
embedding cache is $512$-dimensional, so a run that silently
fell back to the \texttt{all-MiniLM-L6-v2} default would produce
$384$-dimensional vectors and be immediately visible.
\emph{Author to confirm before posting:} v1's appendix names the
backbone \texttt{intfloat/e5-base} while the code default is
\texttt{intfloat/e5-base-v2} and the fallback constant names
\texttt{all-MiniLM-L6-v2}. These are three different models and
exactly one versioned name should survive here.

\paragraph{Quantiser and digest binding.}
The lattice quantiser $\quant{\cdot}$ uses step
$\eta = 10^{-6}$ in every reported configuration. The digest is
computed over a canonical serialisation of the quantised
aggregate together with the protocol parameters and round
identifier, as specified in \S\ref{sec:protocol}; the step value
is part of the bound parameter set, so two signers using
different steps produce different digests by design rather than
disagreeing silently.

\paragraph{Canonical proposal text.}
The encoder input is a deterministic five-field serialisation of
the structured proposal:

\begin{tcolorbox}[colback=gray!5,colframe=gray!50,
                  title=Canonical proposal text]
\small
\begin{verbatim}
[VERDICT: <verdict, lowercased, stripped>]
[CONFIDENCE: <float, 2 decimals>]
[EVIDENCE: <ids, sorted, comma-space joined>]
[RATIONALE: <whitespace-collapsed>]
[CLAIM: <whitespace-collapsed headline claim>]
\end{verbatim}
\end{tcolorbox}

\noindent
Fields are joined by single newlines. Evidence identifiers are
sorted lexicographically; an empty evidence list yields
\texttt{[EVIDENCE: ]}; confidence is rounded to two decimals;
rationale and claim have internal whitespace collapsed. A
missing required field raises rather than defaulting. This
serialisation, not the raw model output, is what
Assumption~\ref{assn:deterministic-cE} quantifies over.

\paragraph{Verified reproduction facts.}
The following were re-derived from the released artefacts during
this revision and answer questions raised in review.
\begin{itemize}[leftmargin=*]
  \item \textbf{Encoding reproduces the frozen cache.}
  Re-encoding the MVR-50 canonical texts with the checkpoint
  above and comparing against the released cache gives a minimum
  cosine similarity of $0.99999990$ and a maximum element-wise
  difference of $3.6\times10^{-7}$. This is reproduction to
  within the quantiser step, not bitwise identity
  (\S\ref{sec:limitations}, item~9).
  \item \textbf{The protocol layer is cheap and CPU-only.}
  Re-running all $6\,000$
  (method $\times$ parameter setting $\times$ mode $\times$
  task) combinations from the cached embeddings takes $1.3$
  seconds on CPU, with no API access and no GPU.
  \item \textbf{Byzantine agents do sign.} The experiments use
  \texttt{signer\_policy = "selected\_core"}: whichever agents
  are selected into the core are the signers, Byzantine agents
  included. The certificate threshold is therefore exercised
  adversarially rather than over an honest-only signer set.
  \item \textbf{\texttt{insufficient\_signers} never fires.}
  The abort reason \texttt{insufficient\_signers} occurs $0$
  times across all $6\,000$ per-task records. Reviewers asked
  for this figure even if zero; it is zero.
  \item \textbf{Pilot and test sets are disjoint.} The $20$-task
  calibration pilot used to set $\thalpha$ and the fallback rule
  shares $0$ tasks with MVR-50 and $0$ with MVR-100.
  \item \textbf{No detectable encoder contamination.} The CRSE
  training set is $488$ Wikipedia \texttt{science\_tech}
  anchors. Exactly $1$ of $488$ titles contains a
  climate-related term; the maximum token Jaccard similarity
  between any training anchor and any MVR claim is $0.176$
  (mean $0.059$); the title intersection with the evaluation
  claims is empty.
\end{itemize}

\paragraph{Artefact gap acknowledged.}
The result files released with v1 do not persist the quantised
aggregate $\widetilde{\mathbf{y}}$, the digest $\hstar$, or core
membership for any round: the digest routine computed them and
returned only the hex digest. The consequence is that v1's
headline claim --- that $74\% / 72\%$ of rounds carry an
embedding-backed digest --- was not independently checkable from
the released files, however true it is of the running code. The
commitment object is now persisted per round.

\subsection{Errata for version 1}
\label{app:errata}

Each item states what v1 said, what is true, and what it
changes. Items E1--E4 affect claims in the main text; E5--E6 are
proof-hygiene items; E7 is bibliographic.

\paragraph{E1 --- The one invalid commit was mis-typed.}
\emph{v1 said:} in \S\ref{sec:eval-setup}, that method
\texttt{B3} is safer than H-CSC ``by a single task (one MVR-50
task on which H-CSC emits an invalid \textsc{verdict\_commit}).''
\emph{Correction:} the event was a \textsc{semantic\_commit}, not
a verdict\_commit. On the affected static-mode task the honest
agents split $5$ \texttt{insufficient} / $3$ \texttt{support}
(honest plurality \texttt{insufficient}, matching gold), both
Byzantine agents voted \texttt{support}, and the delivered view
was $5$--$5$ with margin exactly $0$. The deterministic
tie-break (\texttt{support} $>$ \texttt{refute} $>$
\texttt{insufficient}) elected $v^\star = \texttt{support}$; the
support group's core contained $3$ honest and $2$ Byzantine
members, all inside $\thalpha$, giving
$|\core| = 5 = 2f{+}1$ and a Byzantine core fraction of $0.40$.
\emph{What it changes:} v1's implicit position that the semantic
branch is the safer branch is unsupported by its own data. The
event is an exploitable tie-break capture, analysed in
\S\ref{sec:margin-gate}.

\paragraph{E2 --- The attack-composition claim is wrong for
rushing mode.}
\emph{v1 said:} the four attack types are ``balanced
$25/25/25/25$ across the $100$ Byzantine slots'' and each slot
is paired across static and rushing.
\emph{Correction:} static is indeed $25/25/25/25$; rushing is
$29/9/41/21$. The attack selector chooses a target verdict
adaptively from the honest plurality it observes in the round,
and under rushing that plurality differs from the static
pre-assignment, so composition drifts. The mechanism is
reasonable; the description was not.
\emph{What it changes:} every static-versus-rushing difference
in Table~\ref{tab:hcsc-main} confounds attack timing with attack
composition and should not be read as a timing effect.

\paragraph{E3 --- The V1 baseline used a wider threshold than
the main path.}
\emph{v1 said:} nothing --- Table~\ref{tab:hcsc-main}'s V1 row
is described as ``verdict-conditioned semantic'' with no
threshold stated.
\emph{Correction:} that row was produced at
$\thalpha = 0.70$~rad, while the H-CSC main path runs at
$\thalpha = 0.65$~rad. The baseline was therefore evaluated
under a strictly more permissive admissibility predicate than
the proposed method, and the manuscript did not say so.
\emph{What it changes:} the V1 row is not a matched comparison
as printed. It is re-reported at matched $\thalpha$ in this
version, and the parameter appears in
Table~\ref{tab:label-decoder} so the discrepancy is visible from
the artefact alone.

\paragraph{E4 --- The semantic branch never enforced a verdict
margin.}
\emph{v1 said:} that the main path operates at
$\textit{verdict\_margin}^{\min} = 1$.
\emph{Correction:} that is true of the verdict-commit fallback
only. The semantic branch tested the admissibility predicate
alone, and the candidate-group rule tested only whether the top
verdict count reached quorum, so the principal branch ran at an
effective margin of $0$.
\emph{What it changes:} it is what makes E1 reachable, and it
means v1's operating-point description overstated the
constraint the protocol enforced. The two-path margin gate and
its coverage cost are given in \S\ref{sec:margin-gate}.

\paragraph{E5 --- Core extraction returns components, not
balls.}
\emph{v1 said:} Definition~\ref{def:admissible-round} specifies
a centre node and a radius, i.e.\ a ball.
\emph{Correction:} the implementation returns a connected
component of the overlap graph. Of $73$ extracted cores, $5$ are
not balls under the stated definition.
\emph{What it changes:} less than one might fear. The largest
observed core diameter is $0.964$~rad, well inside the
$2\thalpha = 1.30$~rad bound, so Theorem~\ref{thm:validity}'s
guarantee is not violated on any measured round, and re-running
with a strict ball criterion leaves the reported numbers
essentially unchanged. We record it because the definition and
the code should agree.

\paragraph{E6 --- Theorem~\ref{thm:hcsc-verdict-validity} had no
sufficiency proof.}
\emph{v1 said:} in the appendix, that ``the verdict-validity
argument is fully developed in Section~\ref{sec:correctness}
$\ldots$ and not repeated here.''
\emph{Correction:} the main-text sketch derives that
$\textit{verdict\_margin}^{\min} > f$ is \emph{necessary} for
deterministic honest-reference validity; it does not establish
that the condition is sufficient, and no appendix proof exists.
\emph{What it changes:} the theorem as stated in v1 was not
proved. This version supplies the complete argument, together
with the strictly stronger honest-majority result under the
$|\core| > (n{+}f)/2$ rule, in \S\ref{sec:correctness}.

\paragraph{E7 --- Bibliography hygiene and a substantive
omission.}
v1's \texttt{refs.bib} contains a duplicate \texttt{lewis2020rag}
key, a stale \texttt{dwork1988\_DELETED} entry, and $44$ of $75$
entries that are never cited. The substantive item is not
hygiene: CP-WBFT~\cite{zheng2026cpwbft} was peer-reviewed and
published before v1 appeared and was not cited, and the two
canonical attacks on robust
aggregation~\cite{baruch2019little, xie2020fall} were present in
the bibliography and never cited despite being directly relevant
to this paper's mechanism. Both are addressed in
\S\ref{sec:related_work}.

\paragraph{Scope of this errata.}
The list covers what a full re-run of the released pipeline
against the manuscript surfaced. It is not a guarantee that
nothing else is wrong. Items still under verification at the
time of writing --- among them the episode count reported for
the BCS\_v1 AUC and one sentence in the strict-validation audit
of Appendix~\ref{app:audit} --- will be added here in a later
version rather than corrected silently.

\section{Material Moved from the Main Text}
\label{app:moved}
These subsections were moved out of the body for length; they are unchanged.

\subsection{Protocol-Level Results, Recast}
\label{sec:eval-retained}

The v1 protocol-level measurements are unchanged and remain
correct. What changes is what they are evidence \emph{for}.

\paragraph{Main table.}
Table~\ref{tab:mvr50} reports the per-method, per-mode point
estimates and Table~\ref{tab:mvr50-ci} the $95\%$ task-level
bootstrap CIs ($10\,000$ resamples, seed $42$).
Figure~\ref{fig:hcsc-tradeoff} plots the same CIs. H-CSC
commits $0.90 / 0.92$ (static / rushing) at
$\INVhmaj = 0.02 / 0.00$ and $\VCC = 0.88 / 0.92$, with
$\SEM = 0.74 / 0.72$. M3, the certificate-wrapped margin-gated
majority, commits $0.88 / 0.92$ at $\INVhmaj = 0.00 / 0.00$.
Reading Table~\ref{tab:mvr50} requires two caveats established
above: static-versus-rushing differences confound timing with
attack composition (\S\ref{sec:eval-attack-audit}), and the
$\byzcore$ column measures a quantity the attack suite could not
put under pressure.

\begin{table*}[t]
\centering
\caption{MVR-50 main results ($n=10$, $f=2$, $50$
Climate-FEVER tasks; parameters in \S\ref{sec:eval-setup}).
$\CR$: commit rate. $\SEM$: semantic-commit coverage over all
rounds (verdict-only methods give $\SEM=0$ by construction;
verdict-commit coverage $=\CR-\SEM$, abort $=1-\CR$).
$\INVgold$/$\INVhmaj$: the two validity-mismatch rates
(Definition~\ref{def:two-validities}; lower better).
$\byzcore$: mean Byzantine fraction in the selected core.
\emph{cert?}: emits a $2f{+}1$-quorum certificate.
\textbf{Two caveats from \S\ref{sec:eval-attack-audit}:} the
rushing columns use a different attack-type composition from
the static columns ($29/9/41/21$ vs.\ $25/25/25/25$), so
static$\leftrightarrow$rushing differences are not clean
timing effects; and the Byzantine agents in this suite always
target a non-plurality verdict, so $\byzcore$ measures an
adversary that is structurally ineligible for the semantic core.
$^\ddagger$The V1 row was run at $\thalpha = 0.70$, not
$0.65$; v1 of this paper did not state this.}
\label{tab:hcsc-main}
\label{tab:mvr50}
\small
\renewcommand{\arraystretch}{0.95}
\setlength{\tabcolsep}{3.5pt}
\resizebox{\ifdim\width>\linewidth\linewidth\else\width\fi}{!}{%
\begin{tabular}{l c c c c c c c c c c c}
\toprule
& \multicolumn{5}{c}{\textbf{static}} & \multicolumn{5}{c}{\textbf{rushing}} & \\
\cmidrule(lr){2-6}\cmidrule(lr){7-11}
method & \CR & \SEM & \INVgold & \INVhmaj & \byzcore & \CR & \SEM & \INVgold & \INVhmaj & \byzcore & cert? \\
\midrule
Majority vote                          & 0.90 & 0.00 & 0.22 & 0.00 & 0.054 & 0.94 & 0.00 & 0.24 & 0.00 & 0.000 & no \\
Confidence-weighted vote               & 1.00 & 0.00 & 0.28 & \textbf{0.12} & 0.083 & 1.00 & 0.00 & 0.28 & \textbf{0.12} & 0.023 & no \\
All-nodes GM                           & 1.00 & 0.00 & 0.28 & 0.04 & 0.200 & 1.00 & 0.00 & 0.26 & 0.06 & 0.200 & no \\
Angular-threshold GM                   & 1.00 & 0.00 & 0.28 & 0.08 & 0.035 & 1.00 & 0.00 & 0.28 & 0.08 & 0.006 & no \\
M1: abstaining majority                & 0.94 & 0.00 & 0.24 & 0.06 & 0.111 & 0.94 & 0.00 & 0.24 & 0.02 & 0.011 & yes \\
M2: margin majority ($\marginmin{=}1$) & 0.88 & 0.00 & 0.20 & 0.00 & 0.090 & 0.92 & 0.00 & 0.22 & 0.00 & 0.000 & yes \\
M3: certificate\,$+$\,majority         & 0.88 & 0.00 & 0.20 & 0.00 & 0.088 & 0.92 & 0.00 & 0.22 & 0.00 & 0.000 & \textbf{yes} \\
V1: verdict-conditioned semantic$^\ddagger$ & 0.82 & 0.82 & 0.22 & 0.02 & 0.034 & 0.80 & 0.80 & 0.22 & 0.00 & 0.000 & yes \\
H-CSC (this paper)                     & 0.90 & 0.74 & 0.22 & 0.02 & 0.069 & 0.92 & 0.72 & 0.22 & 0.00 & 0.000 & \textbf{yes} \\
CSC: strict-semantic                   & 0.54 & 0.54 & 0.16 & 0.02 & 0.021 & 0.48 & 0.48 & 0.12 & 0.02 & 0.006 & yes \\
\bottomrule
\end{tabular}}
\end{table*}

\begin{table}[t]
\centering
\caption{Bootstrap $95\%$ CIs on the headline MVR-50 metrics
($10\,000$ task-level resamples, seed 42).}
\label{tab:mvr50-ci}
\small
\begin{tabular}{l l c c}
\toprule
method / mode & metric & point & 95\% CI \\
\midrule
H-CSC static  & $\CR$ & 0.90 & $[0.80, 0.98]$ \\
H-CSC rushing & $\CR$ & 0.92 & $[0.84, 0.98]$ \\
M3 static     & $\CR$ & 0.88 & $[0.78, 0.96]$ \\
M3 rushing    & $\CR$ & 0.92 & $[0.84, 0.98]$ \\
M2 static     & $\CR$ & 0.88 & $[0.78, 0.96]$ \\
M2 rushing    & $\CR$ & 0.92 & $[0.84, 0.98]$ \\
M1 static     & $\CR$ & 0.94 & $[0.86, 1.00]$ \\
M1 rushing    & $\CR$ & 0.94 & $[0.86, 1.00]$ \\
CSC static    & $\CR$ & 0.54 & $[0.40, 0.68]$ \\
CSC rushing   & $\CR$ & 0.48 & $[0.34, 0.62]$ \\
\midrule
H-CSC static  & $\INVhmaj$ & 0.02 & $[0.00, 0.06]$ \\
H-CSC rushing & $\INVhmaj$ & 0.00 & $[0.00, 0.00]$ \\
M3 static     & $\INVhmaj$ & 0.00 & $[0.00, 0.00]$ \\
M3 rushing    & $\INVhmaj$ & 0.00 & $[0.00, 0.00]$ \\
M2 static     & $\INVhmaj$ & 0.00 & $[0.00, 0.00]$ \\
M2 rushing    & $\INVhmaj$ & 0.00 & $[0.00, 0.00]$ \\
M1 static     & $\INVhmaj$ & 0.06 & $[0.00, 0.14]$ \\
M1 rushing    & $\INVhmaj$ & 0.02 & $[0.00, 0.06]$ \\
CWV static    & $\INVhmaj$ & 0.12 & $[0.04, 0.22]$ \\
CWV rushing   & $\INVhmaj$ & 0.12 & $[0.04, 0.22]$ \\
\bottomrule
\end{tabular}
\end{table}

\begin{figure*}[t]
\centering
\includegraphics[width=0.92\textwidth]{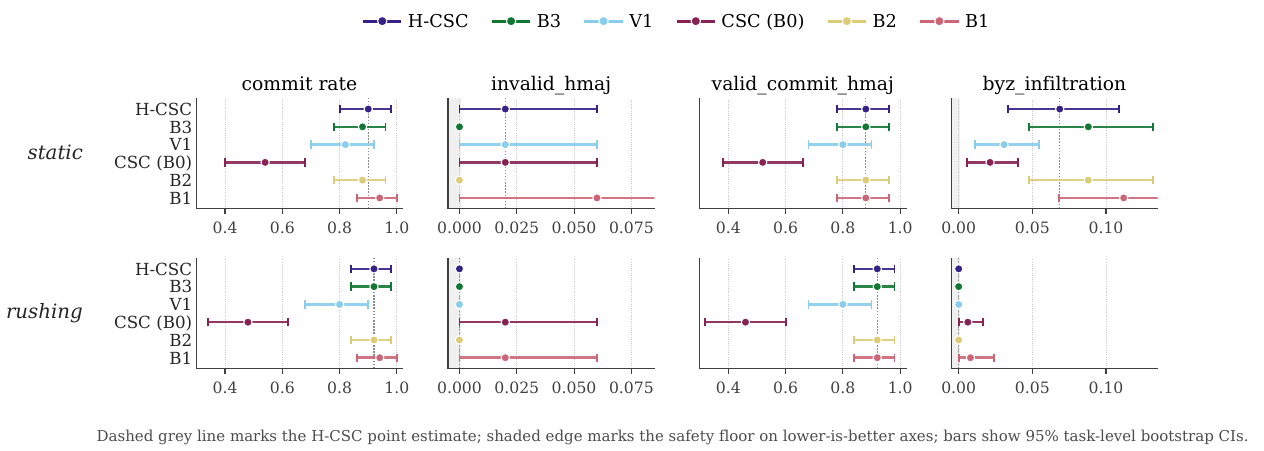}
\caption{MVR-50 strip plot for the six
\emph{certificate-emitting} methods (H-CSC, M3, V1, CSC, M2,
M1): one method$\times$mode per row, one metric per column.
Bars are $95\%$ task-level bootstrap CIs ($10\,000$ resamples,
seed $42$). Naive baselines (no certificate) omitted; see
Table~\ref{tab:mvr50}.}
\label{fig:hcsc-tradeoff}
\end{figure*}

\paragraph{The containment lemma, empirically.}
Lemma~\ref{lem:containment} predicts that at matched
deterministic guarantees the semantic condition is strictly
stronger than the verdict condition, so no coverage separation
is available. The data agree to the task.
\emph{(i)} At the shipped operating point, H-CSC and M3 produce
different decisions on $0/50$ tasks under rushing and $1/50$
under static --- the single static disagreement being the
tie-break capture round of \S\ref{sec:margin-gate}.
\emph{(ii)} At the deterministic operating point demanded by the
lemma ($|\core| > (n+f)/2$), the two methods' commit sets are
\emph{identical task by task}, with $0$ tasks on which one
commits and the other does not, and both achieve coverage
$0.72 / 0.64$.
\emph{(iii)} Under the strict-honest-majority rule of
Theorem~\ref{thm:det-semantic-validity}, semantic coverage
is $0.34 / 0.28$ at $\INVhmaj = 0.00$ --- the price of the
stronger guarantee, paid in coverage exactly as the lemma says
it must be.
\textbf{v1 reported the H-CSC/M3 match as an empirical
coincidence within overlapping CIs. It is a theorem.} The
correct conclusion is not that H-CSC is as good as M3 at the
same job; it is that the coverage axis was never the axis on
which a semantic layer could separate.

\paragraph{Coverage recovery over strict CSC --- the result that
survives.}
Against the strict-semantic configuration of the same framework
($\thalpha=0.55$, no fallback, $\CR = 0.54 / 0.48$ at
$\INVhmaj = 0.02 / 0.02$), the H-CSC main path recovers
coverage by $+0.36$ static / $+0.44$ rushing ($\VCC$ lift
$+0.353 / +0.450$) while holding the $\le 0.04$ safety floor.
The bootstrap CIs are \emph{disjoint}: $[0.80,0.98]$ and
$[0.84,0.98]$ against $[0.40,0.68]$ and $[0.34,0.62]$. Paired
absorption improves from $23/50$ to $44/50$ tasks. This is the
one headline effect in the paper that is large relative to its
uncertainty, and no part of it depends on the encoder being a
good detector: it is produced by (i) conditioning core
extraction on the verdict, so honest agents who agree on the
verdict but phrase rationales differently stay in one core, and
(ii) the verdict-level fallback. \textbf{This result now
carries a different burden than in v1: it justifies the
\emph{hierarchy} --- the typed layering of semantic finality
over verdict finality with a typed abort underneath --- and it
justifies nothing about the instrument used inside it.} It is
also, by \S\ref{sec:eval-r5}, the mechanism that destroys any
damage bound; the hierarchy buys coverage and pays for it in
filtering guarantees.

\paragraph{Cost of the margin-gate fix.}
\S\ref{sec:margin-gate} reports a tie-break capture
vulnerability we found in our own protocol: the semantic branch
never checks the verdict margin, so on an exact tie the
deterministic tie-break can hand $v^\star$ to a Byzantine-chosen
verdict and the semantic path then certifies it. The four-line
fix applies the margin condition to \emph{both} paths.
Table~\ref{tab:margin-gate} prices it. At $\marginmin = 1$ the
fix eliminates the invalid commit at a cost of $0.02$ static
coverage and none under rushing; we recommend it as the default.
$\marginmin = f{+}1 = 3$, the value
Theorem~\ref{thm:hcsc-verdict-validity} requires for
unconditional deterministic verdict validity, costs $0.18$ /
$0.28$ coverage --- which is the honest price of that theorem's
hypothesis and the reason the shipped configuration did not
meet it.

\begin{table}[t]
\centering
\caption{Cost of gating \emph{both} decision paths on the
verdict margin (\S\ref{sec:margin-gate}). $\marginmin=0$
is the shipped v1 behaviour on the semantic branch. MVR-50,
$n=10$, $f=2$.}
\label{tab:margin-gate}
\small
\setlength{\tabcolsep}{6pt}
\begin{tabular}{@{}c cc cc@{}}
\toprule
& \multicolumn{2}{c}{static} & \multicolumn{2}{c}{rushing} \\
\cmidrule(lr){2-3}\cmidrule(lr){4-5}
$\marginmin$ & $\CR$ & $\INVhmaj$ & $\CR$ & $\INVhmaj$ \\
\midrule
0 (as shipped) & 0.90 & 0.02 & 0.92 & 0.00 \\
\textbf{1 (recommended)} & \textbf{0.88} & \textbf{0.00} & \textbf{0.92} & \textbf{0.00} \\
2 & 0.86 & 0.00 & 0.84 & 0.00 \\
3 ($=f{+}1$) & 0.72 & 0.00 & 0.64 & 0.00 \\
\bottomrule
\end{tabular}
\end{table}

\paragraph{Core-extractor fidelity.}
The implementation extracts the largest angular \emph{connected
component} where Definition~\ref{def:admissible-round}
specifies a \emph{ball}. Auditing all $73$ extracted cores:
$5/73$ are not balls, and the largest observed core diameter is
$0.964$~rad, below $2\thalpha = 1.30$, so
Theorem~\ref{thm:hcsc-semantic-validity}'s bound is not
violated on any measured round. Substituting a strict ball
criterion changes the headline numbers negligibly. The
discrepancy is nonetheless real, it is the reason \textsc{B7}
enters the core more often than it enters the ball
(\S\ref{sec:eval-r1}), and the definition --- not the
implementation --- is the object we would keep.

\paragraph{Cross-model robustness.}
Re-running the entire generation pipeline at MVR-100 under four
agent LLMs spanning two vendors and two size regimes
(\texttt{gpt-4o-mini}, \texttt{gpt-4o},
\texttt{claude-sonnet-4-5}, \texttt{llama-3.3-70b-instruct}),
with tasks, prompts, encoder and protocol parameters held
fixed, keeps $\CR \in [0.92, 0.97]$, $\INVhmaj \in [0.00,
0.03]$ (worst-case CI upper bound $0.07$), and semantic-commit
share $\ge 0.68$ on every (model, mode) cell
(Table~\ref{tab:cross-model}). The typed-decision logic is a
deterministic function of the finality-signal vector $\Psi(r)$
(Definition~\ref{def:finality-signal}) and not of the agent
LLM's generation distribution, and the measurement is
consistent with that. This is a \emph{between-run} check: each
run uses a homogeneous agent population. Within-round vendor
heterogeneity remains future work. Note also what this check
does \emph{not} establish: it varies the agent LLM and holds
the encoder fixed, so it says nothing about the encoder
robustness question that \S\ref{sec:eval-r2} answers directly.

\begin{table}[t]
\centering
\caption{H-CSC headline numbers under four agent LLMs on the
same MVR-100 task base ($95\%$ task-level bootstrap CIs;
$n_{\textit{tasks}}=100$, $10\,000$ resamples, seed 42).
Protocol parameters, tasks and encoder are identical across
rows; only the agent LLM changes.}
\label{tab:cross-model}
\scriptsize
\setlength{\tabcolsep}{3pt}
\begin{tabular}{@{}l l l l c@{}}
\toprule
agent LLM & mode & $\CR$ [95\% CI] & $\INVhmaj$ [95\% CI] & $\SEM$ share \\
\midrule
\multirow{2}{*}{\texttt{gpt-4o-mini}}
  & static  & $0.93\,[0.88, 0.98]$ & $0.00\,[0.00, 0.00]$ & 0.82 \\
  & rushing & $0.93\,[0.87, 0.97]$ & $0.01\,[0.00, 0.03]$ & 0.79 \\
\addlinespace
\multirow{2}{*}{\texttt{gpt-4o}}
  & static  & $0.96\,[0.92, 0.99]$ & $0.00\,[0.00, 0.00]$ & 0.81 \\
  & rushing & $0.95\,[0.90, 0.99]$ & $0.01\,[0.00, 0.03]$ & 0.80 \\
\addlinespace
\multirow{2}{*}{\texttt{claude-sonnet-4-5}}
  & static  & $0.92\,[0.86, 0.97]$ & $0.01\,[0.00, 0.03]$ & 0.85 \\
  & rushing & $0.95\,[0.90, 0.99]$ & $0.01\,[0.00, 0.03]$ & 0.88 \\
\addlinespace
\multirow{2}{*}{\texttt{llama-3.3-70b}}
  & static  & $0.97\,[0.93, 1.00]$ & $0.03\,[0.00, 0.07]$ & 0.71 \\
  & rushing & $0.93\,[0.88, 0.98]$ & $0.01\,[0.00, 0.03]$ & 0.68 \\
\bottomrule
\end{tabular}
\smallskip\par
\noindent\footnotesize First-pass \texttt{schema\_valid}
(honest proposals, $n_{\text{honest}}=800$ per model):
\texttt{gpt-4o-mini} $99.8\%$, \texttt{gpt-4o} $100\%$,
\texttt{claude-sonnet-4-5} $100\%$, \texttt{llama-3.3-70b}
$98.4\%$ ($13$ of $800$ honest proposals failed JSON schema on
first attempt; those rounds still met the $2f{+}1$ quorum on the
remaining honest proposals).
\end{table}

\paragraph{Threshold sensitivity, and why the plateau is not
reassuring.}
Sweeping
$\thalpha\in\{0.25,\dots,1.00\}$ rad on the strict-semantic
ablation over the 20-task calibration pilot (8 grid points
$\times$ 2 modes $\times$ 2 methods $\times$ 20 tasks $=$ 640
trials) yields three regimes (Table~\ref{tab:theta}): unusable
below $0.45$ ($\CR \le 0.15$); a single stable grid point at
$\thalpha = 0.55$ meeting $\CR \ge 0.65$, $\INVhmaj = 0$ and
$\byzcore \le 0.04$ under both modes; and a loose regime above
$0.65$ where coverage climbs toward $1.00$ while $\byzcore$
grows monotonically to $0.123 / 0.142$ at $\thalpha=1.00$. The
plateau is $\approx 6^\circ$ wide and benchmark-specific; we do
not claim a universal $\thalpha$. In light of
\S\ref{sec:eval-r5} we add a stronger caveat: this sweep is
calibrated against Family~I. Under Family~II the same sweep is
flat in damage (Table~\ref{tab:frontier}), so a well-chosen
$\thalpha$ is not evidence that the threshold is doing safety
work.

\begin{table}[t]
\centering
\caption{Threshold sensitivity on the 20-task calibration pilot
(disjoint from MVR-50/100), \texttt{strict-semantic CSC} only.
Displayed grid $\thalpha\in\{0.45,0.55,0.65,0.85,1.00\}$; full
8-point grid in
Appendix~\ref{appendix:additional-experiments}.}
\label{tab:theta}
\footnotesize
\setlength{\tabcolsep}{4pt}
\begin{tabular}{c c c c c c l}
\toprule
$\thalpha$ / deg & mode & $\CR$ & $\INVgold$ & $\INVhmaj$ & $\byzcore$ & regime \\
\midrule
0.45 / 25.78 & static  & 0.15 & 0.10 & 0.00 & 0.00 & strict \\
0.45 / 25.78 & rushing & 0.15 & 0.10 & 0.00 & 0.00 & strict \\
\textbf{0.55 / 31.51} & static  & \textbf{0.65} & 0.20 & \textbf{0.00} & 0.04 & \textbf{stable} \\
\textbf{0.55 / 31.51} & rushing & \textbf{0.65} & 0.20 & \textbf{0.00} & 0.02 & \textbf{stable} \\
0.65 / 37.24 & static  & 0.95 & 0.35 & 0.00 & 0.07 & loose \\
0.65 / 37.24 & rushing & 1.00 & 0.40 & 0.05 & 0.04 & loose \\
0.85 / 48.70 & static  & 1.00 & 0.40 & 0.05 & 0.09 & loose \\
0.85 / 48.70 & rushing & 1.00 & 0.40 & 0.05 & 0.08 & loose \\
1.00 / 57.30 & static  & 1.00 & 0.40 & 0.05 & 0.12 & loose \\
1.00 / 57.30 & rushing & 1.00 & 0.40 & 0.05 & 0.14 & loose \\
\bottomrule
\end{tabular}
\end{table}

\paragraph{Topology gate.}
The Jaccard support-overlap gate produces decisions identical
to the main path on BCS-v1, aborts more without reducing
$\INVhmaj$ on MVR-50 static, and aborts more while shaving
$0.02$ off $\INVhmaj$ on MVR-50 rushing only; outside the chosen
$\thalpha$ it aborts strictly more. It is a design-space
transparency artefact, not part of the contribution, and is not
required for any guarantee of \S\ref{sec:correctness}. Full
ablation in Appendix~\ref{appendix:topology}. We note in
passing that the evidence-id overlap signal this gate is built
on is the same signal that scores at chance in
Table~\ref{tab:instrument}.


\subsection{Safe abort and termination}
\label{sec:abort-termination}

The two remaining v1 theorems are properties of the control flow
rather than of a published commit, and are unchanged.

\begin{theorem}[Safe abort under no-admissibility or insufficient
certificate]
\label{thm:hcsc-safe-abort}
\label{thm:safe-abort} 
If the round is neither semantic-admissible nor
verdict-admissible, or if fewer than $2f{+}1$ distinct
signatures arrive on the candidate digest, the protocol
publishes $\mathcal{O}^{\textsc{abort}}$ with a typed reason
from the closed set of
Definition~\ref{def:csc-object}. In particular, the protocol
\emph{never} publishes a \textsc{semantic\_commit} when the
within-verdict semantic core is not admissible, and
\emph{never} publishes a \textsc{verdict\_commit} when
$|G_{v^\star}| < 2f{+}1$ or
$\textit{verdict\_margin} < \textit{verdict\_margin}^{\min}$.
\end{theorem}

\begin{proof}[Proof sketch]
By inspection of the control flow of
Algorithm~\ref{alg:hcsc-round}: each abort branch returns
$\mathcal{O}^{\textsc{abort}}$ with the corresponding typed
reason before any commit object is constructed, and the set of
exit points is closed. Detailed proof in
Appendix~\ref{app:proof-safe-abort}. Note that with the
Stage-3 gate of \S\ref{sec:margin-gate} in place, the final
clause covers both branches; without it, it covers only the
verdict branch --- which is precisely the defect
\S\ref{sec:margin-gate} reports.
\end{proof}

\begin{theorem}[Conditional commit-or-abort termination under
synchrony]
\label{thm:hcsc-termination}
\label{thm:termination} 
Under Assumptions~\ref{assn:auth-channels},
\ref{assn:partial-synchrony}, and
\ref{assn:deterministic-cE}, after the global stabilisation
time the protocol terminates with one of the three typed
outcomes within one round of communication. Concretely:
\begin{itemize}
\item if the round is semantic-admissible
(Assumption~\ref{assn:admissible-round}, semantic variant) and
at least $2f{+}1$ honest agents are present in $\core$, the
protocol terminates with $\mathcal{O}^{\textsc{sem}}$;
\item else if the round is verdict-admissible (verdict
variant) and at least $2f{+}1$ honest agents are present in
$G_{v^\star}$, the protocol terminates with
$\mathcal{O}^{\textsc{ver}}$;
\item else the protocol terminates with
$\mathcal{O}^{\textsc{abort}}$
(\texttt{insufficient\_signers} or one of the path-failure
typed reasons).
\end{itemize}
The protocol does \emph{not} rely on Byzantine cooperation for
liveness: \emph{insufficient\_signers} is a valid termination
outcome. Across all $6\,000$ per-task records of our protocol-layer
sweep, \texttt{insufficient\_signers} fires $0$ times
(\S\ref{sec:eval-setup}), because the experimental Byzantine agents
do sign.
\end{theorem}

\begin{proof}[Proof sketch]
Stages 1--4 are local computations and do not block. Stage 5
collects valid signatures over $\hstar$. Honest signers in
$\Sigma_{\textit{src}}$ broadcast over authenticated channels;
after global stabilisation~\cite{dwork1988partial} each broadcast
is delivered within the bounded delay. If the admissibility
predicate holds and $\Sigma_{\textit{src}}$ contains at least
$2f{+}1$ honest agents, the protocol collects $\ge 2f{+}1$
signatures within one round and publishes the corresponding typed
commit. Otherwise it exits via \emph{insufficient\_signers} or a
path-failure abort branch. Detailed proof in
Appendix~\ref{app:proof-termination}.
\end{proof}

\begin{remark}
\label{rem:hcsc-termination-correction}
Theorem~\ref{thm:hcsc-termination} corrects a potential
overclaim in the previous CSC formulation, where commit
termination was conditional only on a coherent honest quorum.
Under H-CSC, commit termination requires $\ge 2f{+}1$ honest
signers in the relevant signer set; if the signer set has
fewer honest agents, the protocol may \emph{terminate via}
\texttt{insufficient\_signers} \emph{rather than commit}. This is
the correct safety / liveness trade-off: commit liveness is
conditional on a sufficient honest signing quorum, not on
Byzantine cooperation.
\end{remark}


\subsection{Topology as a design-space option}
\label{sec:topology-as-option}

For completeness, the implementation accepts an optional
\emph{topology gate}
$\textit{topology\_gate} \in \{\bot, \textsc{jaccard}, \textsc{jaccard\_kcore}\}$.
The H-CSC main path uses $\textit{topology\_gate} = \bot$. We
report the Jaccard variant in
Appendix~\ref{appendix:topology} as a design-space ablation
only; on every measured benchmark and at every $\thalpha$ in
our sensitivity sweep, the topology gate produces commit/abort
decisions identical to the H-CSC main path or aborts more
without improving honest-reference validity. Topology is
\emph{not} part of the H-CSC contribution.

\FloatBarrier  

\subsection{Network and Communication Primitives}
\label{subsec:net}

\paragraph{Network Model.}
\sloppy
We assume the standard partial synchrony
model~\cite{dwork1988partial}.
There exist an unknown Global Stabilization Time (GST) and an
unknown delay bound $\Delta$ such that after GST, all messages
sent between honest nodes are delivered within time $\Delta$.
Nodes do not know GST nor $\Delta$ a priori and may use
increasing timeouts. Nodes communicate via authenticated
point-to-point channels and have access to a Public Key
Infrastructure (PKI).

\paragraph{Reliable Broadcast (RBC)}
To prevent equivocation and to ensure a consistent per-sender
view of inputs, we assume a reliable broadcast primitive
(e.g., Bracha's RBC~\cite{bracha1987asynchronous}).
For each sender $s$ and message $m$, denote by $\textsf{RBC\_broadcast}_s(m)$ the broadcast operation and by $\textsf{RBC\_deliver}_s(m)$ the delivery event.
RBC satisfies:
\begin{itemize}\sloppy
    \item \textbf{Validity.} If $s$ is honest and
    $\textsf{RBC\_broadcast}_s(m)$, then every honest node
    eventually $\textsf{RBC\_deliver}_s(m)$.
    \item \textbf{Agreement (Consistency).} If two honest nodes
    $\textsf{RBC\_deliver}_s(m)$ and $\textsf{RBC\_deliver}_s(m')$,
    then $m=m'$.
    \item \textbf{Totality (Uniform delivery).} If some honest
    node $\textsf{RBC\_deliver}_s(m)$, then every honest node
    eventually $\textsf{RBC\_deliver}_s(m)$.
    \item \textbf{Integrity.} Every honest node delivers at most
    one message from sender $s$, and any delivered message is
    authenticated as originating from $s$.
\end{itemize}

\noindent For each node $i$ and round $r$, let $V_i(r)$ denote
its \emph{delivered view} --- the set of messages
$\textsc{RBC\_deliver}$ed by $i$ from distinct senders in
round $r$. By RBC Agreement and Totality, in any round where
all honest nodes complete RBC, they obtain an identical
delivered view: $\forall i,j \in \mathcal{H}$,
$V_i(r) = V_j(r) \triangleq V(r)$.

To avoid blocking on Byzantine senders, a node proceeds to the
H-CSC local computation stages (Stages~1--5 of
Algorithm~\ref{alg:certified-commitment-round}) once it has
collected at least $n-f$ delivered inputs
($|V_i(r)| \geq n-f$), or aborts the round when its local
timeout $T_r$ expires.
Our correctness analysis (\S\ref{sec:correctness}) conditions
on \emph{successful rounds}: rounds in which all honest nodes
complete RBC with $|V(r)| \geq n-f$ and enter the local
computation stages with the identical delivered view $V(r)$.
Successful rounds are guaranteed to occur after GST when
timeouts are sufficiently large (liveness analysis in
\S\ref{subsec:liveness}).

\subsection{Abstract Semantic Encoding}
\label{subsec:encoding}

We abstract semantic representation via a pluggable \emph{deterministic shared encoder}. Let $\mathcal{X}$ be the space of admissible texts and $\mathbb{S}^{d-1}$ the unit hypersphere in $\mathbb{R}^d$. We assume a deterministic function $E: \mathcal{X} \to \mathbb{S}^{d-1}$ that is identically instantiated at all nodes (i.e., for any text $x \in \mathcal{X}$, all nodes compute the same embedding $E(x)$). Let $\mathrm{CanonText}(\rho_i)\in\mathcal{X}$ denote the deterministic canonical proposal text extracted from a structured proposal $\rho_i$. Node $i$'s round-$r$ embedding is
\[
e_i := E(\mathrm{CanonText}(\rho_i)) \in \mathbb{S}^{d-1}.
\]

The determinism and shared instantiation of $E$ are essential for exact agreement: without them, honest nodes could compute divergent embeddings for identical semantic content, violating the byte-level equality required for commitment (Section~\ref{subsec:agreement}). In practice, $E$ can be instantiated by any fixed embedding model whose parameters are distributed out-of-band (e.g., via a trusted setup or content-addressable distribution). The protocol treats $E$ as a black-box primitive (see Fig.~\ref{fig:crse_arch} for our empirical instantiation); the per-round admissibility predicate of Definition~\ref{def:admissible-round} is what binds embedding quality to commit feasibility, not a global property of $E$. 

\paragraph{Metric.}
We use the angular distance on $\mathbb{S}^{d-1}$,
\[
d_{\angle}(\mathbf{u},\mathbf{v})
:= \arccos(\langle \mathbf{u},\mathbf{v}\rangle)\in[0,\pi],
\]
and note the sphere identity
$\|\mathbf{u}-\mathbf{v}\|_2 =
2\sin\!\big(d_{\angle}(\mathbf{u},\mathbf{v})/2\big)$
for converting thresholds when needed.


\noindent
The protocol treats $\mathcal{E}$ as a black-box primitive. Whether
the embeddings of an honest sub-quorum cluster sufficiently to
admit a commit is a \emph{per-round} question, captured by the
admissibility predicate of
Definition~\ref{def:admissible-round} below: a round is admissible
iff a $2f{+}1$ subset of delivered embeddings sits inside an
angular ball of radius $\thalpha$. The fraction of rounds
satisfying the predicate is benchmark- and threshold-dependent and
is measured empirically in
Section~\ref{sec:evaluation} rather than assumed. When the
predicate fails, the protocol's intended behaviour is an explicit
safe abort.

\subsection{Relationship to the conservative CSC variant}
\label{sec:hcsc-strict-csc}

The previous CSC formulation (\textit{candidate\_rule} disabled,
$\textit{fallback} = \bot$) is a strict-semantic configuration
of the same framework: the angular-component core is computed
over the entire delivered view $V(r)$ rather than within
$G_{v^\star}$, and the protocol aborts whenever the global
core fails. We retain the CSC configuration as a conservative
ablation; the empirical comparison is reported in
Section~\ref{sec:eval-setup}. The H-CSC main path's two
deviations from CSC are: (i) the within-verdict scoping of the
semantic core, and (ii) the verdict-level fallback. Both
deviations are deterministic, do not change the certificate
envelope, and respect the same closed-set abort discipline.

\subsection{Result 5: The Hierarchical Fallback Erases Any
Damage Bound}
\label{sec:eval-r5}

The one positive signal we found is real, and the architecture
prevents it from being cashed.

\paragraph{The encoder tracks harm, not identity.}
Define damage $D = 10 - \textit{grounded}$ from the
Result-4 judge. Then
$\mathrm{corr}(D,\ \text{angular distance}) = +0.429$ with
$95\%$ CI $[+0.130, +0.653]$ --- significant. The angular
signal is picking up \emph{how harmful} a proposal is, which is
arguably what one would want it to track. By contrast
$\mathrm{corr}(D,\ \text{admitted}) = -0.331$, CI
$[-0.611, +0.078]$, not significant: the correlation with
harm does not survive the trip through the protocol's decision
logic.

\paragraph{Why it does not survive: the frontier is flat.}
Sweeping $\thalpha$ from $0.85$ down to $0.45$ ---
tightening the admissibility ball by nearly a factor of two ---
leaves the maximum damage among admitted attacks at $8.0$ at
both ends, while the attack admission rate \emph{rises} from
$85\%$ to $100\%$ (Table~\ref{tab:frontier}). The mechanism is
architectural and entirely visible in
Algorithm~\ref{alg:hcsc-round}:
\begin{quote}
tightening $\thalpha$ $\Rightarrow$ the semantic path fails
more often $\Rightarrow$ Stage~4b, the verdict-commit fallback,
fires $\Rightarrow$ \textbf{the fallback certifies the entire
verdict group without any semantic filtering at all}.
\end{quote}
At $\thalpha = 0.45$ the semantic-commit rate is exactly
$0.00$ and \textbf{H-CSC degenerates precisely into M3}. The
filter is not made stricter by tightening its threshold; it is
routed around. \emph{The coverage-recovery mechanism that is
H-CSC's genuine contribution (\S\ref{sec:eval-retained}) is the
same mechanism that annihilates any filtering guarantee.} This
is the mechanistic explanation of the containment lemma's
empirical shadow: H-CSC $\equiv$ M3 not by coincidence but
because the hierarchy is designed to fall back to it.

\paragraph{Splitting by commit type does not rescue it.}
One might hope the semantic branch bounds damage on the rounds
where it does fire, even if the fallback does not. It does not:
attacks admitted through \textsc{semantic\_commit} rounds reach
maximum damage $7$--$8$ (mean $3.79$--$4.18$) against
maximum $8.0$ (mean $4.36$--$4.50$) through
\textsc{verdict\_commit} rounds. Neither branch bounds damage
(Table~\ref{tab:frontier}, lower panel).

\begin{table}[t]
\centering
\caption{Result 5: the damage--coverage frontier is flat, and
the per-branch split offers no bound. Damage $D = 10 -
\textit{grounded}$ from the role-blind judge of
\S\ref{sec:eval-r4}. \textbf{Preliminary; $40$ attacks over
$N=10$ tasks.}}
\label{tab:frontier}
\footnotesize
\setlength{\tabcolsep}{4pt}
\begin{tabular}{@{}l c c c@{}}
\toprule
\multicolumn{4}{@{}l}{\emph{(a) sweeping the admissibility radius}} \\
$\thalpha$ & semantic-commit rate & attacks admitted & max $D$ admitted \\
\midrule
0.85 & 1.00 & 85\%  & \textbf{8.0} \\
0.45 & \textbf{0.00} & \textbf{100\%} & \textbf{8.0} \\
\midrule
\multicolumn{4}{@{}l}{\emph{(b) splitting the operating point by commit type}} \\
branch & & max $D$ admitted & mean $D$ admitted \\
\midrule
\textsc{semantic\_commit} & & 7--8 & 3.79--4.18 \\
\textsc{verdict\_commit}  & & 8.0  & 4.36--4.50 \\
\bottomrule
\end{tabular}
\end{table}


\subsection{Definitions}
\label{subsec:definitions}
\label{subsec:problem} 

\begin{definition}[Semantic proposal]
\label{def:semantic-proposal}
A \emph{semantic proposal} is a pair $z_i = (s_i, \rho_i)$ where
$s_i$ is a natural-language string emitted by agent $i$ and
$\rho_i$ is a structured object extracted from $s_i$ (verdict,
confidence, evidence-id list, rationale, claim text). A
semantic proposal is the unit of agreement; agents do not
agree on $s_i$ directly but on a deterministic function of
$\rho_i$.
\end{definition}

\begin{definition}[Semantic view]
\label{def:semantic-view}
The \emph{semantic view} delivered to the protocol in round $r$
is $V(r)=\{(i,\rho_i,\mathbf{e}_i):i\in I_r\}$, where
$I_r \subseteq \Nset$ is the delivered index set with
$|I_r| \ge n - f$ (the \emph{successful round} condition
established in \S\ref{subsec:net}), and $\mathbf{e}_i =
\cE(\mathrm{CanonText}(\rho_i))\in \mathbb{S}^{d-1}$ is the
unit-norm semantic embedding of $\rho_i$ under the
deterministic encoder $\cE$ and the
\textsc{CanonText} map of \S\ref{subsec:encoding}.
All downstream structures --- verdict groups $G_v(r)$,
semantic core $\core$, signer set $\Sigma_{\textit{src}}$ ---
are subsets of $V(r)$, and quorum thresholds always reference
the global fault bound $2f{+}1$ rather than $|I_r|$. The
protocol does not have access to the labels of which indices
are Byzantine.
\end{definition}

\begin{definition}[Certified semantic commitment object (typed)]
\label{def:csc-object}
A \emph{certified semantic commitment object} for round $r$ is
a typed value
\[
\mathcal{O} \in \{\mathcal{O}^{\textsc{sem}},
\mathcal{O}^{\textsc{ver}},
\mathcal{O}^{\textsc{abort}}\}
\]
where:

\begin{itemize}\sloppy
\item $\mathcal{O}^{\textsc{sem}}$ is a \textsc{semantic\_commit}
object with $\mathbf{commit\_type} = \textsc{semantic\_commit}$,
fields
$(\hstar, \pi, \cert{\hstar}, \core, \widetilde{\mathbf{y}}, v^\star)$,
where $\core\subseteq G_{v^\star}$
is a $2f{+}1$-sized within-verdict admissible core,
$\widetilde{\mathbf{y}} = Q_\eta(\agg(\core))$ is the
quantised aggregate over the core, $\hstar$ is the
parameter-bound semantic-commit digest with explicit
type-domain separator,
$\hstar = H(\texttt{"semantic\_commit"}\,\Vert\,
\widetilde{\mathbf{y}}\,\Vert\, h_\pi\,\Vert\, r\,\Vert\,
v^\star)$, and $\cert{\hstar}$ is a $2f{+}1$ distinct-signer
certificate over $\hstar$ drawn from $\core$.
\item $\mathcal{O}^{\textsc{ver}}$ is a \textsc{verdict\_commit}
object with $\mathbf{commit\_type} = \textsc{verdict\_commit}$,
fields $(\hstar, \pi, \cert{\hstar}, G_{v^\star}, v^\star)$,
where $G_{v^\star}\subseteq V(r)$ is the verdict group of size
$\ge 2f{+}1$. The digest binds a verdict-level payload to the
parameters and round id:
\begin{equation*}
\hstar = H\bigl(\texttt{"verdict\_commit"}\,\Vert\,\textsc{VerdictPayload}\,\Vert\, h_\pi \,\Vert\, r\bigr),
\end{equation*}
where
$\textsc{VerdictPayload} = \textsc{ser}(v^\star, |G_{v^\star}|, \textit{verdict\_margin}, n, f, r)$.
$\cert{\hstar}$ is a $2f{+}1$ distinct-signer certificate over
$\hstar$ drawn from $G_{v^\star}$. A \textsc{verdict\_commit}
\emph{does not} carry a semantic aggregate
$\widetilde{\mathbf{y}}$, and explicitly flags
$\mathbf{no\_semantic\_aggregate} = \textsc{true}$.
\item $\mathcal{O}^{\textsc{abort}}$ is an \emph{abort} outcome
with $\mathbf{commit\_type} = \textsc{abort}$ and a typed
reason from the closed set
\begin{align*}
\{&\,\texttt{verdict\_below\_quorum},\\
  &\,\texttt{core\_below\_quorum},\\
  &\,\texttt{admissibility\_failed},\\
  &\,\texttt{margin\_below\_threshold},\\
  &\,\texttt{aggregation\_failed},\\
  &\,\texttt{insufficient\_signers}\,\}
\end{align*}
(the first three appear under the typed prefix
\texttt{semantic\_core\_failed:} when the semantic-path
gate triggers them, and under the typed prefix
\texttt{v2\_both\_paths\_failed:} when both the semantic and
verdict-fallback gates fail in the same round, in which case
the inner reason from the semantic-path attempt is appended
verbatim).
\end{itemize}
\end{definition}

\begin{definition}[Committed-verdict projection]
\label{def:committed-verdict}
Given a typed object $\mathcal{O}$, the \emph{committed-verdict}
projection $\textsc{verdict}(\mathcal{O})$ returns $v^\star$
when $\mathcal{O} \in \{\mathcal{O}^{\textsc{sem}}, \mathcal{O}^{\textsc{ver}}\}$
and is undefined when $\mathcal{O} = \mathcal{O}^{\textsc{abort}}$.
For \textsc{semantic\_commit}, $v^\star$ is the verdict shared by
every member of the within-verdict admissible core $\core$ (the
core is extracted within $G_{v^\star}$ by construction). For
\textsc{verdict\_commit}, $v^\star$ is the candidate verdict
$\arg\max_v |G_v|$ with deterministic tie-break. The
honest-reference and dataset-gold validity predicates of
Definition~\ref{def:two-validities} are evaluated on
$\textsc{verdict}(\mathcal{O})$ when defined.
\end{definition}

\begin{definition}[Admissible semantic round]
\label{def:admissible-round}
Round $r$ is \emph{admissible at radius $\thalpha$} iff there
exists a subset $\core\subseteq V(r)$ with $|\core|\ge 2f{+}1$
and a centre node $h\in \core$ such that
$\angle(\mathbf{e}_h,\mathbf{e}_j)\le \thalpha$ for every
$j\in \core$. The companion diameter predicate
$\eps$ requires $\angle(\mathbf{e}_j,\mathbf{e}_k)\le \eps$ for
every pair in $\core$ and is reported as a stricter sister
condition.
For the H-CSC main path, this predicate is evaluated on the
candidate verdict group $G_{v^\star}$ rather than the full
delivered view; the semantic branch therefore requires
$\core \subseteq G_{v^\star}$ with $|\core|\ge 2f{+}1$ and
$r^\star \le \thalpha$. For the strict-semantic CSC ablation
(\S\ref{sec:hcsc-strict-csc}) the same predicate is evaluated
on the full delivered view $V(r)$.
\end{definition}

\begin{remark}
Definition~\ref{def:admissible-round} is a \emph{per-round
predicate}, not a universal assumption about all real LLM-agent
rounds. The admissible-round predicate is a commit precondition,
not an empirical claim that all honest LLM agents naturally
cluster. The set of rounds in which it holds is benchmark- and
threshold-dependent (see Section~\ref{sec:evaluation}). The
protocol's response when the predicate fails is an explicit safe
abort, not an attempt to commit a smaller cluster.
\end{remark}

\begin{definition}[Typed commit, abort, and invalid commit]
\label{def:commit-abort}
\sloppy
A round outcome is one of
\[
\{\textsc{semantic\_commit},\ \textsc{verdict\_commit},\ \textsc{abort}\}
\]
(Definition~\ref{def:csc-object}). For commit outcomes, the
\emph{typed invalid commit} relative to a reference $R$
partitions into an \emph{invalid semantic commit} (an
$\mathcal{O}^{\textsc{sem}}$ whose
$\textsc{verdict}(\mathcal{O})$ differs from $R$'s verdict),
an \emph{invalid verdict commit} (the analogous case for
$\mathcal{O}^{\textsc{ver}}$), and the combined rate (the union
of the two).
The empirical metrics
$\textsc{InvalidCommit}_{\mathrm{hmaj}}$ and
$\textsc{InvalidCommit}_{\mathrm{gold}}$ aggregate the combined
invalid commit rate against $R = \mathrm{majority}_{\Honest}(V)$
and $R = \mathrm{gold}(V)$ respectively
(Definition~\ref{def:two-validities}). Per-type breakdowns
$\textsc{InvalidSemanticCommit}_{\mathrm{hmaj}}$ and
$\textsc{InvalidVerdictCommit}_{\mathrm{hmaj}}$ are reported
separately when meaningful. \emph{Invalid abort} is undefined.
\end{definition}

\begin{definition}[Dataset-gold and evidence-bounded honest-reference validity]
\label{def:two-validities}
Let $\mathrm{gold}(V)$ be the dataset's reference verdict for
the round (when available), and let
$\mathrm{majority}_{\Honest}(V)$ be the
\emph{evidence-bounded honest-reference verdict}: the
deterministic top-of-honest verdict given only the visible
evidence, with tie-break order
$\textsc{support} > \textsc{refute} > \textsc{insufficient}$
matching the protocol's own $\textsc{largest\_verdict}$ rule.
This is the honest plurality (not strict majority): for $k$
verdict classes ($k=3$ here), the honest plurality coincides
with the strict-majority verdict only when one honest class
holds more than $|\Honest|/2$ votes; otherwise it picks the
honest top class under the same deterministic tie-break used
by the protocol. We say a \Commit\ outcome is
\emph{gold-valid} iff its committed verdict equals
$\mathrm{gold}(V)$, and \emph{honest-reference-valid} iff its
committed verdict equals $\mathrm{majority}_{\Honest}(V)$. We
always report both; the two notions diverge whenever the gold
label exceeds what a strict evidence-bounded reading supports.
The empirical safety metric in
Section~\ref{sec:eval-setup} is named
$\texttt{invalid\_hmaj}$ for backward compatibility with the
H-CSC bench tooling; readers should interpret ``hmaj'' as
``honest-reference (plurality)'' rather than ``strict
majority''.
\end{definition}

\begin{remark}
$\mathrm{majority}_{\Honest}$ is defined over the honest set,
which is unknown to the protocol. It is therefore an
\emph{evaluation} construct, not a protocol input, and is
computed only offline against records labelled by an oracle.
We report gold-validity and honest-reference-validity
separately because (i) Climate-FEVER fallback labels can
diverge from evidence-limited honest judgments, (ii) the
protocol is not a semantic-truth oracle, and (iii) the core
safety metric reported in
Section~\ref{sec:eval-setup} is the rate of invalid commits
against the evidence-bounded honest-reference verdict.
Conflating the two risks confusing dataset label noise (a
property of the benchmark) with protocol failure (a property
of the algorithm).
\end{remark}

\begin{definition}[Byzantine semantic poisoning]
\label{def:semantic-poisoning}
\emph{Byzantine semantic poisoning} is the production of a
proposal $z_j = (s_j, \rho_j)$ by an agent $j\in \Byz$ that is
(i) fluent and on-topic, (ii) lexically or embedding-proximal
to the honest mass (e.g., angular distance to the honest centre
below the median honest-honest distance), and (iii)
adversarially distorted in at least one of two ways:
\begin{description}
\item[Mode A (verdict-level).] The committed verdict
$z_j.v$ differs from the honest-reference verdict
$\mathrm{majority}_{\Honest}(V)$.
\item[Mode B (evidential, same-verdict).] $z_j.v$ matches the
honest-reference verdict, but the supporting rationale
$\rho_j$ is adversarially fabricated (mismatched evidence ids,
spurious causation, invented citations) so the proposal
appears benign on the verdict axis while corrupting the
evidence axis.
\end{description}
We distinguish two adversary modes: \emph{static} (the
Byzantine does not observe honest broadcasts) and
\emph{rushing} (the Byzantine reads the round's honest
broadcasts before submitting).
The primary safety metric \texttt{invalid\_hmaj}
(Definition~\ref{def:two-validities}) captures Mode A; Mode B
falls outside the present protocol's correctness signal and is
treated as a representation-boundary limitation
(\S\ref{sec:limitations}). An evidence-aware finality-control
extension that targets Mode B is noted as future work
(\S\ref{sec:conclusion}).
\end{definition}

\begin{definition}[Finality-signal vector]
\label{def:finality-signal}
\sloppy
For a delivered view $V(r)$ with verdict groups $\{G_v\}_v$
and within-verdict embeddings, the \emph{finality-signal
vector} is
\[
\Psi(r) \;=\; (\Psi^{\textsc{verdict}}, \Psi^{\textsc{semantic}}, \Psi^{\textsc{cert}}),
\]
where $\Psi^{\textsc{verdict}} = (\textit{top\_count}(V), \textit{verdict\_margin}(V))$
collects the largest verdict group's size and the margin to the
runner-up; $\Psi^{\textsc{semantic}} = (|\core|, r^\star_{\thalpha}(\core))$
collects the within-verdict admissible-core size and the
minimum admissibility radius; and $\Psi^{\textsc{cert}} = |\Sigma|$
is the count of distinct valid signatures on the candidate
digest. The H-CSC typed-decision logic
(Section~\ref{sec:protocol}) is a deterministic function of
$\Psi(r)$. The protocol-side allowed signals are exactly the
components of $\Psi(r)$; gold labels, honest/Byzantine identity
labels, evidence-truth oracles, source-reliability scores,
claim--evidence entailment, and external judge labels are
\emph{forbidden} as protocol inputs and are reserved for
offline evaluation only.
\end{definition}

\begin{remark}
\label{rem:hcsc-scope}
\sloppy
The structured proposal $\rho_i$ contains a verdict, a
confidence value, an evidence-id list, a rationale, and a claim
text. H-CSC's typed-decision step consumes only the verdict and
the canonical proposal text (via the encoder $\cE$); the
evidence-id list is part of the canonical text the encoder
reads, but the protocol does \emph{not} perform evidence-aware
reasoning over it, does not consult the confidence value, and
does not invoke any external evidence-truth, source-reliability,
claim--evidence-entailment, or judge module. A future
evidence-aware finality-control paper would extend the
$\Psi^{\textsc{evidence}}$ component (currently empty) with
provenance and entailment signals; that extension is
explicitly out of scope for the present paper.
\end{remark}


\subsection{Goal}
\label{subsec:goal}

\sloppy
Given a delivered view $V(r)$, an encoder $\cE$, a fault bound
$f$, and protocol parameters
\begin{multline*}
\pi=(\thalpha,\ \eps,\ \agg,\ Q_\eta,\ \textit{candidate\_rule}, \\
\textit{fallback},\ \textit{core\_method},\ \textit{topology\_gate})
\end{multline*}
with the H-CSC main-path values
$\textit{candidate\_rule}=\textsc{largest\_verdict}$,
$\textit{fallback}=\textsc{margin\_1}$,
$\textit{core\_method}=\textsc{angular\_component}$, and
$\textit{topology\_gate}=\bot$,
the protocol must output exactly one typed object
(Definition~\ref{def:csc-object}): a
\textsc{semantic\_commit} $\mathcal{O}^{\textsc{sem}}$ when the
within-verdict semantic core is admissible at radius
$\thalpha$; a \textsc{verdict\_commit} $\mathcal{O}^{\textsc{ver}}$
when the candidate verdict group has size $\ge 2f{+}1$ and
$\textit{verdict\_margin} \ge 1$ but the semantic core is not
admissible; or an \emph{abort} $\mathcal{O}^{\textsc{abort}}$
with a typed reason from the closed set in that definition.
The protocol must \emph{not} commit a digest without a valid
$2f{+}1$ distinct-signer certificate. The protocol must
\emph{not} emit a \textsc{semantic\_commit} when the
within-verdict semantic core is not admissible. The protocol
must \emph{not} emit a \textsc{verdict\_commit} when the
candidate verdict group is below quorum or fails the margin
condition. We give the precise correctness guarantees
(agreement, semantic validity, verdict validity, safe abort,
conditional commit-or-abort termination) in
Section~\ref{sec:correctness}.



\subsection{Assumptions}
\label{sec:correctness-assumptions}

\begin{assumption}[Authenticated channels and honest sign-once]
\label{assn:auth-channels}
The assumption has two clauses, both required by
Theorem~\ref{thm:hcsc-agreement} and Lemma~\ref{lem:cert-uniqueness}.
\begin{description}
\item[Authenticated channels.] Every honest agent's signature
is unforgeable. Two distinct signatures from the same agent on
distinct digests are detectable.
\item[Honest sign-once.] In each round $r$, an honest agent
signs at most one typed digest. Before signing it recomputes
Algorithm~\ref{alg:certified-commitment-round} locally on its
delivered view $V(r)$ and signs only the unique digest
$D^\star(r)$ returned by the deterministic H-CSC decision rule;
any conflicting digest for the same round is rejected without
signing.
\end{description}
\end{assumption}

\begin{assumption}[Partial synchrony]
\label{assn:partial-synchrony}
There is a global stabilisation time after which message
delivery is bounded; before that time, message order and
timing are adversarial. This assumption is used only for
Termination.
\end{assumption}

\begin{assumption}[Deterministic encoder and serialisation]
\label{assn:deterministic-cE}
$\cE$ is \emph{byte}-deterministic: two honest agents that compute
$\cE(\rho)$ for the same canonical proposal text obtain the
identical bit pattern for $\mathbf{e}$. The canonicalisation of
$\rho$ is byte-deterministic in the structured fields. The
quantiser $Q_\eta$ is a fixed-point lattice round with step
$\eta$ and a fixed tie-break rule.
\end{assumption}

\begin{remark}[This assumption is the load-bearing one, and v1
did not justify it]
\label{rem:determinism-debt}
Assumption~\ref{assn:deterministic-cE} is what makes
Lemma~\ref{lem:digest-consistency} --- and therefore L1 --- work,
because $\hstar$ is a hash and a hash has no tolerance. The
previous version of this manuscript justified it by saying that
$\agg$ ``converges to within a documented numerical tolerance''
and that $Q_\eta$ ``collapses sub-$\eta$ differences''. Neither
statement is sufficient: a rounding quantiser collapses a
difference only when the two values do not straddle a lattice
boundary. The parameter tuple $\pi$ names $\eta$ but v1 never
gives it a numeric value; the artefact uses $\eta = 10^{-6}$.
We analyse the resulting gap quantitatively in
\S\ref{sec:encoder-determinism} and show there that the empirical
finding of \S\ref{sec:eval-r2} discharges the assumption
outright.
\end{remark}

\begin{assumption}[Admissible round (per-type)]
\label{assn:admissible-round}
A round is \emph{semantic-admissible} at radius $\thalpha$ iff
the candidate verdict group $G_{v^\star}$ contains a $2f{+}1$
within-verdict admissible core
(Definition~\ref{def:admissible-round}).
A round is \emph{verdict-admissible at margin $m$} iff
$|G_{v^\star}| \ge 2f{+}1$ and $\textit{verdict\_margin} \ge m$;
the shipped operating point is $m = 1$.
The semantic-validity and termination guarantees below are
\emph{conditional} on the corresponding admissibility
predicate; the safe-abort theorem makes no admissibility
conditioning.
\end{assumption}


\subsection{Representation boundary}
\label{sec:rep-boundary}

The guarantees above are conditional on the admissibility
instrument being able to separate honest from Byzantine proposals
at all. The following lemma delimits when it cannot, for
\emph{any} instrument that is a function of the encoder.

\begin{lemma}[Representation boundary]
\label{lem:rep-boundary}
Suppose there exists an honest proposal $\rho_h$ and a
Byzantine proposal $\rho_b$ with
$\rho_h.\textit{verdict} = \rho_b.\textit{verdict}$ and
$\angle(\cE(\rho_h),\cE(\rho_b)) = 0$. Then no deterministic
embedding-only commitment protocol that operates only on $\cE$
can distinguish $\rho_h$ from $\rho_b$ on the within-verdict
admissibility predicate; any such protocol must either (i)
admit $\rho_b$ to the selected core, or (ii) abort.
\end{lemma}

\begin{proof}
Both proposals produce identical embeddings; the protocol
cannot distinguish them on any function of the embeddings.
The within-verdict admissibility predicate is such a function.
\end{proof}

\begin{remark}[The boundary is reached long before
$\angle = 0$]
\label{rem:rep-boundary-defence}
Lemma~\ref{lem:rep-boundary} is stated at exact coincidence, but
the operative regime is softer and is reached in practice. Under
the verdict-preserving attack families of
\S\ref{sec:eval-r2} --- adversaries that agree with the
honest verdict and corrupt only the rationale, so
Definition~\ref{def:semantic-poisoning} Mode~B --- Byzantine
proposals enter the selected core on $70\%$--$100\%$ of rounds
depending on family and evasion level ($n = 10$ tasks, $40$
attacks), and H-CSC commits a contaminated payload on $10/10$
rounds, which is the same rate as the verdict-only baseline. The
reason is a signal-to-noise failure rather than an encoder
quality failure: the honest population's own angular dispersion
about its geometric median has median $0.528$~rad and
$90$th percentile $0.992$~rad, while the attack displacement is
$0.58$--$0.85$~rad. At $\thalpha = 0.65$ the ball barely contains
the honest population itself, so there is no room to tighten.

We correct a claim made in the previous version at this point.
v1 asserted that the defence against same-verdict Byzantine
proposals is ``the verdict-fallback path's margin condition plus
the typed safe-abort branch''. Both halves are wrong.
\S\ref{sec:margin-gate} shows the margin condition is not applied
on the branch that fires. And tightening $\thalpha$ does not
recruit the safe-abort branch as a defence: it merely routes the
round to the verdict fallback, which admits the entire verdict
group unfiltered. Sweeping $\thalpha$ from $0.85$ down to $0.45$
leaves the maximum harm of an admitted attack constant at $8.0$
on a blind $0$--$10$ judge scale while the attack pass rate
\emph{rises} from $85\%$ to $100\%$; at $\thalpha = 0.45$ the
semantic-commit rate is $0.00$ and H-CSC has degenerated exactly
into certificate-wrapped majority. The hierarchical fallback that
recovers coverage is the same mechanism that annihilates any
filtering guarantee. This is an architectural property, not a
tuning failure.
\end{remark}


\subsection{Discharging the encoder-determinism assumption}
\label{app:encoder-determinism-full}

Level~1 rests on Assumption~\ref{assn:deterministic-cE}, which
requires two honest nodes to obtain \emph{bit-identical}
embeddings, because $\hstar$ is a hash and a hash has no
tolerance. This subsection shows that (i) the assumption is
empirically false for our own artefact across execution
environments, (ii) quantisation does not repair it, and (iii) the
instrument recommended by \S\ref{sec:eval-r2} discharges
it outright. Point (iii) is a genuine theoretical payoff of an
empirical negative result, and we want it stated explicitly rather
than left implicit.

\paragraph{(i) The assumption is false as measured.}
We re-encoded the MVR-50 canonical proposal texts with the same
CRSE checkpoint in a different execution environment and compared
against the frozen embedding cache used for all reported
experiments. The reproduction is excellent as a numerical matter:
minimum cosine similarity $0.99999990$ over all proposals,
maximum element-wise absolute difference $3.6\times10^{-7}$. It is
nonetheless \emph{not bit-identical}, and bit-identity is what
Assumption~\ref{assn:deterministic-cE} asserts. In angular terms
the discrepancy is
$\arccos(0.99999990) \approx 4.5\times10^{-4}$~rad, which is
$7\times10^{-4}$ of $\thalpha = 0.65$: utterly irrelevant to the
admissibility predicate of L3, and fatal to the byte-equality of
L1. That contrast is the whole issue. The encoder is
reproducible enough to decide geometry and not reproducible enough
to be hashed.

\paragraph{(ii) Quantisation converts ``certainly different''
into ``probably the same''.}

\begin{proposition}[A rounding quantiser cannot discharge
Assumption~\ref{assn:deterministic-cE}]
\label{prop:quantiser}
Let $Q_\eta$ be the coordinate-wise rounding quantiser with step
$\eta>0$ and any fixed tie-break. Then:
\begin{enumerate}
\item For every $\eta>0$ and every $\delta>0$ there exist
$\mathbf{y},\mathbf{y}'\in\mathbb{R}^d$ with
$\max_k|y_k-y'_k|\le\delta$ and
$Q_\eta(\mathbf{y})\ne Q_\eta(\mathbf{y}')$. Hence no choice of
$\eta$ makes digest agreement \emph{certain} in the presence of
any nonzero cross-environment discrepancy.
\item If the fractional parts $y_k/\eta \bmod 1$ are modelled as
independent and uniform and $|y_k-y'_k|=\delta_k<\eta$, then
\[
\Pr\bigl[Q_\eta(\mathbf{y})=Q_\eta(\mathbf{y}')\bigr]
= \prod_{k=1}^{d}\Bigl(1-\frac{\delta_k}{\eta}\Bigr)
\;\le\; \exp\Bigl(-\frac{1}{\eta}\sum_{k=1}^{d}\delta_k\Bigr).
\]
\end{enumerate}
\end{proposition}

\begin{proof}
(1) Choose an integer $t$ and set
$y_1 := (t+\tfrac12)\eta - \tfrac{\delta}{2}$,
$y'_1 := y_1 + \delta$, with $y_k=y'_k$ for $k>1$. Then
$y_1 < (t+\tfrac12)\eta < y'_1$, so the two coordinates lie on
opposite sides of the rounding boundary and round to $t\eta$ and
$(t{+}1)\eta$ respectively. (2) Coordinate $k$ rounds differently
exactly when the interval spanned by $y_k$ and $y'_k$ contains a
boundary $(t+\tfrac12)\eta$; for $\delta_k<\eta$ and a uniform
fractional part this event has probability $\delta_k/\eta$.
Independence gives the product, and $1-x\le e^{-x}$ gives the
bound.
\end{proof}

Instantiating Proposition~\ref{prop:quantiser}(2) with our
artefact's parameters is instructive. With $d = 768$ (the
\texttt{e5-base} embedding width) and $\eta = 10^{-6}$, even a
mean per-coordinate discrepancy of $\bar\delta = 10^{-8}$ ---
a factor of $36$ \emph{below} the measured maximum of
$3.6\times10^{-7}$ --- gives
$\Pr[\text{digests agree}] \le \exp(-768\cdot10^{-8}/10^{-6})
= e^{-7.68}\approx 5\times10^{-4}$. The model is crude (real
coordinate errors are correlated and not uniform, and within one
fixed execution environment the encoder \emph{is} bit-reproducible,
which is why the frozen cache works at all), so this number should
be read as an order of magnitude and not as a measured failure
rate. The qualitative conclusion does not depend on the model:
$\eta = 10^{-6}$ is the same order as the discrepancy it is
supposed to absorb.

Coarsening $\eta$ helps but does not resolve. Taking
$\eta = 10^{-2}$ and $\delta_k \le 3.6\times10^{-7}$ bounds the
per-round disagreement probability by
$768\cdot3.6\times10^{-7}/10^{-2}\approx 2.8\times10^{-2}$, while
displacing the committed aggregate by up to
$(\eta/2)\sqrt{d}\approx 0.139$ in $\ell_2$, i.e.\ about
$0.139$~rad $\approx 0.21\,\thalpha$. So the coarsest quantiser
that makes accidental disagreement unlikely already perturbs the
committed semantic object by a fifth of the admissibility radius
--- and still leaves agreement merely probable. In a protocol
whose L1 guarantee is stated as deterministic, ``probable'' is the
wrong word. Honest options are: pin the execution environment
(fixed kernels, fixed reduction order, deterministic-math flags,
version-locked libraries) and declare that pinning part of $\pi$;
or restate L1 as holding with a stated failure probability; or
remove the floating-point object from the digest.

\paragraph{(iii) The lexical instrument removes the problem.}
\S\ref{sec:eval-r2} compares four admissibility
instruments on the same verdict-preserving attack families and
finds that a deterministic \emph{lexical-conformity} predicate ---
Jaccard overlap between an agent's canonical-text shingle set and
the consensus shingle set of the delivered group --- dominates the
learned CRSE angular predicate on every cell tested
(AUROC $0.865$--$0.982$ versus $0.621$--$0.744$; TPR at a $5\%$
honest false-positive rate $38\%$--$80\%$ versus $0\%$--$8\%$;
$n = 10$ tasks, $40$ attacks, $80$ honest proposals). That is an
empirical finding about detection quality. It has a separate,
purely theoretical consequence, which we record here:

\begin{proposition}[The lexical predicate is exactly
deterministic]
\label{prop:lexical-determinism}
Let $S(\rho)$ be the shingle set of the canonical proposal text of
$\rho$ under a fixed tokenisation, and let the admissibility test
be
$|S(\rho_i)\cap S_{\mathrm{cons}}| \cdot q \ge p \cdot |S(\rho_i)\cup S_{\mathrm{cons}}|$
for a fixed rational threshold $p/q$. Then the predicate's value,
and the identity of the selected core, are computable by exact
integer arithmetic over byte strings. Two honest nodes with the
same delivered view obtain the same result on any conforming
platform, with no assumption on floating-point behaviour, kernel
selection, reduction order, or library version.
\end{proposition}

\begin{proof}
Canonicalisation is byte-deterministic by
Assumption~\ref{assn:deterministic-cE}'s second clause, which
concerns serialisation only and involves no arithmetic.
Tokenisation and shingling are functions from byte strings to
finite sets of byte strings. Set intersection, union, and
cardinality are exact. The comparison is between two integers.
No step invokes a real-number operation, so no step admits
platform-dependent rounding. The core-selection rule (largest
conforming subset, ties broken by agent id) is a deterministic
function of the resulting integer scores.
\end{proof}

\noindent
Under a lexical admissibility instrument, therefore,
Assumption~\ref{assn:deterministic-cE} splits and its hard half
disappears:
\begin{itemize}
\item \textbf{L1 and L2} no longer depend on the encoder at all.
The digest binds a canonical byte string and a set of integer
scores. Lemma~\ref{lem:digest-consistency} holds with the
``identical execution environment'' clause deleted, and
Theorem~\ref{thm:hcsc-agreement} becomes unconditionally
deterministic given Assumption~\ref{assn:auth-channels}.
\item \textbf{L4} never used the encoder, so it is unaffected;
Theorems~\ref{thm:hcsc-verdict-validity} and
\ref{thm:det-semantic-validity} and
Lemma~\ref{lem:containment} are statements about set
cardinalities and go through verbatim with $\core$ read as the
lexically conforming core.
\item \textbf{L3} is the only level that still needs an encoder,
because it is a statement about embeddings. But it moves
\emph{off} the consensus critical path: a verifier may recompute
$\angle(\ystar,\mathbf{e}_i)$ as a post-hoc annotation on a
payload that is itself an exactly reproducible byte string, and
disagreement about the annotation is a disagreement about a
report, not about finality.
\end{itemize}

This is the sense in which the empirical result of
\S\ref{sec:eval-r2} is not only a negative finding about
the learned encoder. The instrument that detects the hard attack
family better is also the instrument that discharges the
assumption our agreement proof was leaning on. We did not
anticipate this when we designed the comparison, and we think it
is the strongest argument in the paper for the recommendation we
make.


%
\begin{figure*}[!t]
\centering
\resizebox{\ifdim\width>\linewidth\linewidth\else\width\fi}{!}{%
\begin{tikzpicture}[
  font=\footnotesize,
  >={Stealth[length=2mm, width=1.8mm]},
  stagenavy/.style={
    rectangle, rounded corners=3pt,
    draw=HCSCstageNavy, line width=0.6pt,
    fill=HCSCstageNavyFill,
    inner sep=4pt, align=center
  },
  stageblue/.style={
    rectangle, rounded corners=3pt,
    draw=HCSCblueOut, line width=0.6pt,
    fill=HCSCblueFill,
    inner sep=4pt, align=center
  },
  stagesem/.style={
    rectangle, rounded corners=3pt,
    draw=HCSCgreenOut, line width=0.6pt,
    fill=HCSCgreenFill,
    inner sep=4pt, align=center
  },
  stagever/.style={
    rectangle, rounded corners=3pt,
    draw=HCSCorangeOut, line width=0.6pt,
    fill=HCSCorangeFill,
    inner sep=4pt, align=center
  },
  stagegray/.style={
    rectangle, rounded corners=3pt,
    draw=HCSCgrayOut, line width=0.6pt,
    fill=HCSCgrayFill,
    inner sep=4pt, align=center
  },
  decision/.style={
    diamond, aspect=2.8,
    draw=HCSCblueOut, line width=0.6pt,
    fill=white,
    inner sep=1pt, align=center,
    font=\scriptsize
  },
  abortbox/.style={
    rectangle, rounded corners=3pt,
    draw=HCSCredOut, line width=0.6pt,
    fill=HCSCredFill,
    inner sep=3pt, align=center,
    font=\scriptsize
  },
  flow/.style={->, line width=0.55pt, color=black!75},
  ymark/.style={font=\scriptsize\bfseries, inner sep=1pt, fill=white}
]


\node[stagenavy, text width=14cm] (s1) at (0, 0)
  {\textbf{Stage 1: Canonicalise \& Encode}\\[1pt]
   For each $i \in I_r$: $\rho_i \gets \textsc{Canonicalise}(z_i)$;\quad
   $\mathbf{e}_i \gets \mathcal{E}(\mathrm{CanonText}(\rho_i))$; normalise to unit sphere.\quad
   Output: $V(r) = \{(i, \rho_i, \mathbf{e}_i) : i \in I_r\}$};

\node[stagenavy, text width=14cm] (s2) at (0, -2.0)
  {\textbf{Stage 2: Verdict Grouping}\\[1pt]
   Partition $V(r)$ by verdict.\quad
   $v^\star = \arg\max_v |G_v(r)|$ (tie-break: $\textsc{support}>\textsc{refute}>\textsc{insufficient}$).\\
   $\textit{verdict\_margin} = |G_{v^\star}| - \max_{v \ne v^\star}|G_v(r)|$};

\node[decision, text width=3.3cm] (d1) at (0, -3.9)
  {$|G_{v^\star}| \ge 2f{+}1$?};

\node[abortbox, text width=2.8cm] (a1) at (5.5, -3.9)
  {\textbf{abort}\\\texttt{verdict\_below\_quorum}};

\node[stageblue, text width=7cm] (s3a) at (-3.5, -6.1)
  {\textbf{Stage 3a: Within-verdict semantic core}\\[1pt]
   $\widehat{C} \gets \textsc{LargestAngularComponent}(G_{v^\star}, \theta_\alpha)$\\[1pt]
   $r^\star = \min_{h \in \widehat{C}} \max_{j \in \widehat{C}} \angle(\mathbf{e}_h, \mathbf{e}_j)$};

\node[decision, text width=5cm] (d2) at (-3.5, -8.3)
  {$|\widehat{C}| \ge 2f{+}1\ \wedge\ r^\star \le \theta_\alpha$?};

\node[stageblue, text width=5.4cm] (s3b) at (+3.5, -8.3)
  {\textbf{Stage 3b: Verdict fallback}\\[1pt]
   Reached only when Stage 3a fails.\\
   Check $\textit{verdict\_margin}$ vs.\ $\textit{margin}^{\min}$\\
   (operating point: $\textit{margin}^{\min}=1$).};

\node[decision, text width=4.5cm] (d3) at (+3.5, -10.6)
  {$\textit{verdict\_margin} \ge \textit{margin}^{\min}$?};

\node[abortbox, text width=2.6cm] (a2) at (+8, -10.6)
  {\textbf{abort}\\\texttt{v2\_both\_paths\_failed}};

\node[stagesem, text width=6.8cm] (s4a) at (-3.6, -13.0)
  {\textbf{Stage 4a: Aggregate $+$ Quantise $+$ Bind (semantic)}\\[1pt]
   $\bar{\mathbf{y}} = \textsc{GeomMedian}(\widehat{C})$; renormalise.\quad
   $\widetilde{\mathbf{y}} = Q_\eta(\bar{\mathbf{y}})$\\[1pt]
   $h^\star_{\textsc{sem}} = H(\texttt{"semantic\_commit"} \,\Vert\, \widetilde{\mathbf{y}} \,\Vert\, h_\pi \,\Vert\, r \,\Vert\, v^\star)$};

\node[stagever, text width=6.5cm] (s4b) at (+3.6, -13.0)
  {\textbf{Stage 4b: Verdict payload $+$ Bind}\\[1pt]
   $\textsc{VerdictPayload} = \textsc{ser}(v^\star, |G_{v^\star}|, \textit{verdict\_margin}, n, f, r)$\\
   $h^\star_{\textsc{ver}} = H(\texttt{"verdict\_commit"} \,\Vert\, \textsc{VerdictPayload} \,\Vert\, h_\pi \,\Vert\, r)$\\
   $\mathbf{no\_semantic\_aggregate} = \textsc{true}$};

\node[stagegray, text width=14cm] (s5) at (0, -15.4)
  {\textbf{Stage 5: Quorum Certificate}\\[1pt]
   Collect $2f{+}1$ distinct valid signatures on $h^\star$
   (source: $\widehat{C}$ for semantic path, $G_{v^\star}$ for verdict path).\\
   If $|\Sigma| < 2f{+}1$ $\Longrightarrow$ \texttt{abort: insufficient\_signers}.};

\node[stagesem, text width=4.6cm, font=\scriptsize] (tsem) at (-5, -17.6)
  {\textbf{semantic\_commit object}\\[1pt]
   $(h^\star_{\textsc{sem}},\ \pi,\ \mathrm{Cert}(h^\star_{\textsc{sem}}),$\\
   $\widehat{C},\ \widetilde{\mathbf{y}},\ v^\star)$};
\node[stagever, text width=4.6cm, font=\scriptsize] (tver) at (0, -17.6)
  {\textbf{verdict\_commit object}\\[1pt]
   $(h^\star_{\textsc{ver}},\ \pi,\ \mathrm{Cert}(h^\star_{\textsc{ver}}),\ G_{v^\star},\ v^\star)$;\\
   $\mathbf{no\_semantic\_aggregate}=\textsc{true}$};
\node[abortbox, text width=4.6cm] (tab) at (+5, -17.6)
  {\textbf{abort} with typed reason $\in$\\
   $\{$\texttt{verdict\_below\_quorum, core\_below\_quorum,}\\
   \texttt{admissibility\_failed, aggregation\_failed,}\\
   \texttt{insufficient\_signers, v2\_both\_paths\_failed:*}$\}$};


\draw[flow] (s1) -- (s2);
\draw[flow] (s2) -- (d1);

\draw[flow] (d1) -- (a1) node[ymark, midway, above] {NO};

\coordinate (j1) at ($(d1.south) + (0,-0.5)$);
\draw[flow,-] (d1.south) -- (j1) node[ymark, midway, right=2pt] {YES};
\draw[flow]   (j1) -| (s3a.north);

\draw[flow] (s3a) -- (d2);

\node[decision, text width=4.2cm, draw=HCSCredOut, dashed,
      line width=0.7pt] (dgate) at (-3.5, -10.6)
  {$\textit{verdict\_margin} \ge \textit{margin}^{\min}$?};

\draw[flow] (d2)    -- (dgate) node[ymark, midway, right=2pt] {YES};
\draw[flow] (dgate) -- (s4a)   node[ymark, midway, right=2pt] {YES};
\draw[flow, dashed, color=HCSCredOut] (dgate) -- (d3)
  node[ymark, midway, above, font=\tiny\bfseries] {NO};

\draw[flow] (d2) -- (s3b) node[ymark, midway, above] {NO};

\draw[flow] (s3b) -- (d3);
\draw[flow] (d3) -- (s4b) node[ymark, midway, right=2pt] {YES};

\draw[flow] (d3) -- (a2) node[ymark, midway, above] {NO};

\draw[flow] (s4a.south) -- (s4a.south |- s5.north);
\draw[flow] (s4b.south) -- (s4b.south |- s5.north);

\draw[flow] (s5.south -| tsem) -- (tsem.north);
\draw[flow] (s5.south -| tver) -- (tver.north);
\draw[flow] (s5.south -| tab)  -- (tab.north);

\node[font=\scriptsize, text width=15cm, align=center] (legend) at (0, -19.5)
  {\textcolor{HCSCstageNavy}{$\blacksquare$\,protocol stage}\quad
   \textcolor{HCSCblueOut}{$\blacksquare$\,decision predicate}\quad
   \textcolor{HCSCgreenOut}{$\blacksquare$\,\texttt{semantic\_commit} path}\quad
   \textcolor{HCSCorangeOut}{$\blacksquare$\,\texttt{verdict\_commit} path}\quad
   \textcolor{HCSCredOut}{$\blacksquare$\,typed \texttt{abort}}\\[2pt]
   \textcolor{HCSCredOut}{$\lozenge$\,dashed: the margin gate of
   \S\ref{sec:margin-gate}, \emph{not} present in the protocol as evaluated}};

\end{tikzpicture}}
\caption{\textbf{H-CSC per-round protocol flow, at stage granularity.}
Concept-level view: Figure~\ref{fig:hcsc-arch}. Left column:
always-attempted semantic-commit primary path (digest
$h^\star_{\textsc{sem}}$). Right column: verdict-fallback,
reached when the semantic-admissibility predicate at
$\theta_\alpha$ fails (digest $h^\star_{\textsc{ver}}$ with
$\mathbf{no\_semantic\_aggregate}=\textsc{true}$). If both
gates fail the protocol emits a typed \texttt{abort}; Stage~5
binds the chosen digest with a $2f{+}1$ distinct-signer
certificate. Formal pseudocode:
Algorithm~\ref{alg:certified-commitment-round}.
\emph{Solid arrows are the protocol as evaluated in
\S\ref{sec:evaluation}, i.e.\ the version the tie-break capture of
Example~\ref{ex:tiebreak-capture} exploits: \textsc{yes} at the
admissibility diamond goes straight to Stage~4a, so the semantic
branch never consults $\textit{verdict\_margin}$. The dashed diamond
is the margin gate added in \S\ref{sec:margin-gate}, shown at the point
where it inserts; taking its \textsc{no} edge hands the round to the
verdict fallback, whose own margin test then decides between a
verdict-commit and a typed abort.}}
\label{fig:hcsc-protocol-flow}
\end{figure*}
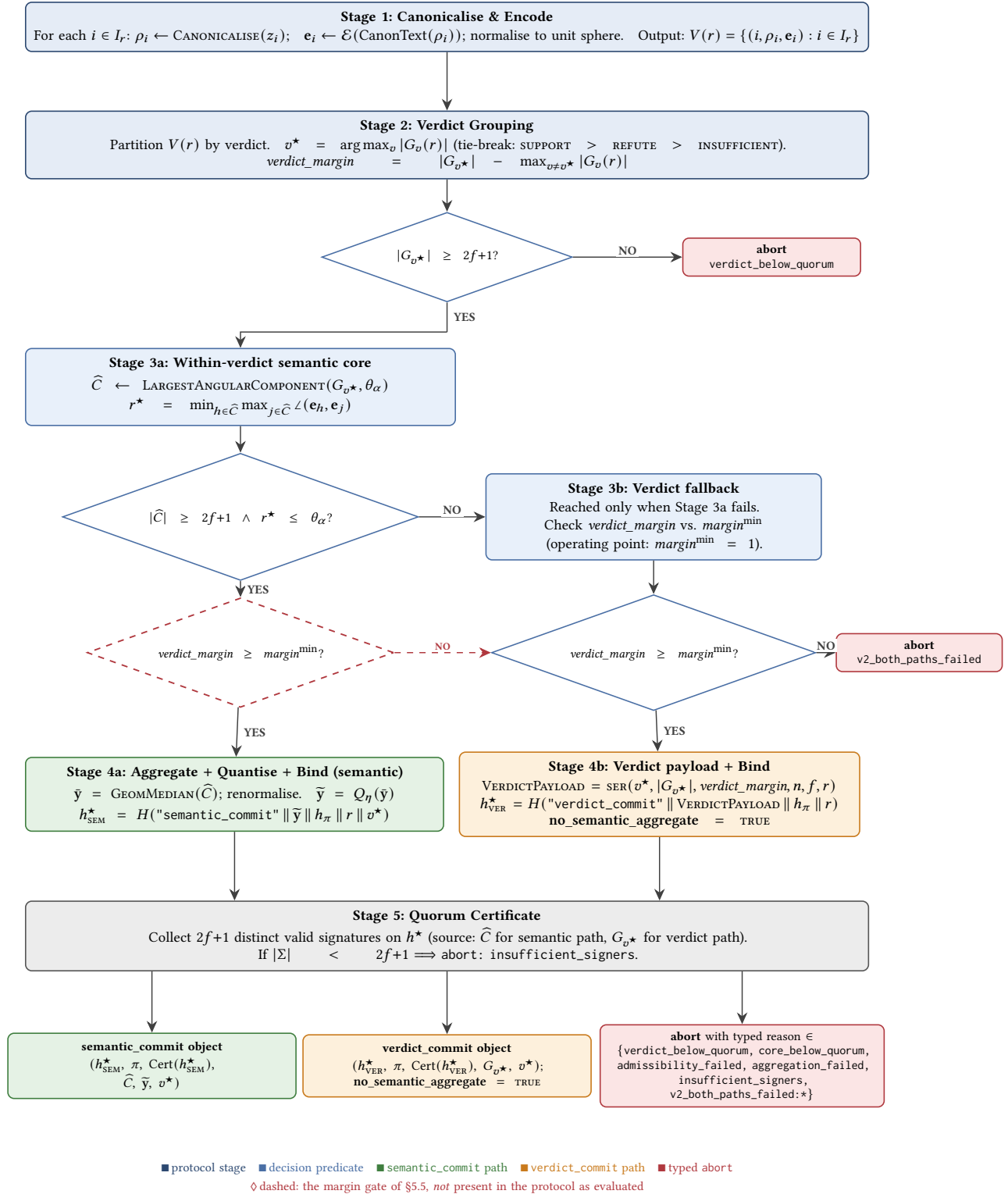

\end{document}